\documentclass{lmcs}

\usepackage{amsmath}
\usepackage{amssymb}
\usepackage{graphicx}
\usepackage{booktabs}
\usepackage{mathpartir}
\usepackage{listings}
\usepackage{xcolor}

\definecolor{rule}{HTML}{BBBBBB}
\definecolor{code}{HTML}{F6F6F4}

\definecolor{estColor}{HTML}{009E73}   
\definecolor{claimColor}{HTML}{D55E00} 
\definecolor{typeColor}{HTML}{0072B2}  

\everymath{\color{typeColor}}
\everydisplay{\color{typeColor}}

\usepackage{array}
\usepackage{colortbl}
\newcolumntype{L}{>{\normalcolor}l}
\newcolumntype{C}{>{\normalcolor}c}
\newcolumntype{R}{>{\normalcolor}r}
\newcolumntype{P}[1]{>{\normalcolor}p{#1}}
\arrayrulecolor{black}

\newcolumntype{N}{>{\color{typeColor}}r}
\newcolumntype{M}{>{\color{typeColor}}c}

\lstdefinestyle{tacetcode}{
  language=Python,
  commentstyle=\upshape, 
  morekeywords=[1]{observe,select,vs,where,rate,Stat},
  keywordstyle=[1]{\color{estColor}},
  morekeywords=[2]{claim,note,preregister,Note},
  keywordstyle=[2]{\color{claimColor}},
  morekeywords=[3]{audit,static},
  keywordstyle=[3]{\color{typeColor}},
}

\newenvironment{whole}{\par\noindent\begin{minipage}{\columnwidth}}
                      {\end{minipage}\par}

\newcommand{\sys}[1]{\textsc{#1}}

\newcommand{\Stat}[4]{\mathsf{Stat}_{#1}[#2;#3]^{#4}}
\let\oldtextsc\textsc
\renewcommand{\textsc}[1]{{\upshape\mdseries\oldtextsc{#1}}}

\newcommand{\clean}{\mathsf{clean}}
\newcommand{\rd}{\mathsf{read}}
\newcommand{\den}[1]{[\![#1]\!]}
\newcommand{\Claim}[1]{\mathsf{Claim}[#1]}
\newcommand{\kw}[1]{\ensuremath{\mathsf{#1}}}

\newcommand{\est}[1]{\textcolor{estColor}{#1}}
\newcommand{\clc}[1]{\textcolor{claimColor}{#1}}
\newcommand{\ty}[1]{\textcolor{typeColor}{#1}}

\newcommand{\Description}[1]{}

\begin{document}

\title[Tacet: Statistical Validity Accounting]{Tacet: A Language and Type System for Automatic
Statistical Validity Accounting}

\author[C.~Abuah]{Chiké Abuah\lmcsorcid{0000-0003-1860-2360}}

\address{Walla Walla University, College Place, Washington, USA}
\email{chike.abuah@wallawalla.edu}

\begin{abstract}
Empirical comparisons between systems are a standard form of evidence in
computer science research, but few are checked for statistical validity:
most are never framed as statistical tests at all. Existing
multiple-comparison procedures could control the resulting error, but need
inputs (what an analysis examined, and how its observations are arranged)
that are not recoverable from a list of $p$-values.

We introduce \sys{Tacet}, a language in which an analysis declares what it
generated, states what it expects to find, and is refused any claim it
cannot afford or cannot properly test. Its core calculus $\mathcal{T}$ pairs
a free estimation sublanguage, carrying a reported footprint and a purity
bit that records whether any outcome was consulted in building a value,
with a priced claim sublanguage, carrying a wealth transformer, connected
only by a mechanism that prices a comparison. A sample selected by reading
outcomes sets the purity bit and is recorded as having examined everything
it read, permanently, so it can never be granted a one-sided or
confirmatory price, without the system ever asking whether the analyst
intended to cherry-pick. Whether a comparison is paired or clustered is
computed statically from the artifact schema, from declared functional
dependencies between key fields alone and before any data is read, and a
mechanism that assumes that structure away is refused rather than priced.
Because the wealth transformer is antitone in the realized $p$-value,
affordability can be checked before the analysis runs too, turning
pre-registration into a typing rule. We prove the metatheory
machine-checked in Lean~4 with no admitted gaps, and demonstrate the
approach on a reference implementation and two case studies on published
artifacts, the SWE-bench Verified leaderboard and BIG-Bench Hard.
\end{abstract}


\subjclass{Theory of computation---Type structures; Mathematics of
computing---Hypothesis testing and confidence interval computation}

\maketitle

\section{Introduction}
\label{sec:intro}

Empirical comparison has become the standard method of evaluation in machine
learning and software engineering research: new systems are measured on
shared benchmarks, and the results are published as leaderboards and results
tables. However, experience has shown that the statistical validity of these
comparisons is rarely checked. Dror et al.~\cite{dror-hitchhiker} annotated
180 experimental papers from ACL and TACL and found that only 63 checked statistical
significance at all, that 42 named the test they used, and that 36 used a
correct one. Marie et al.~\cite{marie-credibility} hand-annotated 769 machine
translation papers published between 2010 and 2020 and titled the relevant
section \emph{The Disappearing Statistical Significance Testing}: the rate
never exceeded 65\% in any year, it has fallen sharply since 2016, and most
papers, in their words, drew conclusions without checking whether their
results were coincidental.

\textbf{Two facts that live in the program.} This paper claims that two of
the inputs multiple-comparison procedures need are properties of the
\emph{program}, recoverable from it alone.

First, the price a claim is eligible for depends on what selecting its
sample required looking at. Consider two clinical trials reporting the same
subgroup finding. The first reports on patients over 65, a subgroup read
off the enrollment records and fixed before any outcome was observed. The
second reports on the patients who responded to treatment, a subgroup that
became available only after every patient had been measured. The two trials
can report identical $p$-values, and a multiple-comparison procedure given
one of them has no way to tell which trial produced it. What separates them
survives in the program that computed the subgroup: selecting on a property
of the key (``the patients over 65'', ``the instances from one
repository'') records a footprint of the observations it keeps, while
selecting on a property of the outcome (``the patients who responded'',
``the instances this system failed'') records a footprint of every
observation it examined. The record is permanent: an outcome-selected sample can
never satisfy a pre-registration and is never granted a one-sided or
confirmatory price (Section~\ref{sec:lambdat}). QUDE~\cite{qude}, the
closest prior system, integrates multiple-hypothesis control into
interactive data exploration, but it assumes the hypotheses are independent
and names dependence as an open limitation (Section~\ref{sec:related}).

Second, the correct statistical instrument depends on how the observations
are arranged. Every comparison in a leaderboard is paired, because both
systems attempt the same instances, and an unpaired test is the wrong
instrument for it. This fact is legible from the artifact keys, where a
checker can see it, and absent from the results, where a reader cannot.
Section~\ref{sec:leaderboard} measures what moving this paper's running
example, the SWE-bench Verified leaderboard, from the instance level to the
declared cluster level does to individual $p$-values and to the family as a
whole. A one-sentence prose disclaimer acknowledging the clustering, the common substitute for running the
correct test, costs the author about as many words to write as the
correct test costs to run.

Our goal in this work is to answer the following four questions:

\begin{itemize}\itemsep2pt
\item Can the inputs that multiple-comparison procedures require (what
selecting an analysis's sample required looking at, and how its
observations are arranged) be recovered from the program rather than from its results?
\item Can an outcome-selected sample be caught, and refused the prices that
assume no selection occurred, without the system asking whether the analyst
intended to cherry-pick?
\item Can the experimental design of a comparison (whether it is paired or
clustered) be computed statically, and soundly?
\item Can the affordability of an analysis be checked before the analysis
runs?
\end{itemize}

We answer all four questions in the affirmative, in \sys{Tacet}: a
language in which an analysis declares what it generated, states what it
expects to find, and is refused any claim it cannot afford or cannot
properly test.

\textbf{Contributions.} In summary, this paper makes the following
contributions:

\begin{itemize}\itemsep2pt
\item We introduce $\mathcal{T}$ (Section~\ref{sec:lambdat}), a core calculus with a free estimation
sublanguage carrying a footprint and a constraint, a priced claim sublanguage
carrying a wealth transformer, and a mechanism as the only route between
them. Each sublanguage is graded in a sense with a long lineage
(Section~\ref{sec:related}), but by a grade outside that lineage's usual
numeric, commutative shape: a dependency judgment for what was read, a
non-commutative monoid of functions for what can still be afforded
(Section~\ref{sec:claims}'s ``Grades outside the numeric lineage'').
\item We show that a footprint reports what was read rather than pricing it
(Section~\ref{sec:lambdat}): a narrowing predicate that reads outcomes is
recorded, permanently, as having examined everything it read, so an
outcome-selected sample can never match a pre-registration and is never
granted a one-sided or confirmatory price, with no model of the analyst's
intentions involved. Appendix~\ref{app:estimation} proves the tracking
sound: a $\est{\clean}$ selection is a function of the artifact keys alone,
and outcomes outside a value's footprint cannot affect it.
\item We show that the experimental design can be read off the artifact
schema (Section~\ref{sec:lambdat}): whether a comparison is paired or
clustered is computed statically from declared functional dependencies
between key fields, and Theorem~\ref{thm:design}
(Section~\ref{sec:metatheory}), machine-checked in Lean~4, proves that
computation sound, given
declared dependencies and an identifying key that are genuine facts about
the run. A mechanism that assumes structure away is refused rather than
priced.
\item We show that declared bounds make affordability static (Section~\ref{sec:grade}), soundly,
because the wealth transformer is antitone in the $p$-value. This turns
pre-registration from a methodological norm into a typing rule.
\item We establish the metatheory (Section~\ref{sec:metatheory}) in Lean~4 with no admitted gaps:
the checker under-approximates the monitor
(Theorem~\ref{thm:static}), the monitor controls $\mathrm{mFDR}_\eta$ at
$\alpha$ under explicitly stated side conditions
(Theorem~\ref{thm:mfdr}), and checker acceptance implies
$\mathrm{mFDR}_\eta$ control without conditioning on the honesty of
declared bounds (Theorem~\ref{thm:bridge}).
\item We present a reference implementation (Section~\ref{sec:implementation}) requiring two declarations
per analysis (the artifact key and the cluster level) and an evaluation
(Section~\ref{sec:evaluation}) on a published leaderboard and a published benchmark corpus.
\end{itemize}

The rest of the paper is organized as follows. First we introduce \sys{Tacet}
through a worked example (Section~\ref{sec:analysis}) and summarize the
necessary statistical background (Section~\ref{sec:background}). We then
describe what an analysis must declare (Section~\ref{sec:declarations}), formalize the core
language $\mathcal{T}$ (Section~\ref{sec:lambdat}) and its grade structure
(Section~\ref{sec:grade}), and establish its metatheory
(Section~\ref{sec:metatheory}). We describe our implementation (Section~\ref{sec:implementation}) and evaluate the
system on published artifacts (Section~\ref{sec:evaluation}). Finally we state the
limitations of the guarantee and the threats to this paper's validity
(Section~\ref{sec:threats}), outline related work (Section~\ref{sec:related}), and conclude
(Section~\ref{sec:conclusion}).

\section{\sys{Tacet} by Example}
\label{sec:analysis}

\sys{Tacet} is a library, embedded in Python because that is the language
analyses are written in. ``Static'' in this
paper means before the data is read, not before the host language runs: a
runtime monitor enforces the budget as an analysis executes, and a checker
separately inspects a program's declared bounds and schema before
\lstinline[style=tacetcode]|observe| ever runs, with
Theorem~\ref{thm:static} relating what the two accept. Both surface as
ordinary Python exceptions; the refusals below are the runtime monitor's,
and Section~\ref{sec:implementation} shows the checker's own transcript,
produced before a line of this section's code would even run.

The running example throughout is the SWE-bench Verified
leaderboard~\cite{swebench}: 134 public submissions, each with the set of
benchmark instances it resolved, whose total order asserts 8{,}911 pairwise
claims that one system beats another (Section~\ref{sec:leaderboard} prices
the family in full). An analysis begins by declaring what was generated
once, and how the observations nest. Those two are the declarations;
\lstinline|passed| merely specifies how to read an outcome out of a row:

\begin{whole}\begin{lstlisting}[style=tacetcode]
study = Study(alpha=0.05)
board = study.observe(runs, artifact=key,
                      cluster="repo", passed=resolved)
top   = board["20251205_sonar-foundation-agent_claude-opus-4-5"]  # rank 1
rank2 = board["20251215_livesweagent_claude-opus-4-5"]
rank3 = board["20250928_trae_doubao_seed_code"]
rank4 = board["20251127_openhands_claude-opus-4-5"]
\end{lstlisting}\end{whole}

\lstinline[style=tacetcode]|observe| folds every observation onto its
artifact on the way in. Where a benchmark instance carries several test
assertions behind one generated patch, those assertions are not addressable
afterwards, so counting them as separate trials has no expression in the
language at all. Folding buys this protection at a cost of its own, taken
up in Section~\ref{sec:folding}: it changes the estimand, and the change is
not always in the conservative direction.

Describing the board is free:

\begin{whole}\begin{lstlisting}[style=tacetcode]
study.note(
    "the top two entries are tied at 396 resolved instances")
\end{lstlisting}\end{whole}

A \lstinline[style=tacetcode]|Note| is not a
\lstinline[style=tacetcode]|Stat|, and no claim will accept one. The boundary
between description and inference is therefore drawn by the analyst rather
than inferred, a consequence of observing programs
(Section~\ref{sec:related}).

Setting two samples against each other is also free:

\begin{whole}\begin{lstlisting}[style=tacetcode]
top.vs(rank2)
\end{lstlisting}\end{whole}

\begin{whole}\begin{lstlisting}
Comparison(
  <20251205_sonar-foundation-agent_claude-opus-4-5: 396/468> vs
  <20251215_livesweagent_claude-opus-4-5: 396/468>, two-sided)
\end{lstlisting}\end{whole}

\lstinline[style=tacetcode]|vs| returns a \emph{comparison}: the two
samples in correspondence by artifact key, two-sided because nothing
records that a direction was chosen before the data was read. A
\lstinline[style=tacetcode]|Comparison| is not yet an inference: it has no
truth value, and converting one raises, with instructions to claim it. Nor
does it name a test. The instrument is chosen from the declarations made at
\lstinline[style=tacetcode]|observe| at the moment the comparison is
claimed, and Section~\ref{sec:lambdat} computes the same choice statically,
before any data arrives.

Inference, by contrast, is priced:

\begin{whole}\begin{lstlisting}[style=tacetcode]
study.claim("rank 1 beats rank 2", top.vs(rank2))
study.claim("rank 1 beats rank 3", top.vs(rank3))
study.claim("rank 1 beats rank 4", top.vs(rank4))
\end{lstlisting}\end{whole}

The first two are charged and refused, on this real data; the third's
refusal is the one this section examines, and its message reports
exactly the two claims made before it:

\begin{whole}\begin{lstlisting}
cannot afford 'rank 1 beats rank 4': p=0.3438 exceeds
the threshold 0.003125 this pool affords, and costs
0.003125 either way.
  pool has 0.003125 left, covers 468 artifacts, charged
  on 2 earlier claims (0 kept): rank 1 beats rank 2,
  rank 1 beats rank 3
\end{lstlisting}\end{whole}

This error message is central to the design: an analyst reading it learns
the stake the claim was refused at, the remaining balance (the pool
\emph{after} this claim's charge, which coincides with that stake whenever
$i = 0.5$), how
many earlier claims in this pool were charged, and how many of those were
actually kept.

A different refusal follows when the instrument is wrong, if Fisher's exact
test is named explicitly against the same two samples:

\begin{whole}\begin{lstlisting}
'rank 1 beats rank 4': these observations are
paired+clustered, and Fisher exact assumes
independent.
  the artifact keys put the two samples in
  correspondence and span several 'repo' values.
  use a mechanism declaring paired+clustered.
\end{lstlisting}\end{whole}

We deliberately distinguish this failure from a budget failure. Here the
price is acceptable and the instrument is wrong, and no amount of wealth
makes an unpaired test correct on paired data.

\begin{figure}[t]
\centering
\includegraphics[width=\textwidth]{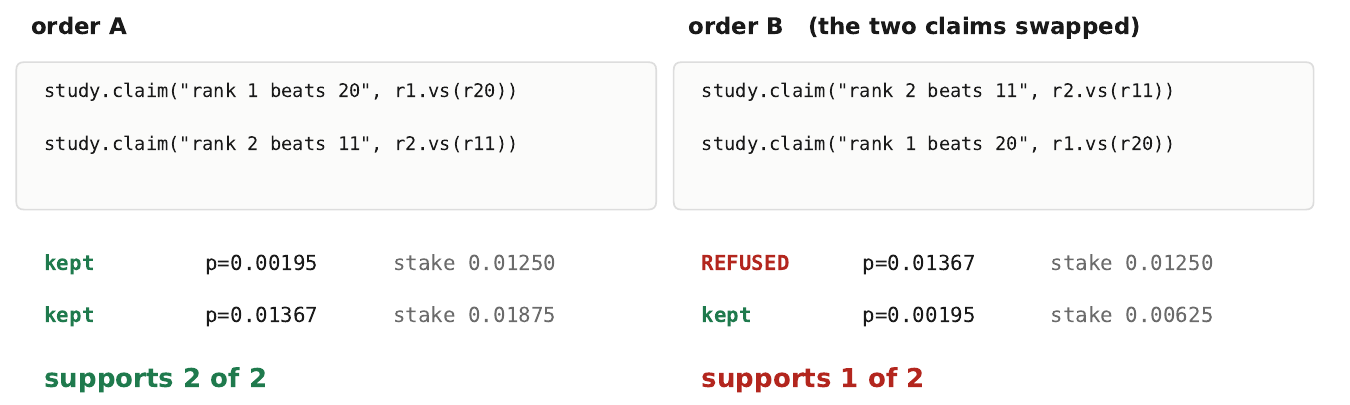}
\Description{Two Tacet programs side by side, differing only in the order of
their two claim statements, with the resulting p-values, stakes and verdicts
below each. Order A supports both claims; order B is refused on the first and
supports only the second.}
\caption{Two orderings of the same two claims, over the shared setup of
Section~\ref{sec:analysis}: order A supports both claims, and order B
supports one.}
\label{fig:order}
\end{figure}

\textbf{Order dependence.} The two programs in Figure~\ref{fig:order} differ
only in the order of their last two lines. Order A pays for the cheaper
claim, which is supported, and its refund raises the stake enough to cover
the more expensive claim behind it; both claims are supported. Order B meets
the more expensive claim first, cannot cover it at the opening stake, and
pays for that failure anyway, which halves the stake. The cheaper claim
behind it is still supported,\footnote{``Supported,'' ``kept,'' and
``accepted'' are this paper's synonyms for one runtime fact: the claim was
made and the pool's stake covered its realized $p$-value. ``Certified'' and
``decided'' belong to the \emph{static} verdict instead, computed from a
declared bound with no data present, which this paper calls
\emph{affordable}; ``solvency'' is a third, distinct fact about whether the
pool could still pay a stake if a claim lost. Appendix~\ref{app:semantics}
gives all three their precise definitions in one place.} so the analysis has
bought one finding where it could have had two. Note that the evidence never
changed. Alpha-investing refunds wealth on a rejection and charges for a
failure, so a supported claim funds the claims behind it, and a claim that
misses starves them.

A \emph{predictable} ordering, throughout this paper, is one fixed
independently of the claims' own $p$-values or effect sizes: reading order
and submission date are the two the case studies use. At the scale of the
whole leaderboard, the gap Figure~\ref{fig:order} illustrates is the
difference between the 0 supported claims either predictable ordering produces and
the up to several hundred a random draw can reach
(Section~\ref{sec:leaderboard}).

\section{Background}
\label{sec:background}

Section~\ref{sec:analysis}'s ``budget'' and ``refund'' are
this section's \emph{wealth} $W$ and $\omega$; its ``stake'' (raised,
halved, opened at a level) is the \emph{amount} actually staked,
$\alpha_j = i \cdot W_{j-1}$ in Section~\ref{sec:mfdr}'s notation, not the
constant \emph{invest fraction} $i$ itself, which Section~\ref{sec:analysis}
never names directly.

\textbf{Families of tests.} A $p$-value is valid for a null hypothesis
$H_0$ when $\Pr_{H_0}(p \le t) \le t$ for every $t$, so a single test at
$p < \alpha$ is wrong \emph{at most} a fraction $\alpha$ of the time when
there is nothing to find; the exact tests this paper uses on discrete
data (Fisher's~\cite{fisher1935}, McNemar's~\cite{mcnemar}, and the cluster
sign-flip~\cite{winkler-permutation}) are super-uniform under the null,
conservative at the thresholds discreteness makes unachievable, which is
why later sections compare each stake against the sharpest $p$-value the
data's own discreteness can produce: a stake below that floor can be met
by no comparison, at any effect size. Running $k$ such tests where nothing is real gives at
least one false rejection with probability at most $1-(1-\alpha)^k$
\emph{if the tests are independent} (at
$\alpha = 0.05$ and $k = 8{,}911$, as on the leaderboard of
Section~\ref{sec:leaderboard}, that already exceeds $0.999999$), and, if they are
not, anywhere between $\alpha$ and $\min(1, k\alpha)$: which of these
applies is not a fact about the $p$-values, and
Section~\ref{sec:leaderboard}'s 8{,}911 comparisons are not independent. The claims a paragraph reads off one
table form exactly this kind of family, whether or not anyone frames them
as tests.

\textbf{Classical procedures.} The classical corrections are offline: each
consumes all $k$ $p$-values as a completed batch rather than one at a time
as they arrive. Bonferroni~\cite{dunn1961} divides $\alpha$ by $k$ and controls the
probability of \emph{any} false rejection, the family-wise error rate.
Benjamini and Hochberg~\cite{benjamini-hochberg} instead bound the
\emph{expected fraction} of rejections that are false, the false discovery
rate, which is less conservative because it charges in proportion to how
many discoveries there are rather than to how many tests were run, but
only under independence or positive dependence among the tests; arbitrary
dependence needs Benjamini and Yekutieli's
extension~\cite{benjamini-yekutieli}, and a family drawn from correlated
observations, like the leaderboard of Section~\ref{sec:leaderboard}, is exactly
the setting the plain procedure's guarantee does not cover. Both
procedures need the full batch of $k$ tests in hand before either can
compute a single rejection, and both are silent on an analyst who runs a
test, likes the result, and decides to run one more: that analyst has
already left the offline setting either procedure assumes.

\textbf{Sequential procedures.} Sequential procedures are online: they spend
a budget instead of dividing one, test by test, as each arrives. Foster and
Stine's alpha-investing~\cite{foster-stine} treats $\alpha$ as \emph{wealth}
$W$, carried from one test to the next
rather than split among a number fixed in advance: each test stakes a
fraction of the current wealth against the $p$-value it returns, earning a
reward $\omega$ back on a rejection and, in Foster and Stine's own rule,
paying a deduction strictly larger than the stake itself on a non-rejection.
The rule this paper actually uses, defined precisely as
Equation~\ref{eq:transformer} in Section~\ref{sec:grade}, deducts the stake
$i \cdot W$ directly instead, whether or not the test rejects, a different,
simpler choice in the same family, sound by Aharoni and Rosset's generalized
alpha-investing~\cite{aharoni-rosset} rather than by matching Foster and
Stine's rule itself. Either way, wealth is a stochastic process rather than
a constant, and the guarantee is on
\[
  \mathrm{mFDR}_\eta = \frac{\mathbb{E}[V]}{\mathbb{E}[R]+\eta}, \qquad \eta > 0,
\]
the ratio of expected false rejections $V$ to expected rejections $R$, with
a constant added to the denominator so that the ratio is defined before
anything has been rejected; \sys{Tacet}'s own choice of $\eta$ is fixed later,
in Theorem~\ref{thm:mfdr}. Because the rule only ever consults wealth
already accumulated, $k$ never has to be known. This property is exactly
what the leaderboard requires: nobody declares in advance how many pairwise
comparisons a table will invite. LORD~\cite{lord} and SAFFRON~\cite{saffron}
are later investment rules in the same family, aimed at FDR itself rather
than the weaker $\mathrm{mFDR}_\eta$ this paper controls; the wealth
transformer of Section~\ref{sec:grade} is one more, and
Section~\ref{sec:limitations} states the guarantee's limits. The
wealth-as-betting-process framing is the testing-by-betting
view of statistical evidence, in which a $p$-value-controlling procedure is
recast as a gambler's wealth against a null~\cite{shafer-betting}, later
given a general safe-testing account via e-values that compose correctly
under optional continuation~\cite{safe-testing}. Alpha-investing fits that
same betting picture, staking wealth against a false-discovery budget
rather than a single hypothesis's null; the wealth this section defines is
that stake.

\textbf{Experimental design.} The shape of the data decides which instrument
a comparison licenses. Two samples measured on the same units,
such as two systems attempted on the same benchmark instances, are
\emph{paired}: McNemar's test~\cite{mcnemar} uses the correlation between
them, while an independent-samples test aimed at the same data discards that
correlation and answers a different question at the number it returns.
Observations sharing a cluster, such as instances drawn from the same
repository, are not independent trials either, and treating them as such is
a separate error, stackable with the first. Importantly, Bonferroni and
alpha-investing alike take the $p$-value as given; neither distinguishes a
correct instrument from an incorrect one computed to the same number of
decimal places. Each of the three exact tests this paper serves has its
own null, precisely: Fisher's is the hypergeometric null that group
membership is independent of outcome, so the two groups' successes are
however many a random split of the pooled outcomes would place in each;
McNemar's is marginal homogeneity, so each discordant pair is equally
likely to favor either side, independent of the others; and the cluster
sign-flip's is that the sign of each cluster's own net difference is
exchangeable, so flipping any subset of clusters' signs leaves the null
distribution unchanged. Section~\ref{sec:lambdat}'s design premise decides
which of these three a comparison's shape licenses; none of the three is
a claim about what happens \emph{within} a cluster beyond that exchange
symmetry, and Section~\ref{sec:design-premise} states exactly how much
dependence each is willing to tolerate there.

\textbf{Pre-registration.} Pre-registration makes cherry-picking visible
from outside the analysis. Gelman and Loken's garden of forking
paths~\cite{gelman-loken} names the risk: an analyst who does not fix in
advance which comparisons matter can, in good faith, follow a data-dependent
path that looks like a single test but carries the multiplicity of every
path not taken. Adaptive data analysis~\cite{adaptive-data-analysis} treats
this directly as a composition problem. Declaring a bound before looking is
the practical answer already in use; Section~\ref{sec:grade} shows how to
turn that declaration into one a checker can decide before the analysis
runs.

\textbf{Selective inference.} Cherry-picking is not always a deliberate
choice: an analyst who looks at several outcomes and reports the one that
came out favorably has conditioned inference on that selection, whether or
not the choice was intentional. The statistics literature already prices
this event: post-selection inference conditions the sampling distribution on
the selection having occurred~\cite{berk-posi,taylor-tibshirani-si}, so a
$p$-value computed as though no selection took place is invalid regardless
of intent. Section~\ref{sec:lambdat}'s \est{\textsc{T-SelectOutcome}}
recovers this conditioning event from the program that did the selecting,
rather than from a model of what the analyst meant to do. The type system
uses the recovered event to withhold the prices that assume no selection
occurred, the one-sided and confirmatory discounts; the $p$-value itself is
never recomputed, and Section~\ref{sec:threats} states what that leaves
open.

Each procedure above needs facts about the analysis as input, and two of
them (what was read, and how the observations are arranged) are properties
of the \emph{program} that produced the numbers, invisible in the numbers
alone.

\section{Declarations in \sys{Tacet}}
\label{sec:declarations}

Three inputs to \sys{Tacet} are declared rather than inferred. All three are
modeling judgments that no analysis of the data can recover, and the
checker's job is to hold the analyst to them consistently, with no attempt
to guess them. The guarantee that follows is conditional on that consistency
being honest: nothing in the type system checks that a declared unit,
schema, or bound is true of the run, only that the analysis proceeds as
though it were (Section~\ref{sec:metatheory}).

\textbf{The artifact.} The artifact is the output of one act of generation:
the \emph{unit of analysis}, counted a single time however many outcomes
derive from it. For example, when a
model produces a single patch that touches seven functions, checked by nine
differential tests, whether the unit is the patch or the function depends on
whether the model can misread one function's specification without breaking
the others. That is a claim about the failure mode, something the data
alone cannot decide. The program supplies the structure, and the analyst
supplies the judgment.

This same artifact declaration is easy to misread. It is tempting to read it
as repairing a mis-measured number: the naive count was too
large, the artifact count corrects it, and the evidence weakens by whatever
factor separates the two. That reading is wrong: folding substitutes a
different estimand for the evidence rather than attenuating it, and
Section~\ref{sec:folding} gives a counterexample in which two
observation-by-observation indistinguishable groups separate after
folding. This is why
the artifact cannot be inferred even in principle: choosing the unit
determines which question is being asked.

We should be careful about how much this declaration buys, because we
measured it and the answer is smaller than the motivation suggests. Across
17 analyses in two corpora, 8 analyses of this project's own and 9 public
benchmarks fetched fresh (HumanEval, MBPP, and SWE-bench Lite among
them), the
ratio between a naive count of outcomes and a count of generated artifacts is
bimodal: exactly 1.00 for 11 of them, where an analysis asks one question
per generation, and between 2.53 and 9.00 for the other 6, where several
checks apply to one artifact. Nothing falls in between, and 35\% of the
analyses are affected. However, no analysis in either corpus is answering a
question it did not mean to ask, because the mature benchmarks in our
corpus write their intended unit into their own harness:
\lstinline|pass@k|~\cite{humaneval,mbpp} is computed at the level of the problem, with assertions
playing no role in the count. \lstinline|pass@k| exists
because somebody reasoned about the unit of analysis once and wrote the
answer into a harness; that reasoning is per-benchmark, invisible in the
results, and unavailable to anyone writing a one-off analysis. A harness is,
in effect, an informal type system for a single benchmark. Our approach
generalizes it to any analysis.

\textbf{The schema.} The schema declares which key fields identify
observations, which fields determine which, and which field is the level at
which independence is assumed. Declaring that an instance determines its
repository lets the checker decide statically whether a
comparison is paired or clustered (Section~\ref{sec:lambdat}), and it is the
declaration whose absence costs a factor of 48 on the leaderboard
(Section~\ref{sec:leaderboard}).

\textbf{The bound.} The bound declares what each claim expects to find. This
declaration is optional, and an analysis that omits it still runs under a
runtime monitor. However, an analysis that supplies it can be checked before
it is run, because the wealth transformer is antitone in the $p$-value and
monotone in the wealth already accumulated, and a declared upper bound
therefore under-approximates the wealth an honest execution will have
(Section~\ref{sec:grade}). Pre-registration, in this
setting, is precisely the information the checker needs to answer the
affordability question early.

Everything else is inferred. Footprints propagate through ordinary
arithmetic with nothing declared, and the design follows from the schema and
the selections, as the worked analysis of Section~\ref{sec:implementation}
shows.

\section{The Core Language \texorpdfstring{$\mathcal{T}$}{T}}
\label{sec:lambdat}

$\mathcal{T}$ consists of two mutually embedded sublanguages, following the
design of \sys{Duet}~\cite{duet}. Estimation is free and tracks what data a value
rests on. Claim is priced and tracks what the budget does. The only route
from one sublanguage to the other is a mechanism, and the grammar below is
small enough to state this directly: only \kw{claim} crosses from $e$ to
$c$, and nothing crosses back.

Every rule and equation from here to the appendices is colored in \sys{Duet}'s own
convention, stated once: \est{estimation}, $\est{\mathsf{Stat}}$, is
\est{bluish green}; \clc{claim}, $\clc{\mathsf{Claim}}$, is
\clc{vermillion};\footnote{The three colors are Wong's colorblind-safe
palette~\cite{wong2011}, checked under a color-vision-deficiency
simulation~\cite{machado2009} to remain pairwise distinguishable. Color
is a reading aid rather than a carrier of information no other mark
makes: which sublanguage a term belongs to is also decidable from the
notation alone, so a grayscale print or a colorblind reader loses
convenience, not content.} and the numeric and type-level
machinery the two sublanguages share (the design lattice, and $T_p$ and
$T_\bot$ as functions rather than as a claim's payload) is \ty{blue}. A
rule is colored line by line instead of as a whole, so a rule whose premise
is estimation and whose conclusion is a claim shows both: \clc{\textsc{T-Claim}}
(Section~\ref{sec:claims}) has an \est{estimation} premise, a \ty{blue} side-condition,
and a \clc{claim} conclusion, because it is the rule where all three meet.

{\setlength{\arraycolsep}{2pt}
\[
\begin{array}{llll}
\est{\text{sample}} & \est{v} & \est{::=} &
  \est{\kw{observe}_\Sigma(r, a) \mid v[k] \mid v\langle q \rangle} \\
\est{\text{estimation}} & \est{e} & \est{::=} &
  \est{v \mid e_1 \oplus e_2} \\
\clc{\text{claim}} & \clc{c} & \clc{::=} &
  \clc{\kw{note}(s) \mid \kw{claim}^m_{\bar p}(\ell, v_1 \mathbin{\kw{vs}} v_2) \mid c_1 \,;\, c_2}
\end{array}
\]}

\noindent Selection ($[k]$ and $\langle q \rangle$) applies only to a
\est{sample} $v$, never to an arithmetic combination: $e_1 \oplus e_2$
already denotes a single reduced value (Section~\ref{sec:estimation}
fixes what that means), with no per-observation structure left for a key
literal to filter or a predicate to examine. A sample embeds into
estimation ($v$ is always a well-formed $e$), so every rule below that
concludes a $\est{\Stat{\Sigma}{\varphi}{\kappa}{t}}$ type for some $e$
covers a sample as the case where $e = v$; only \est{\textsc{T-SelectKey}}
and \est{\textsc{T-SelectOutcome}} are stated over the narrower sort,
because only they need to be. \kw{vs} is likewise restricted to $v_1,
v_2$, because a mechanism consumes per-observation structure and an
aggregate has none left for $\mathit{price}$ to align by key; this keeps
every claim's two sides samples a mechanism can actually run on, matching
the shipped \lstinline[style=tacetcode]|Comparison|.

\noindent A claim term $c$ is, up to associativity of $;$ (which we
identify throughout), a single sequence $c_1 \,;\, \dots \,;\, c_n$ of
\kw{note} and \kw{claim} leaves; the rules below assign it a type only if
every leaf does. A \emph{program} additionally pairs a claim term with a
plan $\pi$ (Appendix~\ref{app:syntax}), and this section's judgment extends
to $\pi \vdash c : \tau$ there. Appendix~\ref{app:syntax} restates the
grammar with the metavariables it ranges over, and Appendix~\ref{app:semantics}
steps through this normal form to prove Theorem~\ref{thm:static}.

\subsection{Schemas}
\label{sec:schemas}

An artifact key is structured, and its structure determines whether a
comparison is paired or clustered. A schema is declared once, at
\kw{observe}.

\begin{defi}[Schema]
\label{def:schema}
$\Sigma = (F, I, D, \gamma)$ where $F$ is the key's fields, $I \subseteq F$ those that
\emph{identify} observations within a selection, $D \subseteq F \times F$ a set of
functional dependencies written child $\mapsto$ parent, and $\gamma \in F \cup
\{\cdot\}$ the level at which independence is assumed.
\end{defi}

\noindent $D$ encodes the nesting: a benchmark instance determines its
repository, and a task determines its mode. Declaring $D$ makes the design
check static, with no runtime inspection of values. $I$ is the alignment
key, and its scope is deliberately relative: on the leaderboard, two
submissions' rows share the instance field, so $I = \{\mathit{instance}\}$
keys no run containing both. $I$ is asked to key the observations any one
selection keeps, all of which agree on the fields that selection pinned;
Definition~\ref{def:conforms} (Section~\ref{sec:metatheory}) states the
relativization precisely, and Theorem~\ref{thm:design} is proved against it.
This full schema is a separate, optional declaration from the artifact key
and cluster level \kw{observe} takes (Section~\ref{sec:declarations}): where
it is supplied
(\lstinline[style=tacetcode]|Schema(fields=..., identifies=..., determines=..., cluster=...)|
passed to \kw{preregister}),
the checker computes $\mathit{design}$ from $\Sigma, \kappa_1, \kappa_2$
alone, before any data is read; where it is not, the same judgment is read
directly off the observed key values. Theorem~\ref{thm:design} proves the
static route sound, and Section~\ref{sec:metatheory} checks the two
routes' agreement empirically.

A selection carries a \emph{constraint} $\kappa : F \rightharpoonup \mathit{Val}$,
the fields it pinned. A key literal $k$ is a pair $(f, v) \in F \times
\mathit{Val}$; an observation $u$ \emph{matches} $k$ when $f(u) = v$; and
$\kappa \uplus (f, v)$ extends $\kappa$ with $f \mapsto v$, defined only
when $f \notin \mathrm{dom}(\kappa)$, so pinning a field twice (the same
value or a conflicting one) has no derivation rather than a silent
overwrite. Dually, $\kappa_1 \sqcap \kappa_2$ is the partial function
defined exactly where $\kappa_1$ and $\kappa_2$ are both defined and agree,
taking that shared value, so a derived value keeps only the fields both
operands pinned the same way. Estimation types are
$\est{\Stat{\Sigma}{\varphi}{\kappa}{t}}$, indexed by the schema and
carrying a footprint, a constraint, and a \emph{purity bit}
$\est{t \in \{\clean, \rd\}}$; claim types are $\clc{\Claim{w}}$, carrying a
wealth transformer.

\textbf{Why the type is indexed by $\Sigma$.} The schema is declared at
\kw{observe} and nowhere else, so without the index it appears in the design
premise of \clc{\textsc{T-Claim}} as a free metavariable, buried inside
$e_1$ and $e_2$, erased by the type, and therefore satisfiable by choosing
whichever schema makes the premise hold. Indexing the type binds it: the
premise's $\Sigma$ is the one \est{\textsc{T-Observe}} declared, propagated
by every estimation rule. The index also closes a cross-schema hole in
\est{\textsc{T-Arith}}, where sharing the metavariable is now what requires
both operands to have been observed under the same schema; combining two
estimates whose keys have different fields has no meaning that
$\kappa_1 \sqcap \kappa_2$ could summarize.

\textbf{The purity bit.} $\est{t}$ records whether any outcome was consulted
in building the value, and it is the one fact $\kappa$ cannot record.
$\kappa$ says which key literals a selection pinned; it says nothing about
\emph{how} the surviving observations were chosen, so an outcome-dependent
narrowing leaves $\kappa$ intact while changing which observations the value
rests on. The two bits order as $\est{\clean \sqsubset \rd}$, with join
$\est{\sqcup}$; $\est{\rd}$ is absorbing, because once an outcome has been
read no later operation un-reads it. Pre-registration is a promise about a
comparison fixed before the data was seen, so the two rules that cash a
registration in, \clc{\textsc{T-OneSided}} and
\clc{\textsc{T-Registered}} (Appendix~\ref{app:plan}), demand
$\est{\clean}$ on both sides.

\textbf{The role of the footprint.} $\varphi$ is diagnostic: no typing rule
below inspects it in a conclusion or a side condition, so deleting it from
$\est{\Stat{\Sigma}{\varphi}{\kappa}{t}}$ changes which programs typecheck
not at all, and Section~\ref{sec:implementation}
makes the same concession about the shipped audit. The weight a first
reading might assign to $\varphi$ is carried by $\kappa$: the contrast
between \est{\textsc{T-SelectKey}} and \est{\textsc{T-SelectOutcome}} below
is, for pricing purposes, a contrast in what each does to $\kappa$, not to
$\varphi$, since it is $\kappa$ alone that reaches \clc{\textsc{T-Claim}}'s
design premise. The footprint stays in the type because its consumer is
outside the rules: refusal messages and the audit record print it, and a
report of what was read becomes load-bearing in (H1)'s external, timestamped
record (Section~\ref{sec:mfdr}); every rule below is deliberately
indifferent to it. The report those consumers receive is also faithful, by
Lemma~\ref{lem:footprint} (Appendix~\ref{app:estimation}): the sample a
value rests on lies inside its footprint, and outcomes outside the footprint
can affect neither the sample nor any value computed from it.

\subsection{Estimation}
\label{sec:estimation}

\begin{mathpar}
\inferrule[T-Observe]{ }{\est{\vdash \kw{observe}_\Sigma(r, a) :
  \Stat{\Sigma}{\mathit{units}(r,a)}{\emptyset}{\clean}}}

\inferrule[T-SelectKey]{\est{\vdash v : \Stat{\Sigma}{\varphi}{\kappa}{t}} \\
  \est{\psi = \{u \in \varphi \mid u \text{ matches } k\}}}
  {\est{\vdash v[k] : \Stat{\Sigma}{\psi}{\kappa \uplus k}{t}}}

\inferrule[T-SelectOutcome]{\est{\vdash v : \Stat{\Sigma}{\varphi}{\kappa}{t}}}
  {\est{\vdash v\langle q \rangle : \Stat{\Sigma}{\varphi}{\kappa}{\rd}}}

\inferrule[T-Arith]{\est{\vdash e_i : \Stat{\Sigma}{\varphi_i}{\kappa_i}{t_i}} \\
  \est{\oplus \in \mathit{Ops}}}
  {\est{\vdash e_1 \oplus e_2 :
   \Stat{\Sigma}{\varphi_1 \cup \varphi_2}{\kappa_1 \sqcap \kappa_2}{t_1 \sqcup t_2}}}
\end{mathpar}

In \est{\textsc{T-Observe}}, $\est{\mathit{units}(r,a)}$ is the set of
artifact keys $a$ assigns to the rows of $r$: a fresh observation rests on
every artifact the run generated and pins nothing.

The key rule is \est{\textsc{T-SelectOutcome}}. Its conclusion keeps the
\emph{input's} footprint, not the subset the predicate retained, so a report
of what was read never shrinks just because looking was followed by
discarding. It actually changes the purity bit instead: it sets $\est{t}$ to
$\est{\rd}$, and that is the effect $\kappa$ and $\varphi$ together cannot
record on their own. Contrast \est{\textsc{T-SelectKey}}, which narrows the
footprint precisely because the predicate never consulted an outcome, and
which therefore passes $\est{t}$ through unchanged. Which rule applies is
syntactic, decided by the selection form itself rather than by inspecting a
predicate's body: $v[k]$ is always \est{\textsc{T-SelectKey}}, a literal
matched against the key, and $v\langle q \rangle$ is always
\est{\textsc{T-SelectOutcome}}, a predicate given the outcome. A predicate
that in fact reads neither key nor outcome still enters through $v\langle q
\rangle$ and is still priced as $\est{\rd}$: conservative, since nothing
downstream needs to know a predicate ignored its argument,
only whether it could have used it. Under information-flow
terminology~\cite{sabelfeld-myers}, the purity bit is exactly the question
``did this value depend on an outcome?'', and \est{\textsc{T-SelectOutcome}}
is a dependency judgment in the sense of the Dependency Core
Calculus~\cite{dcc}: it is deciding whether a selection is in the image of
an outcome-typed value, not bounding a distance in a metric space.
Appendix~\ref{app:estimation} proves the corresponding noninterference
statement (Lemma~\ref{lem:clean}): a $\est{\clean}$ selection is a function
of the artifact keys alone, so the two rules that demand $\est{\clean}$ on
both sides (Appendix~\ref{app:plan}) cash a semantic guarantee that the
comparison was fixed before any outcome existed, beyond the naming
convention the bit would otherwise be. The two rules differ the same way on
$\kappa$, and that difference is the one with pricing consequences:
\est{\textsc{T-SelectKey}} pins $k$ into $\kappa$, while
\est{\textsc{T-SelectOutcome}} leaves $\kappa$ exactly as it found it, so no
amount of outcome-reading can manufacture a pairing or clustering structure
for \clc{\textsc{T-Claim}}'s design premise to find below that the key
literals alone did not establish.

In \est{\textsc{T-Arith}}, $\kappa_1 \sqcap \kappa_2$ keeps only the fields
both sides pinned to the same value, which is the most a derived value can
soundly claim about its own provenance, and $t_1 \sqcup t_2$ keeps
$\est{\rd}$ if either operand carries it.

\textbf{What $\est{\oplus}$ ranges over.} $\est{\oplus}$ is not an arbitrary
binary operation on estimates. If it were, an outcome-sensitive combinator
would do the work of \est{\textsc{T-SelectOutcome}} without triggering it,
and \est{\textsc{T-Arith}} would return a type indistinguishable from a
clean one. $\mathit{Ops}$ is therefore a fixed signature: each
$\est{\oplus \in \mathit{Ops}}$ has type
\[
  \est{\oplus \;:\; \mathit{Val} \times \mathit{Val} \rightharpoonup
   \mathit{Val}}
\]
and combines the operands' own values directly: each estimate already
denotes a single aggregate $\mathit{Val}$ (a rate, say, reduced over its
own footprint), with no
notion of a value indexed by unit for $\est{\oplus}$ to range over, and the
partiality in $\rightharpoonup$ is the ordinary kind field arithmetic
already has ($\div$ undefined at a zero divisor). The footprint union
$\varphi_1 \cup \varphi_2$ in \est{\textsc{T-Arith}}'s conclusion is
fixed by the operands' footprints alone, so $\est{\oplus}$ cannot drop,
keep, or reweight a unit as a function of that unit's outcome.
Concretely, $\mathit{Ops}$ is the ordinary field arithmetic ($+$, $-$,
$\times$, $\div$, and scalar versions of the same). Anything outside
that signature is a selection and must be written as one, where
\est{\textsc{T-SelectOutcome}} can see it.

\subsection{Claims}
\label{sec:claims}

\begin{mathpar}
\inferrule[T-Note]{ }{\clc{\vdash \kw{note}(s) : \Claim{\mathit{id}}}}

\inferrule[T-Claim]{\est{\vdash v_i : \Stat{\Sigma}{\varphi_i}{\kappa_i}{t_i}} \\
  \ty{\mathit{design}(\Sigma, \kappa_1, \kappa_2) = \langle \mathit{pr},
   \mathit{cl} \rangle} \\
  \ty{\mathit{assumes}(m) = \langle \mathit{pr}, \mathit{cl}' \rangle} \\
  \ty{\mathit{cl} \sqsubseteq \mathit{cl}'} \\
  \ty{\mathit{sided}(m) = 2}}
  {\clc{\vdash \kw{claim}^m_{\bar p}(\ell, v_1 \mathbin{\kw{vs}} v_2) :
   \Claim{T_{\bar p}}}}

\inferrule[T-Seq]{\clc{\vdash c_i : \Claim{w_i}}}
  {\clc{\vdash c_1 ; c_2 : \Claim{w_2 \circ w_1}}}
\end{mathpar}

\kw{note} yields the identity transformer. It is free, and because its type
is not a $\est{\mathsf{Stat}}$, no claim will accept it. Transformers
compose in program order, which is where non-commutativity enters and why
$\circ$ here is not symmetric. $\mathit{sided}(m)$, also a premise above, is
a fixed fact about the mechanism $m$ itself (one-sided or two-sided) that
Appendix~\ref{app:plan}'s \clc{\textsc{T-OneSided}} and
\clc{\textsc{T-Registered}} also consult; \clc{\textsc{T-Claim}} requires
$\mathit{sided}(m) = 2$ because it is the rule for two-sided mechanisms,
and a one-sided one can only be typed through \clc{\textsc{T-OneSided}}
instead.

\textbf{The subscript at $\bar p = \bot$.} \clc{\textsc{T-Claim}}'s
conclusion $\Claim{T_{\bar p}}$ is written for every $\bar p \in [0,1] \cup
\{\bot\}$: at a declared $p$ it is Equation~\ref{eq:transformer}'s $T_p$,
and at $\bar p = \bot$ (no bound declared) it is
Equation~\ref{eq:static-transformer}'s $T_\bot$, the no-refund
transformer (Section~\ref{sec:grade}). No rule needs to disambiguate
which equation a given $T_{\bar p}$ came from, because the index already
says.

\textbf{No context.} The judgments above have no typing context, because
$\mathcal{T}$ has no variables and no binders: $e$ is built from \kw{observe}
by selection and arithmetic, $c$ is a sequence of leaves, and nothing in
either grammar can be named. Carrying a context no rule reads would add
nothing, and would invite a comparison with \sys{Duet}'s
substructural context splitting that $\mathcal{T}$ cannot support, since
there is nothing here a program could name twice, and correspondingly no
substitution theorem to prove.

\textbf{$\clc{\mathsf{Claim}}$ as a graded sequential composition.} What
remains, and what \clc{\textsc{T-Seq}} states, is that
$\clc{\mathsf{Claim}}$ is graded by sequential composition in a monoid of
wealth transformers. \sys{Duet}~\cite{duet} splits into two languages: a
sensitivity language, following \sys{Fuzz}~\cite{fuzz}, whose grade is a real
number and is allowed to scale; and a privacy language that composes
differentially private computations by threading one privacy-typed term
through \textsc{Bind}, whose grade is also a number but is not allowed to
scale, because $(\epsilon,\delta)$-privacy cost does not form the metric
that scaling under \sys{Fuzz}'s linear types depends on.
$\clc{\mathsf{Claim}}$ plays the role of \sys{Duet}'s privacy language.
\clc{\textsc{T-Seq}} sequences claims as \textsc{Bind} sequences
privacy-typed terms, with no binding occurrence to carry a value forward,
since nothing a claim computes is consumed by the next one; only the wealth
it leaves behind is. The shape of the grade differs. \sys{Duet}'s
privacy cost is a number, and composing two of them is addition, which
commutes. Alpha-investing's refund is conditional on whether a claim was
rejected, so $T_p$ must be a function of $W$ rather than a number, and
composing two of them is $\circ$, which does not commute; this is the source
of Section~\ref{sec:grade}'s non-commutativity, a stateful, data-dependent
refund forcing the grade to stop being a scalar. A graded system whose
grades commute may reorder its effects freely; $\mathcal{T}$'s may not, at
the cost Figure~\ref{fig:order} shows.

\textbf{Grades outside the numeric lineage.} Both sublanguages are graded in a
sense already familiar from this lineage, and neither grade is the
familiar one. \sys{Fuzz}, \sys{DFuzz}, \sys{Duet}, \sys{Solo}, and
\sys{Jazz} all grade by a \emph{numeric} resource, sensitivity or privacy
cost, composing under addition or a scaling action, commutative in every
case, \sys{Solo}'s own graded treatment of privacy cost~\cite{solo}
included, since the grade it threads is still a number under an
associative, commutative combination, only inferred relative to declared
units rather than to variables. Estimation's grade, $\varphi$ paired with
the purity bit $\est{t}$, is not a cost at all: it is a set of declared
units under union together with a two-point lattice recording whether an
outcome was consulted, a dependency judgment in the sense of the
Dependency Core Calculus~\cite{dcc} (Section~\ref{sec:estimation}) rather
than a sensitivity bound. Claim's grade, $T_p$, is numeric but composes
under function composition rather than addition, and $\circ$ does not
commute where every numeric grade in this lineage's $+$ or scaling action
always would. Neither choice is a variation on the other's theme, and the
paper's technical content is choosing both correctly and connecting them
through the one mechanism able to cross from one sublanguage to the
other: \kw{claim} reads a $\est{\clean}$ purity bit as a precondition and
produces a $\clc{T_{\bar p}}$ transformer as its only output
(\clc{\textsc{T-Claim}}, \clc{\textsc{T-OneSided}}), so a dependency fact
about what was read becomes an input to a resource fact about what can
still be afforded, and nothing else in the calculus can make that
crossing. Once the two grades are fixed this way, the short inductions of
Section~\ref{sec:metatheory} are short \emph{because} the grades were
chosen to make them so, exactly as in the coeffect and graded-(co)monad
literature this lineage belongs to~\cite{coeffects,gaboardi-grading}: the
substantial work in a graded soundness result is finding the semiring or
resource monoid that fits the domain, and the induction that verifies it,
once found, is routine by design rather than in spite of the design.

\subsection{The Design Premise}
\label{sec:design-premise}

$\mathit{assumes}(m)$ is mechanism $m$'s declaration about its inputs, drawn
from the design lattice: the product of the two facts computed independently
below, namely whether the observations are in correspondence across the two
selections (paired) and whether they nest inside larger units (clustered).
That is four points, and they form the evident product lattice, with
independent at the bottom and paired+clustered at the top; the rest of
this paper names the points by that order. The premise, however, checks
something else: neither the three-point chain the three served instruments
suggest nor the pointwise $\sqsubseteq$ of the product itself gives the
right premise, because the two coordinates are not the same kind of fact.

Clustering \emph{is} slack: a mechanism that models dependence the data
lacks loses only power and remains valid, so the premise need only check
that nothing is assumed away, and $\sqsubseteq$ on that coordinate is the two-point chain
$\text{no} \sqsubseteq \text{yes}$. Pairing is not slack, for the reason
stated under ``the fourth lattice point'' below: it is the correspondence
the statistic is computed \emph{over}, and a mechanism assuming it where it
is absent breaks the alignment rather than weakening it. On that coordinate
the premise is equality. So \clc{\textsc{T-Claim}}'s premise is two separate
2-chains, one ordered and one discrete, which writing $\mathit{design}$ and
$\mathit{assumes}(m)$ as pairs $\langle \mathit{pr}, \mathit{cl} \rangle$
and sharing the metavariable $\mathit{pr}$ between them expresses, and
which a single $\sqsubseteq$ on the diamond would not. The other side of
the premise reads what the data actually is, from the schema and the two
constraints alone:
{\everydisplay{} 
\begin{align*}
\ty{\mathit{differing}} &\ty{= \{ f \in F \mid \kappa_1(f), \kappa_2(f)
   \text{ defined, unequal} \}} \\
\ty{\mathit{distinguishing}} &\ty{= \text{least } X \supseteq
   \mathit{differing} \text{ such that}} \\
  &\ty{\qquad (f \mapsto g) \in D \text{ and } g \in X \text{ imply } f \in X} \\
\ty{\mathit{residual}} &\ty{= F \setminus \mathit{distinguishing}} \\
\ty{\mathit{paired}} &\ty{\iff I \subseteq \mathit{residual}} \\
\ty{\mathit{clustered}} &\ty{\iff \gamma \neq \cdot \text{ and not }
   (\kappa_1(\gamma) = \kappa_2(\gamma) \text{ defined})}
\end{align*}
}
and $\ty{\mathit{design}(\Sigma,\kappa_1,\kappa_2) = \langle
\mathit{paired}, \mathit{clustered} \rangle}$.

\textbf{Why $\mathit{paired}$ needs the whole key.}
The soundness argument of Section~\ref{sec:metatheory} needs the whole key
to survive, $I \subseteq \mathit{residual}$: if every field
of $I$ is residual, two distinct observations agreeing on the residual tuple
agree on all of $I$; both also agree on every field the selection keeping
them pinned, and $I$ together with the pinned fields identifies
observations (Definition~\ref{def:conforms}), so they are the same
observation. A single surviving field of a compound key
would not give this: with $I = \{\mathit{task}, \mathit{attempt}\}$ and only
$\mathit{task}$ residual, two attempts at the same task would collide on the
residual key, and the alignment would have to keep one and drop the other, a
$p$-value that depends on row order.

\textbf{Why $\mathit{clustered}$ does not consult $\mathit{distinguishing}$.}
Pinning $\mathit{repo}{=}A$ against $\mathit{repo}{=}B$ is a cross-cluster
comparison, each side confined to a single repository, the most clustered
case there is. The one case genuinely safe to call unclustered is both sides
pinned to the \emph{same} value, which confines every observation to one
cluster by construction, and that is the only condition the formula above
checks. ``Safe'' is relative to the dependence model
Theorem~\ref{thm:folding}'s antecedent and (H2) already state (arbitrary
dependence within a shared cluster value, independence across distinct
ones): pinning both sides to one cluster value removes every pair of
observations that model permits to be dependent. The guard
$\gamma \neq \cdot$ is necessary: without it, a
schema that declares no cluster level at all ($\gamma = \cdot$) makes
$\kappa_1(\gamma)$ and $\kappa_2(\gamma)$ queries against a symbol that is not a
field, so ``$\kappa_1(\gamma) = \kappa_2(\gamma)$ defined'' is vacuously false and
$\mathit{clustered}$ would read true for every comparison under a schema
that never mentions clustering at all.

\textbf{Non-degeneracy.} $\mathit{design}$ is \emph{partial}. It is defined
only when
\[
  \ty{\mathrm{dom}(\kappa_1) = \mathrm{dom}(\kappa_2)
   \quad\text{and}\quad \mathit{differing} \neq \emptyset,}
\]
and a claim whose two constraints fail either condition is refused, with
the reason stated, and is assigned no design; we call this uniform
response the \emph{opaque-key policy}, and every case below that cannot be
read off the keys receives it.
Both conditions are needed, and neither is a formality. $\mathit{differing}$
asks which fields the two sides pinned to \emph{different} values, so it
requires both sides to pin $f$ at all; with $\kappa_1 = \emptyset$ and
$\kappa_2 = \{\mathit{repo} \mapsto A\}$ it is empty, $\mathit{distinguishing}$
is empty, $\mathit{residual}$ is all of $F$, and $\mathit{paired}$ holds for
any schema that identifies observations at all. That is a comparison of a
set against its own subset typing as paired, and the degenerate case
$\kw{claim}(\ell, v \mathbin{\kw{vs}} v)$ (a set against itself) typing as
paired too. $\mathit{clustered}$ holds in the same situation, so the design
lands at the top of the lattice and would serve the cluster sign-flip on two
overlapping selections; the residual key then aligns each side against
itself or drops the non-overlap silently, which is exactly the collision
whose cost is measured below. Requiring $\mathrm{dom}(\kappa_1) =
\mathrm{dom}(\kappa_2)$ rules out the nesting case, and requiring
$\mathit{differing} \neq \emptyset$ rules out the identical one. Neither is
a design the checker should guess at, and the language refuses any
comparison whose shape it cannot read off the keys.

Appendix~\ref{app:design} records two further details of the closure: why
it is stated as a least fixed point running parent to child rather than as
a closure along $D$'s edges, and why restricting $D$ to single-field
dependencies loses nothing the rest of the development needs.

\textbf{Instrument derivation.} Each point of the lattice names its
instrument, and the order is read in one direction only. When the analyst
names no mechanism, the library derives one from the design by a rule stated
in three lines: a clustered design demands the cluster sign-flip permutation
test, whose unit of randomness is the cluster; a paired one demands
McNemar's exact test on the discordant pairs; an independent one receives
Fisher's exact test. The cluster sign-flip statistic is the sum, within each
cluster, of matched pairs where the left side passed and the right did not
minus pairs where the right passed and the left did not, taking the
absolute value of the sum of these per-cluster net differences across
clusters; the two-sided $p$-value is the fraction of the $2^k$ sign
assignments over $k$ nonempty-difference clusters whose statistic is at
least as extreme as the one observed, exact rather than
asymptotic~\cite{winkler-permutation}, and computed by
dynamic-programming convolution rather than enumeration
(Appendix~\ref{app:runtime} reports the difference that makes at large
$k$). One function is the sole authority for the rule,
called by the checker and the monitor alike, so the two cannot drift
(Section~\ref{sec:metatheory}). In the premise's notation,
$\mathit{assumes}(\text{Fisher}) = \langle \text{no}, \text{no} \rangle$,
$\mathit{assumes}(\text{McNemar}) = \langle \text{yes}, \text{no} \rangle$,
and $\mathit{assumes}(\text{sign-flip}) = \langle \text{yes}, \text{yes}
\rangle$. A
mechanism may still be named explicitly. On the clustering coordinate the
premise refuses it only for assuming \emph{away} structure that is present,
which produces the second refusal of Section~\ref{sec:analysis}, where
Fisher's exact test is named against data that turned out paired and
clustered. Over-modeling clustering is admitted, for the reason the clustering
coordinate was granted its slack above; under-modeling is the checker's to
stop. The pairing coordinate is not slack
in this way at all, in either direction; the next subsection states why and
what happens when a mechanism assumes pairing the data does not have.

\textbf{The fourth lattice point.} The fourth lattice point is refused: no
instrument is offered for it. Three points carry instruments; the fourth,
clustered without pairing, does not, because the cluster sign-flip is a
paired procedure, built on
differences of matched observations, and on unmatched data the
correspondence it consumes does not exist. The point is not exotic: two
disjoint sets of instances drawn from the same repositories have no such
correspondence, and the derivation above reads it off their keys correctly. It
therefore receives the opaque-key policy (refused with the reason
stated) rather than being priced by an instrument built for a different
design. Pricing it anyway has a measured cost: the alignment key
degenerates to the cluster alone and the resulting $p$-value depends on
the order the rows arrived in, turning $p = 1.0$ into $p = 0.0078$ on data
carrying no signal at all (Appendix~\ref{app:design}). The refusal is a
fact about \sys{Tacet}'s three-mechanism menu, not about the design point
itself: a two-sample cluster-level permutation test, or a cluster-robust
variance estimator over the aggregated cluster
means~\cite{cameron-miller-cluster}, is a standard instrument for exactly
this design, and either could serve it; neither is implemented here, and
the opaque-key policy refuses the point rather than serving it with the
one instrument the library has and the wrong correspondence for.

The refusal also marks where the asymmetry above stops holding, and it is
why the premise is two chains and not one order. A single $\sqsubseteq$
on the diamond would admit the sign-flip on this point by the same rule
that lets a paired test run harmlessly on independent data, which is
precisely what the pairing coordinate's demand for equality excludes.
Equality only checks that
$\mathit{assumes}(m)$ matches what $\mathit{design}$ computed, and says
nothing about whether $\mathit{design}$ computed the truth, which rests on
the schema's declared key and can be wrong if the key is mis-declared. The
runtime alignment therefore checks its own injectivity independently: any
residual key that fails to identify observations uniquely on either side
is refused rather than silently collapsed. The check is implemented in the
alignment routine both mechanisms share, so it covers every mechanism that
aligns by key; the first line of defense is sound given the key is
genuine, and the second holds when it is not.

Failing this premise is a wrong-instrument failure, distinct from a budget
failure: the claim is priced correctly and the instrument is wrong
(Section~\ref{sec:analysis}).

\section{Wealth Transformers}
\label{sec:grade}

With invest fraction $0 < i \le 1$ and refund $\omega \ge 0$, one claim's
transformer is
\begin{equation}
  \ty{T_p(W) \;=\; (1-i)\cdot W \;+\; \omega \cdot [\, p \le i \cdot W \,]}
  \label{eq:transformer}
\end{equation}
and composition is function composition. Every numeric example in this
paper, including the ones already given in Sections~\ref{sec:intro}
and~\ref{sec:analysis}, uses $i = 0.5$, $w_0 = \alpha/2$, and
$\omega = \alpha/2$, the instantiation Section~\ref{sec:mfdr}'s (H4)
requires ($0 \le \omega \le \alpha \le 1$ and $0 \le w_0 \le \alpha\eta$,
which this choice of $w_0$ satisfies for every $\alpha \le 1/2$, the range
every example in this paper uses), unless stated otherwise. This choice is
aggressive, staking half the pool on every claim, but the leaderboard's
headline numbers are no artifact of it: rerunning
Section~\ref{sec:leaderboard}'s reading order and strongest-first orderings at
$i \in \{0.1, 0.25, 0.5, 0.75, 0.9\}$ leaves both unchanged (0 of 8{,}911 and
6{,}545 of 8{,}911 at every value tried), because the claims that sink reading
order and the claims that fund strongest-first are the same ones regardless
of how large a bite each stake takes; for strongest-first specifically,
``The ceiling stake'' below gives the sharper reason. Every
law in this section is
property-tested over a uniform grid of 400 wealth values, and the three the proof of
Theorem~\ref{thm:static} depends on (monotone in $W$, antitone in $p$, and
the no-refund lower bound) are additionally machine-checked in Lean~4,
algebraically rather than by sampling; we use ``machine-checked'' for the
Lean proofs alone and ``property-tested'' for numeric sampling. The two laws
the grid alone stands behind (associativity and non-commutativity) do not
need the grid's density to hold: associativity is a fact about function
composition, true of any $T_p$ whatever, and non-commutativity is
established by a single exhibited witness pair, not by how many wealth
values agree elsewhere.

\textbf{Scope of Equation~\ref{eq:transformer}.} It is one investment rule,
not a family. Equation~\ref{eq:transformer} is \lstinline|fixed_fraction|:
the stake is a fixed fraction of the \emph{current} pool, so the transformer
is a function of $W$ alone and Appendix~\ref{app:proof}'s induction can
carry $\hat W_j \le W_j$ as its whole invariant. Rules whose stake depends
on the verdict history rather than the current pool do not have that form.
The implementation ships five of them (\lstinline|tacet.qude|); QUDE's third
rule, for instance, sets $\alpha_j$ from the wealth at the most recent
rejection, which is not recoverable from $W_j$, so a monotone-in-$W$
argument does not reach it. Equation~\ref{eq:transformer},
Appendix~\ref{app:semantics}, Appendix~\ref{app:proof}, and the
mechanization are all stated for \lstinline|fixed_fraction| only. For the
induction to go through, a different rule would have to make its stake a
monotone function of the state the invariant tracks; we have not established
that for any of the five (Section~\ref{sec:threats}).

\textbf{Where the grades live.} Grades are wealth transformers, and they
compose in $(\mathbb{R} \to \mathbb{R}, \circ, \mathit{id})$, the monoid of
endofunctions on wealth: $\circ$ is associative and $\mathit{id}$ is its
unit, which \clc{\textsc{T-Seq}} and \clc{\textsc{T-Note}} need. The
set $\{T_p\}$ itself is \emph{not} a monoid: it is not closed under $\circ$
and does not contain $\mathit{id}$.
Within that ambient monoid, $T_p$ is monotone in $W$ (Lemma~\ref{lem:mono}),
and $T_p$ is \emph{antitone} in $p$ (Lemma~\ref{lem:anti}), since a smaller
$p$-value only ever earns more refund.

\textbf{Non-commutativity.} The dynamic transformer does not commute. For
example, at $W = 0.0005$, composing $T_{0.0002}$ and then $T_{0.0011}$
gives $0.037625$, while the other order gives $0.000125$. This is the order
dependence of Section~\ref{sec:analysis}, stated algebraically.

\textbf{The ceiling stake.} $T_p$ has a fixed point at $W^* = \omega/i$: at
that wealth, a claim clearing the stake ($p \le i W^* = \omega$) sends
$T_p(W^*) = (1-i)(\omega/i) + \omega = \omega/i = W^*$, so wealth returns
exactly to $W^*$ and a run of accepted claims holds there indefinitely,
the ceiling a run of accepted claims reaches and holds before the
evidence runs out; at the study defaults ($i = 0.5$,
$\omega = \alpha/2 = 0.025$), $W^* = 0.05$. The stake charged at the fixed point,
$i \cdot W^* = \omega$, does not depend on $i$: whatever fraction of the
pool a claim stakes, the wealth level adjusts so the stake itself is
always $\omega$ once the pool reaches its ceiling. A discovery-rich
sequence spends most of its length there, so its outcome is governed by
whether each claim's own $p$-value clears $\omega$, a threshold the choice
of $i$ never touches; this is why Section~\ref{sec:leaderboard}'s
strongest-first count is unaffected by the $i$-sweep above, once the
short transient before wealth first reaches $W^*$ has passed, and it is
why Section~\ref{sec:bbh}'s budget keeps exactly the tasks naive
thresholding at $0.025$ would.

\textbf{Why these defaults and not the permitted maximum.} (H4) permits
$\omega$ up to $\alpha$ and $w_0$ up to $\alpha\eta$, strictly more
generous than the study defaults $\omega = w_0 = \alpha/2$, and the
ceiling-stake identity just established puts the cost of the more
conservative choice in closed form: at $\omega = \alpha$ the ceiling stake
a discovery-rich sequence eventually faces is $\alpha$ itself, twice the
shipped one, so BIG-Bench Hard's two refused tasks, whose $p$-values fall
between the two stakes (Section~\ref{sec:bbh}), would
clear at the more generous choice; the
``accounting tax'' that section reports is therefore a fact about this
default, not a law the design forces. We use $\alpha/2$ for both, symmetric
and simple to state, in preference to either boundary permitted, and every
numeric example in this paper inherits that choice; a study wanting less
conservatism can raise $\omega$ toward $\alpha$ without leaving (H4), at
the cost stated just above. The same choice of $w_0 = \alpha/2$ also sets the
price of the single-claim analysis, the modal case in the population
Section~\ref{sec:intro} cites: the first claim's stake is
$i \cdot w_0 = \alpha/4$ under the shipped defaults, so an analysis making
one comparison, textbook-style, is tested at a quarter of the nominal
level, for no multiplicity reason, since nothing has been asked of the
pool yet. This is not a defect specific to
$\alpha/2$: any $w_0 < \alpha/i$ opens below $\alpha$, and $w_0 = \alpha\eta$,
(H4)'s permitted maximum, narrows but does not close the gap
($\eta = 1 - \alpha$ under the shipped $\eta$, so $w_0 = \alpha(1-\alpha)$
still opens below $\alpha$ for every $\alpha \in (0,1)$). A single-claim
analysis is therefore always priced somewhat more strictly than a bare
textbook test under this family of rules, a fact the defaults did not
create and cannot fully undo.

\textbf{The static transformer.} The static no-refund transformer, by
contrast, does commute. A checker without access to values cannot know
whether the refund arrives, so the natural static grade is the worst case
\begin{equation}
  \ty{T_\bot(W) = (1-i)\cdot W}
  \label{eq:static-transformer}
\end{equation}
This is scalar multiplication, and scalars commute. Running the two programs
of Figure~\ref{fig:order} through it, both end at the same wealth, while the
runtime separates them cleanly. Absent a declared bound, affordability is
therefore a runtime property, exactly as \sys{Duet}~\cite{duet} certifies a
composition without computing the noise.

\textbf{Declared bounds.} Declared bounds buy the static answer back. A
claim promising $p \le \bar p$ can be evaluated statically \emph{and priced at
$T_{\bar p}$}, and the promise is conservative precisely because $T$ is
antitone in $p$ (Lemma~\ref{lem:anti}): an honest run's realized price is at
most $\bar p$, so $T_{\bar p}(\hat W) \le T_p(\hat W)$ and the shadow wealth
never overstates the real one. This is the inequality Theorem~\ref{thm:static}
turns on at every declared claim. Running
Figure~\ref{fig:order} through declared bounds reproduces the runtime
separation exactly. Importantly, pre-registration in this setting is far
more than a methodological convention external to the language: it
functions as the mechanism that moves affordability from run time to
type-checking time.

Three independent design pressures converge on this feature: fixing the
direction of a one-sided test, closing data-dependent control flow, and
making the grade static.

\section{Metatheory}
\label{sec:metatheory}

Theorem~\ref{thm:static} below is a short induction, and it stays short
because Section~\ref{sec:claims}'s two grades were chosen to make it so.
Not everything in this section is short on the same terms. Theorem~\ref{thm:design} is a genuine argument
about which fields distinguish two samples, not a one-line consequence of
monotonicity, and Theorem~\ref{thm:mfdr}'s supermartingale is
measure-theoretic rather than a finite calculation
(Section~\ref{sec:mechanization}).

\begin{thm}[Static under-approximation]
\label{thm:static}
Let $0 < i \le 1$ and $\omega \ge 0$, let the checker and the runtime
monitor run the same program from the same initial wealth $w_0$, and let
every declared bound be honest at run time (Definition~\ref{def:honest},
Appendix~\ref{app:semantics}). Then every claim the checker certifies is
accepted by the runtime monitor, whether or not the checker certifies the
claims around it; in particular, when the checker accepts the whole
program, the monitor accepts every claim in it that declared a bound.
\end{thm}

\noindent The three side conditions are load-bearing hypotheses. $i > 0$
makes the chain $\bar p \le i \hat W \le i W$ go through, $i \le 1$ keeps
$(1-i)W$ non-decreasing in $W$, and $\omega \ge 0$ makes the no-refund
transformer a lower bound. The Lean development
carries all three as explicit hypotheses, and every statement in this
paper that mentions $i$ or $\omega$ assumes them.

\noindent Antitonicity in $p$ and monotonicity in $W$ give this directly.
Appendix~\ref{app:proof} spells the argument out as an induction over the
operational semantics of Appendix~\ref{app:semantics}, rather than leaving
it at the one-line justification above; that proof is additionally
machine-checked in Lean~4 (Section~\ref{sec:mechanization}).

\noindent The statement is per claim rather than per program, and the
shipped tooling depends on that. The audit's fold does not stop at a
refused claim: it charges the refusal and keeps folding, reporting a
verdict for every claim behind it, so Section~\ref{sec:implementation}'s
transcript can read ``2 of 3 decided claims are affordable'', and the
validation below deliberately certifies claims after an earlier refusal
has already drained the shadow pool. A statement hypothesized on the
checker accepting the whole program would cover neither, since a program
with one refused claim falls outside that hypothesis; the invariant the
proof carries, $\hat W_j \le W_j$, holds through certified, refused, and
undeclared claims alike (Appendix~\ref{app:proof}).

\textbf{Empirical validation.} We also test the theorem, by generating
random claim sequences, declaring bounds $\bar p = p \times \mathit{slack}$,
and comparing the checker against the monitor: 599 programs across two
corpora (4{,}880 claims), one drawn over a fixed code-model comparison
corpus (8 open-weight models scored on 15 paired coding tasks in two
prompting modes; Appendix~\ref{app:validation}) and one tuned so that $p$-values
span decades and certified claims land across the full range of wealth
trajectories, including after refusals have drained the pool and after
refunds have refilled it. At honest bounds ($\mathit{slack} \ge 1$) the
checker never certifies a claim the monitor refuses, in every stratum;
dishonest bounds, a negative control, break the implication on both
corpora, as they must for the test to mean anything. The same runs also
price imprecision: on data whose resolution is low, the cost of an
overstated bound is a cliff (within $1.5\times$ of the truth free,
threefold fatal), and on data with resolution to spare it is a slope, so
pre-registration buys the most when the declared bound is close to the
truth relative to the data's own resolution.
Appendix~\ref{app:validation} reports the corpus construction and the full
counts (Table~\ref{tab:validation}).

\textbf{One authority, enforced by a gate.} The mechanism a design demands
is decided by one function, and the checker and the monitor both call it,
so the two sides price the same number by construction; a checker whose
bounds came from one test while the monitor priced another would be
recorded as violating a theorem it does not violate. The same judgment is
enforced repository-wide: a gate walks this project and refuses any test
statistic computed outside the library unless the module derives its design
from the artifact keys, or carries a written reason for not doing so
(timed in Appendix~\ref{app:runtime}). Since the thesis of this paper is that
this judgment should be discharged by a machine rather than remembered by
a person, we apply the discipline to our own repository.

\textbf{The design check.} The design check is static. Computing
$\mathit{design}$ from the schema and constraints agrees with a value-level
inspection on all 210 pairs of selections our test corpus can express: the code-model
comparison corpus, 7 models each selectable whole
or narrowed to one of its two prompting modes, 21 selections in all, over a
schema with one identifying field (\lstinline[style=tacetcode]|task|), two
declared dependencies (\lstinline[style=tacetcode]|task| determines both
\lstinline[style=tacetcode]|mode| and \lstinline[style=tacetcode]|col|),
and no declared cluster field; $\binom{21}{2} = 210$ is every pair those 21
selections form, and the two purpose-built cases below extend the corpus
along the two coordinates it does not otherwise exercise, a compound
identifying key and a declared cluster field. The same check also agrees
on the running example itself, where the leaderboard's declared schema
predicts paired+clustered and the cluster sign-flip for the adjacent-rank
comparisons before any data is read and the realized keys of all 134
submissions return the same verdict (Section~\ref{sec:implementation});
this agreement is why the check belongs in \clc{\textsc{T-Claim}} rather
than in the monitor. The general statement this evidence gestures at is not true without
saying what the schema's declarations are declarations \emph{of}:
$\mathit{design}$ reads $\Sigma$ and $\kappa_1, \kappa_2$ alone, never the
run, so any statement connecting it to the run needs a name for what makes a
run consistent with what $\Sigma$ claims about it.

\begin{defi}[Conforming run]
\label{def:conforms}
A run $r$ \emph{conforms} to schema $\Sigma = (F, I, D, \gamma)$ at pin
domain $S \subseteq F$ if (i) every
declared dependency $(f \mapsto g) \in D$ holds in $r$ (for all observations
$o, o' \in r$, $f(o) = f(o')$ implies $g(o) = g(o')$) and (ii) $I$ is a
genuine key of $r$ \emph{relative to $S$}: no two distinct observations of
$r$ agree on every field of $I \cup S$.
\end{defi}

\noindent Clause (ii) is relative rather than global, and the leaderboard is
why. Two submissions' rows agree on every field the submission field does
not determine, so no $I$ disjoint from it keys the whole run, while any $I$
containing it fails $I \subseteq \mathit{residual}$; a global clause (ii)
would therefore leave the $\mathit{paired}$ case of Theorem~\ref{thm:design}
unsatisfiable on exactly the comparisons it exists for. The alignment
consumes, and (ii) states, one requirement: that $I$ key the observations
one selection keeps, and those agree on every field of $S$ by construction.

\begin{thm}[Design premise soundness]
\label{thm:design}
Let $s_1, s_2$ be samples drawn from a run $r$,
$\kappa_i$-consistent (every observation of $s_i$ satisfies every pin
$\kappa_i$ declares), and let $r$ conform to $\Sigma$ at
$\mathrm{dom}(\kappa_1)$. If $\mathit{design}(\Sigma, \kappa_1, \kappa_2) =
\langle \mathit{paired}, \mathit{clustered} \rangle$ is defined, then:
(a) if $\mathit{clustered}$ is false and $\gamma \neq \cdot$, every
observation of $s_1 \cup s_2$ shares one value of $\gamma$; (b) if
$\mathit{paired}$ is true, no two distinct observations within $s_1$, nor
within $s_2$, agree on every field of $\mathit{residual}$.
\end{thm}

\begin{proof}
(a) $\mathit{clustered}$ false and $\gamma \neq \cdot$ together mean
$\kappa_1(\gamma) = \kappa_2(\gamma)$, some value $v$; $\kappa_i$-consistency
gives $\gamma(o) = v$ for every $o \in s_i$, so every observation of $s_1
\cup s_2$ has $\gamma = v$.

(b) First, $\mathit{distinguishing}$ under-approximates the run's actual
disjointness: for every $f \in \mathit{distinguishing}$, the values $f$
takes on $s_1$ and on $s_2$ are disjoint sets. By induction on the fixed
point computing $\mathit{distinguishing}$. Base case, $f \in
\mathit{differing}$: $\kappa_1(f) = x \neq y = \kappa_2(f)$ are both
defined, so consistency gives $f \equiv x$ on $s_1$ and $f \equiv y$ on
$s_2$, disjoint singletons since $x \neq y$. Step, $f$ added because
$(f \mapsto g) \in D$ and $g$'s values are already known disjoint across
$s_1, s_2$: if some value were common to $f$ on both sides, witnesses
$o_1 \in s_1, o_2 \in s_2$ with $f(o_1) = f(o_2)$ would, by
conformance~(i), force $g(o_1) = g(o_2)$, contradicting $g$'s disjointness;
so $f$'s values are disjoint too.

Now suppose $\mathit{paired}$, i.e.\ $I \subseteq \mathit{residual} = F
\setminus \mathit{distinguishing}$, and suppose two distinct observations
$o \neq o'$ within $s_1$ (the argument for $s_2$ is symmetric, since
$\mathrm{dom}(\kappa_2) = \mathrm{dom}(\kappa_1)$ wherever
$\mathit{design}$ is defined) agree on
every field of $\mathit{residual}$, in particular on every field of $I$
since $I \subseteq \mathit{residual}$. By $\kappa_1$-consistency both also
agree on every field of $\mathrm{dom}(\kappa_1)$, each taking the pinned
value, so they agree on every field of $I \cup \mathrm{dom}(\kappa_1)$,
and conformance~(ii) forces $o = o'$, contradicting distinctness.
\end{proof}

\noindent This soundness statement is deliberately one-directional:
$\mathit{design}$ may still say \emph{clustered} or \emph{paired} where the
realized sample would not have needed it (an unpinned cluster field, or a
residual set larger than $I$ actually requires), the same over-approximation
already licensed for the fourth lattice point above. The theorem rules out
the dangerous direction: $\mathit{design}$ licensing a mechanism whose
alignment then silently misaligns observations. Our 210-pair corpus does not
exercise this: its one identifying field is never compound, and it never
declares a cluster field, so neither the full-key requirement on
$\mathit{paired}$ nor the pinned-apart case of $\mathit{clustered}$ is put
to the test there. Two purpose-built cases close that gap and check the
theorem directly on data rather than by argument alone: a compound key with
one field excluded from the residual through the closure rather than by
direct pinning, where the data does collide and the mechanism's alignment
refuses it; and a cluster field pinned to two different repository values,
where two repositories are visibly present and $\mathit{design}$ agrees.

The theorem rules out less than its name suggests, and
the gap is worth stating plainly. Clause (b) is within-side: it says two
distinct observations inside $s_1$ cannot collide on the residual key, and
symmetrically for $s_2$, so the residual key is an injective label for each
sample on its own. It says nothing about the two sides
against each other: that a $\mathit{residual}$ key shared across $s_1$ and
$s_2$ denotes the \emph{same} experimental unit on both, the fact a paired
test's correspondence actually rests on, follows only from what $I$ is
declared to mean, not from anything clause (b) proves. A schema is free to
declare an $I$ that is a genuine key of the run (satisfying
Definition~\ref{def:conforms}) and still not identify observations across
selections the way the analyst intends, if the field it names measures
different quantities on the two sides; Theorem~\ref{thm:design} is silent on
that fact by construction, because $\mathit{design}$ reads $\Sigma$ and
the constraints alone, never what an identifying field is supposed to
identify. The runtime alignment additionally checks within-side
injectivity for itself (below); cross-side meaning is not checked
anywhere, and is taken on trust the same way the artifact and schema
declarations of Section~\ref{sec:declarations} are.

Theorem~\ref{thm:design} is conditional on Definition~\ref{def:conforms},
and nothing in $\mathcal{T}$ checks it: a schema whose declared $D$ omits a
real dependency, or whose $I$ fails to key the observations a selection
keeps, voids the theorem silently, the same way a dishonest declared bound
voids Theorem~\ref{thm:static}. The theorem's other premise,
$\kappa_i$-consistency of the two samples, carries no such trust: every
estimation term's denotation is consistent with its own constraint
(Lemma~\ref{lem:footprint}, Appendix~\ref{app:estimation}), so for samples
the language itself produced the premise holds by construction, and it
remains an assumption only for data that reached a mechanism around the
language. The alignment itself is no longer silent:
whichever mechanism the design licenses re-derives its own residual key from
the actual observations and refuses, rather than collapses, if that key is
not injective, for every mechanism that aligns by key, so a schema that
violates Definition~\ref{def:conforms} in a way that would corrupt an
alignment is caught there even though it is never caught here.

\subsection{mFDR Control}
\label{sec:mfdr}

The guarantee the title's accounting carries is a statement about the
runtime monitor:
the process that prices claim $j$ at the stake $\alpha_j = i\,W_{j-1}$,
observes the realized $p$-value, and pays out $\omega$ on a rejection. Let
$R(t)$ count the rejections among the first $t$ claims, $V(t)$ the
rejections whose null hypothesis is true, and let $\mathcal{F}_{j-1}$ be the
$\sigma$-algebra generated by everything the account records before claim
$j$ is priced. The wealth $W_{j-1}$, and with it the stake $\alpha_j$, is
$\mathcal{F}_{j-1}$-measurable: the level of test $j$ is fixed before $p_j$
is seen, from the recorded past alone, and the mechanization proves the
wealth process adapted rather than assuming it. (H3) below is stated with respect to this
same, narrow filtration (what the account records) and is only as
plausible as that record is complete: an analyst whose choice of test $j$
depends on something the account never logged is conditioning on
information $\mathcal{F}_{j-1}$ does not contain, and (H3) can hold for the
account's filtration while failing for the analyst's true one. (H1)'s
completeness closes this gap.

\begin{thm}[mFDR control]
\label{thm:mfdr}
Let $0 < i \le 1$ and $\omega \ge 0$, let $\pi$ be a declared plan
(Appendix~\ref{app:plan}), and let $c_1; \dots; c_n$ be
the claims the monitor prices against it, whether or not the checker
accepts $\pi$. Suppose:
\begin{enumerate}
\item[(H1)] \emph{Complete account.} $\pi, c$ is a \emph{complete account}
  (Definition~\ref{def:complete}, Appendix~\ref{app:plan}): the claims
  priced are all the claims made.
\item[(H2)] \emph{Within-claim validity.} Every claim's raw observations
  are dependent only in the way its own comparison's design
  declares, whether or not that claim was ever pre-registered: the
  antecedent of Theorem~\ref{thm:folding}, under which each
  mechanism's $p$-value is valid for its declared estimand.
\item[(H3)] \emph{Conditional super-uniformity at the stake.} For every $j$
  whose null hypothesis is true,
  $\Pr[\,p_j \le \alpha_j \mid \mathcal{F}_{j-1}\,] \le \alpha_j$.
\item[(H4)] \emph{Bounded refund.} $0 \le \omega \le \alpha \le 1$ and
  $0 \le w_0 \le \alpha\,\eta$, with $\eta > 0$.
\end{enumerate}
Then for every horizon $t$,
$\mathrm{mFDR}_\eta(t) = \mathbb{E}[V(t)]\,/\,(\mathbb{E}[R(t)]+\eta) \le
\alpha$.
\end{thm}

\noindent The independence hypothesis (H2) is stated with care: read
literally as ``mutually independent'' it would exclude the paired and
clustered designs Section~\ref{sec:evaluation} actually serves, so the
antecedent Theorem~\ref{thm:folding} needs is the weaker one stated
above (dependent only in the declared way), and our design check
(Section~\ref{sec:lambdat}) lets a program compute that declaration
statically instead of assuming it, soundly, by Theorem~\ref{thm:design}, provided the
schema's declared dependencies and identifying key are genuine facts about
the run (Definition~\ref{def:conforms}). That proviso keeps this from
being a full discharge: it moves the assumption from ``these artifacts
are independent'' to ``this schema's declarations are true of the data'',
smaller and checkable in different places, but still an assumption.

\begin{proof}[Proof sketch]
Foster and Stine's supermartingale argument~\cite{foster-stine}, run on
$X(t) = V(t) - \alpha R(t) + W(t)$. One claim changes wealth by $\Delta W_j
= -\alpha_j + \omega\,[p_j \le \alpha_j]$, so on a true null (H3) gives
$\mathbb{E}[\Delta X_j \mid \mathcal{F}_{j-1}] \le (1 + \omega -
\alpha)\,\alpha_j - \alpha_j \le 0$ exactly when $\omega \le \alpha$, and on
a non-null $\Delta X_j = (\omega - \alpha)\,[p_j \le \alpha_j] - \alpha_j
\le 0$ holds pathwise. Hence $\mathbb{E}[X(t)] \le X(0) = w_0$. Because the
stake never exceeds the pool ($0 < i \le 1$), $W(t) \ge 0$ on every path,
so $\mathbb{E}[V(t)] \le \alpha\,\mathbb{E}[R(t)] + w_0 \le
\alpha\,(\mathbb{E}[R(t)] + \eta)$ by (H4). The argument is machine-checked
in Lean~4 (Section~\ref{sec:mechanization}): the mechanized wealth process
steps by the transformer $T$ of Section~\ref{sec:grade} itself, (H3) enters in its
standard online-testing form $\mathbb{E}[\,[p_j \le \alpha_j] \mid
\mathcal{F}_{j-1}] \le \alpha_j$ and is consumed once through the tower
property, and the supermartingale computation is a single induction on the
horizon: $\mathbb{E}[V(t)] +
\mathbb{E}[W(t)] \le \alpha\,\mathbb{E}[R(t)] + w_0$.
\end{proof}

\noindent Without (H3) or (H4) the statement is false. (H3) is the
hypothesis Foster and Stine's argument
actually runs on, and it licenses choosing hypothesis $j+1$ after seeing
outcome $j$; (H2) is a \emph{within-claim} condition and says
nothing across claims, so on its own it would admit a mechanism returning
$p \equiv 0$, which satisfies every other antecedent and refutes the
conclusion. Benjamini and Yekutieli's extension of FDR control to arbitrary
dependence addresses the across-claim analogue in the classical
setting~\cite{benjamini-yekutieli}: a condition on dependence across the
family of tests, not within one test's own artifacts, so it is (H3)'s
territory, not (H2)'s. (H4) makes $X$ a supermartingale: at $\alpha = 0.05$,
$w_0 = 0.025$, $i = 0.5$ and $\omega = 0.5$, the first test runs at level
$0.0125$, a rejection lifts wealth to $0.5125$ and the next level to
$0.256$, and so on; under a complete null $\mathbb{E}[V] = \mathbb{E}[R]$
grows without bound and $\mathrm{mFDR}_\eta \to 1$. \sys{Tacet}'s defaults,
$w_0 = \omega = \alpha/2$ with $\eta = 1 - \alpha$, satisfy (H4) for every
$\alpha \le 1/2$, mechanically. The stronger side
condition $w_0 \le \alpha\eta(1-\omega)$ our earlier analysis derived is
sufficient but not necessary: the mechanized proof needs only $\omega \le
\alpha$ and $w_0 \le \alpha\eta$.

Two of the theorem's obligations are discharged by the calculus itself.
\emph{Every claim passes through a mechanism}: $\clc{\mathsf{Claim}}$'s
only introduction forms are \clc{\textsc{T-Note}}, which contributes to
neither $R$ nor $V$, and \clc{\textsc{T-Claim}} /
\clc{\textsc{T-OneSided}}, which both demand a mechanism satisfying the
design premise. \emph{The investment rule never stakes more than the pool
holds}: $\alpha_j = i \cdot W_{j-1} \le W_{j-1}$ by construction, since
$0 < i \le 1$, so that arithmetic alone makes $W(t) \ge 0$ pathwise, with no
extra assumption: the non-negativity Foster and Stine's own rule instead
has to earn with a cap on the stake (Appendix~\ref{app:semantics}).

(H1) is the obligation nothing in Section~\ref{sec:lambdat} or
Section~\ref{sec:grade} can discharge on its own. $\mathrm{mFDR}_\eta$ is a
guarantee about the sequence of claims actually priced, and
Section~\ref{sec:limitations} measures what an incomplete account costs an
analyst who runs a claim, dislikes its price, and deletes it before it is
ever logged.
Definition~\ref{def:complete} names the assumption that rules this out, and
it is unchecked by any rule above, and an external, timestamped record has
to supply it (Section~\ref{sec:implementation});
\clc{\textsc{T-Plan}} contributes legibility instead
(Appendix~\ref{app:plan}).

(H3) is likewise the analyst's to meet, and our own running example does
not meet it: the 8{,}911 leaderboard comparisons of Section~\ref{sec:leaderboard}
are all computed from the same pool of SWE-bench Verified instances, so
that $p$-value sequence is heavily dependent, and we have not verified (H3)
for it or for the BIG-Bench Hard study (Section~\ref{sec:evaluation}
states what the case studies do and do not certify). Making the dependence
structure a typed input relocates the assumption into the open without
discharging it, and our reading of QUDE's independence gap in the related
work should be taken in that light.

\subsection{Composing the Two Theorems}
\label{sec:bridge}

Theorem~\ref{thm:static} and Theorem~\ref{thm:mfdr} do not compose by
conditioning, and the reason belongs before the theorem that does
compose them. Theorem~\ref{thm:static}'s conclusion holds on the
data-dependent event that every declared bound is honest, whereas
Theorem~\ref{thm:mfdr} is an unconditional statement about expectations;
conditioning a true null's $p$-value on $\{p_j \le \bar p_j\}$ makes it
stochastically \emph{smaller} than uniform, which is the direction that
destroys (H3). The composite below therefore never conditions. Honesty
enters as a property of the analyst's declaration \emph{procedure} (the
bounds $\bar p_j$ are fixed by the plan, before any data), and the monitor
prices every claim at its own stake on the realized $p$-value whether or
not its declared bound was honest: a violated bound is charged first and
refused after (Section~\ref{sec:implementation}), never skipped. Error control is therefore
enforced by the monitor unconditionally, and what checker acceptance adds
is the static feasibility guarantee on honest paths.

\begin{thm}[Checker acceptance implies mFDR control]
\label{thm:bridge}
Assume the hypotheses of Theorem~\ref{thm:mfdr}, and let the declared
bounds $\bar p_1, \dots, \bar p_n$ be fixed by the plan. If the checker
accepts the declared program from $w_0$, then:
\begin{enumerate}
\item[(a)] the executed analysis controls $\mathrm{mFDR}_\eta$ at $\alpha$,
  with no conditioning on any honesty event; and
\item[(b)] on every sample path on which each realized $p_j$ respects its
  declared bound, every declared claim is also accepted by the runtime
  monitor (Theorem~\ref{thm:static}).
\end{enumerate}
\end{thm}

\begin{sloppypar}
\noindent The theorem is machine-checked (Section~\ref{sec:mechanization}),
and two ingredients carry it. First, the checker's verdict is proved to read
only the declared bounds, never the realized prices, so ``the checker accepts'' is a
deterministic hypothesis: trivially $\mathcal{F}_0$-measurable, carrying
no information about the data that could reweight any $p$-value. Second,
the wealth process Theorem~\ref{thm:mfdr} is proved about is identified,
pathwise, with the fold of the runtime monitor's step relation over the
program it executes, so parts (a) and
(b) are statements about a single process. The asymmetry in
the statement is deliberate and is the resolution of the composition
problem: honesty upgrades feasibility, never error control, so a dishonest
bound can cost the analyst a claim the certificate promised was affordable,
but cannot cost the reader the $\mathrm{mFDR}_\eta$ bound.
\end{sloppypar}

\subsection{The Mechanization}
\label{sec:mechanization}

The Lean development\footnote{Available at
\url{https://github.com/abuach/tacet-mech}.} (about 1{,}150 lines of
Lean~4.33.1 against a pinned Mathlib) covers the accounting layer of
Sections~\ref{sec:grade}--\ref{sec:metatheory}, the design premise, and
part of the estimation sublanguage's own soundness. It defines the two transformers and proves the three
grade lemmas of Appendix~\ref{app:proof}; it defines the runtime monitor
and the static checker as step relations over claim sequences, with
honesty, checker acceptance, and monitor acceptance each defined by
recursion over the program, exactly the definitions of
Appendix~\ref{app:semantics}, a mechanized claim being a realized price
and an optional declared bound with the mechanism left as the external
oracle that appendix already takes it to be; and it proves
Theorem~\ref{thm:static} by the induction of Appendix~\ref{app:proof},
generalized over the two starting wealths so that the invariant
$\hat W \le W$ can be carried and restored at each step, proved per claim
with whole-program acceptance as its special case. Theorem~\ref{thm:design} is proved with
$\mathit{Field}$, $\mathit{Val}$, and $\mathit{Obs}$ left abstract, since
the theorem is a statement of predicate logic about a schema, two
constraints, and a run, and needs no decidability or finiteness
assumption; $\mathit{Distinguishing}$ is an inductive proposition
mirroring Section~\ref{sec:design-premise}'s two closure rules, and in
Lean an inductive \texttt{Prop} \emph{is} the least fixed point of its
own constructors. Parts (a) and (b) are proved exactly as stated, over
Definition~\ref{def:conforms}'s two clauses.

Appendix~\ref{app:estimation}'s two lemmas are now partly mechanized. The
development defines the sample and estimation sorts, their denotation (Definition~\ref{def:denote}), and the footprint,
constraint, and purity-bit recursions of
Sections~\ref{sec:schemas}--\ref{sec:estimation} directly as structural
recursion over the syntax, since the estimation fragment is syntax-directed
(Appendix~\ref{app:semantics}'s ``Decidability of checking'' makes this
precise): defining $\varphi$, $\kappa$, and $t$ this way \emph{is}
mechanizing \est{\textsc{T-Observe}} through \est{\textsc{T-Arith}}, not a
stand-in for them. Clauses (a) and (b) of Lemma~\ref{lem:footprint} are
proved outright, by induction on the term, one
constructor at a time, exactly as the paper's proof does; clause (b) is the
$\kappa$-consistency Theorem~\ref{thm:design} assumes of its two samples,
so this closes the one appendix-level gap between the two theorems'
hypotheses that was previously taken on faith. Lemma~\ref{lem:clean} is
proved in full: a clean-typed term's denotation is
provably unchanged under \emph{any} replacement of the outcome assignment,
because a clean purity bit means \est{\textsc{T-SelectOutcome}} occurs
nowhere in the term, so the denotation's own recursion never inspects an
outcome to begin with (mirroring the paper's own proof: ``$t_1 \sqcup t_2 =
\clean$ forces $t_1 = t_2 = \clean$''). Clause (c) of Lemma~\ref{lem:footprint} remains unmechanized, a genuine two-run substitution
property (invariance of $\den{e}$ under changing the outcome of an
observation \emph{outside} $\varphi$, on two otherwise-identical runs)
that needs machinery beyond a single term's structural recursion, and the
plan machinery of Appendix~\ref{app:plan}: the mechanized checker consumes
the claim-sequence normal form rather than $\mathcal{T}$'s full grammar,
so the connection between $\mathcal{T}$'s typing rules and the mechanized
statement is itself a paper-level correspondence for the plan layer, even
though it is now a Lean-level one for the estimation layer's footprint and
purity bit. The 210-pair agreement of Section~\ref{sec:metatheory}
still matters alongside the proof for the same reason: a Lean proof of
Theorem~\ref{thm:design} is a proof about the mathematical statement, and
the 210 pairs check that the concrete
\lstinline[style=tacetcode]|static_design| the library ships actually
computes that statement's $\mathit{design}$ function.

The probabilistic layer behind Theorems~\ref{thm:mfdr}
and~\ref{thm:bridge} is measure-theoretic rather than a finite stand-in:
an arbitrary probability space, a filtration over $\mathbb{N}$, and
Mathlib's Bochner integral and conditional expectation. The wealth
process is defined by recursion through $T$ itself; its adaptedness and
integrability are proved rather than assumed; the $p$-values are
arbitrary random variables, each revealed at time $j+1$; and the set of
true nulls is a fixed set of indices, matching the fixed-horizon scope
Table~\ref{tab:assumptions} records. (H3) enters as an almost-everywhere
bound on the conditional expectation of the rejection indicator and is
consumed once, through the tower property; the supermartingale induction
carries the invariant of Theorem~\ref{thm:mfdr}'s proof; and a separate
lemma identifies the stochastic process pathwise with the fold of the
monitor's step relation, which entitles Theorem~\ref{thm:bridge} to speak
of one process. Every theorem cited in this section reports exactly
the axioms \texttt{propext},
\texttt{Classical.choice}, and \texttt{Quot.sound} under
\texttt{\#print axioms}, the standard classical trio, with no
\texttt{sorry} anywhere in the development, and \texttt{lake build}
reproduces the check.

The folding remains, and it is sound only relative to a further declaration,
described next.

\subsection{Non-Conservativity of Folding}
\label{sec:folding}

The statement one would like is that folding is conservative: it reduces the
sample from observations to artifacts, so it should only make evidence look
weaker, and the type system's restriction would then carry no soundness risk
at all. That statement is false, and the counterexample is small enough to
give in full.

Consider two groups, each of 25 artifacts carrying 4 observations apiece;
each group has exactly ten failed observations, and the groups differ only
in how those failures are distributed (Table~\ref{tab:folding}).

\begin{table}[t]
\centering
\small
\caption{The folding counterexample: two groups of 100 raw observations
(25 artifacts $\times$ 4 each), identical in raw failure count and
differing only in the distribution of those failures, under Fisher's exact test on the
raw and on the folded counts.}
\label{tab:folding}
\begin{tabular}{LNN}
\toprule
failures & \multicolumn{1}{R}{raw} & \multicolumn{1}{R}{folded} \\
\midrule
spread over 10 artifacts & 90/100 & 15/25 \\
concentrated in 3 & 90/100 & 22/25 \\
\midrule
\multicolumn{1}{R}{comparison} & $p = 1.0000$ & $p = 0.0507$ \\
\bottomrule
\end{tabular}
\end{table}

The groups are indistinguishable observation by observation.
However, folding manufactures a near-significant difference between them.
We swept 4{,}000 random configurations, each drawing two independent groups
of $n \in [8, 40]$ artifacts with $k \in [2, 5]$ observations apiece, group
pass probabilities drawn uniformly from $[0, 1]$, and each observation an
independent Bernoulli draw at its group's rate; raw and folded Fisher
$p$-values are then compared on the same configuration. Folding is
conservative (folded $p$ larger) 86.9\% of the time, anti-conservative
(folded $p$ smaller) 6.3\% of the time, and leaves $p$ unchanged, an exact
tie, the remaining 6.8\%. In the worst anti-conservative case swept, folding
moves a $p$-value from $0.044$ to $0.0016$, a factor of 27, and in 11 of
the 4{,}000 configurations it carries the $p$-value across the conventional
threshold outright, in the sharpest case from $0.31$ to $0.038$.

The reason is that folding is all-or-nothing, so a group whose failures are
spread loses more than one whose failures are concentrated: folding
\emph{changes the estimand}, from the fraction of observations that passed
to the fraction of artifacts that passed completely. These are different
quantities, and the second is the one the declaration named.

\begin{thm}[Folding is sound relative to its declaration]
\label{thm:folding}
If the raw observations are dependent only in the way
$\mathit{design}(\Sigma, \kappa_1, \kappa_2)$ declares (across the two
sides of the same artifact when paired, within a shared cluster value when
clustered, and mutually independent otherwise), then the folded per-artifact
outcomes are dependent only in that same pattern or less, and any mechanism
valid for it yields a $p$-value valid for the folded estimand.
\end{thm}

\noindent The proof is immediate: folding one artifact's raw observations,
on one side, into a single bit is a deterministic function of exactly those
observations, so it cannot create dependence between two folded bits whose
raw observations were disjoint and independent. It can certainly destroy
dependence a deterministic reduction is free to erase (folding two
correlated raw observations into one bit apiece may leave the two bits
less correlated than the observations were, or none at all), and that
direction is harmless for validity: a mechanism sound for the declared
pattern stays sound when the true dependence is weaker than declared, the
same over-approximation Section~\ref{sec:design-premise} already licenses
for the design premise itself.

Note that all of the content is in the antecedent, which the analyst
supplies. We consider this the right shape for the guarantee. It certifies
that the declared question was answered at the stated budget; it offers no
correction to a $p$-value computed at the wrong unit. Only the first was
ever available, because choosing the unit chooses the question, and only the
analyst can do that.

\section{Implementation}
\label{sec:implementation}

\sys{Tacet}'s core\footnote{Available at
\url{https://github.com/abuach/tacet-python}, \sys{Tacet}'s reference
implementation and the case study replication code.} is 1{,}994 non-blank, non-comment-only lines of
dependency-free Python (1{,}566 for the runtime monitor and 428 for the
static front end), importing nothing outside \lstinline|math|,
\lstinline|itertools|, \lstinline|dataclasses|, \lstinline|ast|, and
\lstinline|sys|. \lstinline|tacet.qude|, the five optional investment rules
outside the certified perimeter (Section~\ref{sec:grade}), adds 186 more and
imports nothing beyond the package itself. \kw{observe} folds observations
onto their artifact and records the schema; \kw{Tracked} propagates
footprints through ordinary arithmetic; key indexing and \kw{where} implement
\est{\textsc{T-SelectKey}}'s $v[k]$, and \kw{select} implements
\est{\textsc{T-SelectOutcome}}'s $v\langle q \rangle$, with one difference
of surface convenience: the calculus applies $q$ to the outcome alone,
while \kw{select}'s Python predicate receives the key too, since a
selection that reads the key is already priced as $\est{\rd}$ by the form
it entered through and gains nothing from also being denied the key's
value; a \kw{Mechanism}
carries the \kw{Design} it
assumes; and \kw{spend} prices a $p$-value from any test, so the budget is
indifferent to which instrument produced it, while the design premise is
not.

\textbf{Threat model.} The guarantee is against an honest-but-fallible
analyst, one who makes mistakes but is not
trying to defeat the accounting. The monitor is a library object living inside the analyst's own Python
process: a pool's wealth is a plain mutable float, \kw{spend} will price
any $p$-value it receives (above), and nothing stops an analyst from
assigning to a pool's wealth directly. The trusted computing base is
therefore all of \sys{Tacet}, the Python interpreter running it, and the
Lean development only insofar as it describes what the shipped code
does, a correspondence tested on 599 random programs across two corpora
(Section~\ref{sec:metatheory}), not proved. This is the same threat model
every language-embedded accounting system in Section~\ref{sec:related}
adopts, and we state it explicitly because the
rest of this paper argues from proofs, and a proof about a model leaves
open what the process running it actually does.

\textbf{Two declarations.} In the worked analysis, the artifact key and the
cluster level are the only declarations; the third declared input of
Section~\ref{sec:declarations}, the bound, is optional and arrives only
with a plan. Four derived statistics and three
selections follow with nothing further stated, because footprints are
relative to declared artifacts rather than to variables, and propagation is
therefore a union over operands. 

\textbf{Value-level footprint tracking.} A \kw{Tracked} value prints its
footprint, so a few lines of ordinary arithmetic show what propagation
does:

\begin{whole}\begin{lstlisting}[style=tacetcode]
qwen.rate, llama.rate, qwen.rate - llama.rate
\end{lstlisting}\end{whole}
\begin{quote}\ttfamily\small
Tracked(1, 15 artifacts)\\
Tracked(0.8667, 15 artifacts)\\
Tracked(0.1333, 30 artifacts)
\end{quote}

\noindent Two disjoint models, fifteen artifacts each; the difference unions
them to thirty, with nothing declared to make that happen. A footprint is a
set and union is idempotent, so an artifact that enters an expression twice
is counted once, and no chain of operations grows a footprint past the
artifacts the analysis declared. A numeric grade behaves differently: a
sensitivity bound multiplies along the way, and squaring a tracked value is
enough to send it to infinity~\cite{dduo}.

The interesting pair narrows the same board in two ways:

\begin{whole}\begin{lstlisting}[style=tacetcode]
zoo["Stable-Code", "fim"].rate
\end{lstlisting}\end{whole}
\begin{quote}\ttfamily\small
Tracked(1, 6 artifacts)
\end{quote}
\begin{whole}\begin{lstlisting}[style=tacetcode]
stable.select(lambda k, ok: not ok).rate
\end{lstlisting}\end{whole}
\begin{quote}\ttfamily\small
Tracked(0, 15 artifacts)
\end{quote}

\noindent The first narrows to six artifacts, because a key literal is
decided before any outcome is read. The second keeps eight of the fifteen
and reports a rate of zero, because every artifact it kept is one that
failed. Its footprint is still fifteen: it read all fifteen to decide which
eight to keep, and \textsc{T-SelectOutcome} charges for the reading rather
than for the survivors. The two calls look like the same kind of narrowing,
and they leave different footprints, different constraints, and different
purity bits.

\textbf{Escapes and untracked control flow.} Calling \kw{float} on a tracked
value returns a bare number, and the footprint is gone. We count that one
escape route rather than blocking it (a class-level counter increments on
every such call, inspectable on request), consistent with the
honest-but-fallible threat model stated above; we do not count or report every route a
value can leave tracking by (a direct \lstinline|.value| read, for one), and
nothing surfaces the counter automatically in \kw{summary}. Data-dependent
control flow is likewise invisible: a claim
suppressed because the number looked wrong leaves no trace, which is the
garden of forking paths~\cite{gelman-loken} in the one place our footprints
cannot reach, and it is the implicit-flow problem under its own
name~\cite{sabelfeld-myers}.

\textbf{Pre-registration as a scoping rule.} \kw{observe} is the scope
boundary: a plan may be opened only before it and is sealed by it, so a
declared claim cannot have been chosen by looking, and the property is
syntactic rather than a fact about the analyst's intentions. This is the
mechanism blind analysis names in the sciences that already
practice it~\cite{maccoun-perlmutter}: hide the answer until the plan is
fixed, so the choice of what to test cannot be shaped by what was seen. The
seal does the enforcing: a registration is a label, an optional bound, a
sidedness flag, and two tuples of key literals, nothing that can hold a
reference to observed data, and declaring one on a plan \kw{observe} has
already sealed raises immediately, so no later registration can be added no
matter what it was computed from. A registration names the
comparison and not merely the claim, which it can do because a key literal
is not data; a subset chosen by reading outcomes pins nothing and therefore
can never satisfy one. Declared claims must be made in declared order, and
each at most once, since
Section~\ref{sec:grade}'s transformers do not commute and the audit charges
each entry a single stake: a reordered or repeated run is a run the audit
never folded. Skipping is permitted and reported, which is
the answer to the omission attack of Section~\ref{sec:limitations}: the deletion stays invisible
to every footprint, but the declaration is on the record.

The plan licenses two prices the monitor cannot. One-sided tests become
available where the direction was fixed in advance, and declared bounds make
Section~\ref{sec:metatheory}'s audit computable with no data in the process
at all: affordability stopped depending on footprints when the budget
collapsed to one pool, so the whole judgment is a fold over the plan. The
checker of Section~\ref{sec:metatheory} is that audit, so
Theorem~\ref{thm:static} is tested against the shipped feature rather than
against a second implementation of it.

A second judgment names what the audit decides:

\begin{mathpar}
\inferrule[Ok]{\clc{\pi \vdash c : \Claim{w}} \\ \ty{\mathit{affordable}(\pi, w_0)}}
  {\ty{\pi \vdash_\alpha c \;\mathsf{ok}}}
\end{mathpar}

\noindent $\pi \vdash c : \Claim{w}$ is Appendix~\ref{app:plan}'s extended
judgment, and $\pi$ is the declared plan (Section~\ref{sec:declarations}).
$\mathit{affordable}$ walks \emph{the plan} (not the claim sequence) in
declared order from the initial wealth $w_0 = \alpha/2$, requiring each
declared bound to be covered by the stake available when it is reached and
charging the stake whether or not a matching leaf appears in $c$:
Appendix~\ref{app:semantics}'s static semantics, read as a single top-level
verdict over $\pi$, exactly as \lstinline|Plan.audit()| computes it, from
the declared bounds alone and before a single leaf of $c$ has been typed,
which lets the audit skip straight to $\pi$ and never touch
\kw{observe}; the identification carries one scope condition, stated as
the fourth caveat below. \ty{\textsc{Ok}} holds exactly when every step at
which \clc{\textsc{S-Claim}} fires is certified, which is also why a
declared label with no matching leaf in $c$ is still charged: order and
omission are properties of $\pi$, not of any derivation of $c$'s type
(Appendix~\ref{app:plan}).

\textbf{The static audit.} The audit requires no execution at all, not even
up to \kw{observe}. \lstinline|python3 -m tacet.static <file.py>| reads a
file's source without running it and answers \ty{\textsc{Ok}} on whatever it
can statically resolve. It is not a second implementation: it walks the AST
for one pattern (a
\lstinline|Study|/\lstinline[style=tacetcode]|preregister|/\lstinline[style=tacetcode]|claim|
sequence built from literal arguments in straight-line code), reconstructs
the \kw{Plan} that pattern would build, and calls
\lstinline[style=tacetcode]|Plan.audit()| on it, the same function
Theorem~\ref{thm:static} was tested against above. Every argument must be a
literal, since a plan entry is never data; a claim built inside a loop, a
branch, or from anything else is reported as excluded rather than guessed
at. Run on the first declared plan of \lstinline|prereg.py|, a six-plan
example file, before a single row of data has been read:

\begin{whole}\begin{lstlisting}[basicstyle=\ttfamily\scriptsize]
$ python3 -m tacet.static prereg.py
scope 'part1', plan 'plan' (line 43, study 'study')
===================================================
static audit of 3 declared claims, opening wealth 0.02500
claim                           bound     stake     verdict    wealth
---------------------------------------------------------------------
Qwen beats CodeQwen            0.0100   0.01250  affordable   0.03750
Qwen beats Stable-Code         0.0100   0.01875  affordable   0.04375
Qwen beats OpenCoder           0.0400   0.02188     REFUSED   0.02188

2 of 3 decided claims are affordable
\end{lstlisting}\end{whole}

\noindent Two of three claims are certified and the third refused, all
before \kw{observe} runs; the file has five more plans behind this one, each
checked and reported the same way, under its own scope header naming the
source line and the \kw{Study} variable it audits. The transcript's
``affordable'' is the static verdict of Section~\ref{sec:analysis}'s
terminology footnote, computed with no data present, and distinct from
both solvency and support. The
process exits
nonzero on any refused or out-of-scope claim, so a pre-commit hook or a CI
step can gate on it; an excluded claim is a separate case, with its own
exit code (the third caveat below).

\lstinline|prereg.py| is a toy; the same audit runs on the paper's
own running example. Declaring the leaderboard's first six adjacent-rank
claims at slack $1.0$ (Table~\ref{tab:validation}'s convention: the declared
bound is the exact cluster-level $p$-value computed for that pair, the
sharpest honest bound possible) and auditing before \kw{observe} runs:

\begin{whole}\begin{lstlisting}[basicstyle=\ttfamily\scriptsize]
$ python3 -m tacet.static leaderboard_prereg.py
static audit of 6 declared claims, opening wealth 0.02500
claim                           bound     stake     verdict    wealth
---------------------------------------------------------------------
rank 1 beats rank 2            1.0000   0.01250     REFUSED   0.01250
rank 2 beats rank 3            0.8906   0.00625     REFUSED   0.00625
rank 3 beats rank 4            0.5000   0.00313     REFUSED   0.00313
rank 4 beats rank 5            1.0000   0.00156     REFUSED   0.00156
rank 5 beats rank 6            0.7539   0.00078     REFUSED   0.00078
rank 6 beats rank 7            1.0000   0.00039     REFUSED   0.00039

0 of 6 decided claims are affordable
\end{lstlisting}\end{whole}

\noindent All six are refused with no row of data read, exactly predicting
Section~\ref{sec:leaderboard}'s runtime finding on the same six comparisons.

Declaring the schema adds the design premise to the same audit.
\lstinline|leaderboard_schema.py| opens the plan with the
leaderboard's $\Sigma$ (the key's three fields, $I = \{\mathit{instance}\}$,
$\mathit{instance} \mapsto \mathit{repo}$, and $\gamma = \mathit{repo}$)
and re-declares the six claims above with their submissions as pins, so
\lstinline[style=tacetcode]|Plan.audit| decides \clc{\textsc{T-Claim}}'s
design premise from the schema and each entry's own pins, still before any
data is read. With the affordability columns of the transcript above
elided, the audit adds, for every one of the six claims:

\begin{whole}\begin{lstlisting}[basicstyle=\ttfamily\scriptsize]
$ python3 leaderboard_schema.py
static audit of 6 declared claims, opening wealth 0.02500
claim                        design             mechanism
---------------------------------------------------------
rank 1 beats rank 2          paired+clustered   cluster sign-flip
  ...
rank 6 beats rank 7          paired+clustered   cluster sign-flip

value-level route agrees with the declared schema on
6 of 6 comparisons
\end{lstlisting}\end{whole}

\noindent The last line is the same script's runtime half: once the 134
submissions' resolved sets are loaded, \kw{required\_design}, the
value-level route every other example uses, is asked the same six
questions on the realized keys and returns the same design and the same
demanded mechanism for each, the agreement Theorem~\ref{thm:design}
requires. The schema's pins are the submission names and its bounds are
the transcript above's, so nothing in the static half reads an outcome;
the design premise, like affordability, is decided from declarations
alone.

Four limits apply to the audit. First, sealing is per \kw{Study}
instance, so nothing stops a second, fresh \kw{Study} later in the same
file from re-observing the same rows under a plan written after the first
\kw{Study}'s results were already read; that case still requires a
timestamped external record, and no language feature replaces one
(Section~\ref{sec:threats}'s account-boundary threat). Second, a violated
bound voids only the static verdict: the claim is priced at its real
$p$-value first, and only then does the run either continue with the
claim marked or abort (\texttt{on\_violation}); neither setting can skip
the charge, since an uncharged refusal would let an analyst read a price
without paying for it. Theorem~\ref{thm:bridge} leans on this accounting: a dishonest bound
can cost the analyst a certified claim, never the reader the
$\mathrm{mFDR}_\eta$ bound, though a single dishonest bound can still
strand every \emph{honest} certified claim behind it, since the shadow
wealth is a running total. Third, a claim the walker cannot resolve to a
literal is excluded, and an excluded claim is distinct from a refused
one: a file with at least one declared plan and nothing resolvable exits
a third code, separate from both a clean pass and a refusal, so a CI step
reading only the exit code cannot mistake an uncheckable file for a
passing one. A loop over a literal collection written in the source is
unrolled and audited exactly as if written out by hand (demonstrated on
this section's own six-claim excerpt, with identical output); a
collection derived from the run itself (the ranked board, the released
task files) does not exist before \kw{observe} runs, so the leaderboard's
8{,}911-claim family and BIG-Bench Hard's 26-claim family remain
excluded, and for them the static fragment and the runtime monitor that
priced every claim in Section~\ref{sec:evaluation} remain two different
tools. Fourth, the certificate speaks about runs whose priced claims are
the plan's own, in declared order, with skips permitted. A skipped label
only leaves the run richer than the shadow, since the fold charged a
stake the run never paid; an \emph{exploratory} claim, one with no
registered label, priced while declared entries remain unclaimed drains
the pool by a factor the fold never charged, so a certified claim behind
it can be refused at run time with every declared bound honest. The
checker of Appendix~\ref{app:semantics} covers the mixed program, because
it walks the full claim sequence and charges the exploratory claim
$T_\bot$ (\clc{\textsc{S-Unbound}}); the certificate
\lstinline|Plan.audit| computes from $\pi$ alone holds for a run exactly
when no unplanned claim is priced before the plan's last certified one.

\section{Evaluation}
\label{sec:evaluation}

The four questions posed in Section~\ref{sec:intro} are answered by
construction, in the type system and its metatheory
(Sections~\ref{sec:lambdat}--\ref{sec:metatheory}). Two empirical questions
remain. Does the checker refuse the mistakes it was built to catch, on
programs constructed to make each mistake (Section~\ref{sec:gallery})? And
on real, honest analyses that made no such mistake, does the budget survive
contact with real effect sizes, keeping most of a paper's claims when the
evidence is strong and refusing most of them when it is marginal, instead
of doing one or the other regardless of the data
(Sections~\ref{sec:leaderboard}--\ref{sec:comparing-case-studies})?

One boundary applies to every count below and is stated once here rather
than per table: (H3) is unverified for both corpora
(Table~\ref{tab:assumptions}, Section~\ref{sec:mfdr}), and both claim
families are reconstructions made with the published outcomes already in
hand, so each case study prices a family chosen retrospectively rather than
declared in advance. The case studies demonstrate that the inputs the
procedure needs are recoverable from a typed analysis and that the resulting
prices are informative; they are not certified runs of the guarantee.
Section~\ref{sec:certified} exhibits one, on synthetic data built so that
(H1)--(H4) hold by construction.

\subsection{Refusal Case Studies}
\label{sec:gallery}

Section~\ref{sec:analysis} already showed two of the refusals: a claim the
pool cannot cover (\lstinline[style=tacetcode]|Unaffordable|), and a
two-sided test aimed at data that is paired and clustered
(\lstinline[style=tacetcode]|DesignMismatch|). Five more follow, each a
minimal program with one defect, run against the real library. Whether
the budget also prices real, honest analyses correctly is the harder
question, and
Sections~\ref{sec:leaderboard}--\ref{sec:comparing-case-studies} take it
up.

Every program below shares one setup: two systems, sixteen tasks each,
folded onto a \lstinline|Key(system, task)| artifact, with system A stronger
than system B by design.

\begin{whole}\begin{lstlisting}[style=tacetcode]
study = Study(alpha=0.05)
board = study.observe(rows, artifact=lambda r: r["key"],
                      passed=lambda r: r["passed"])
\end{lstlisting}\end{whole}

\textbf{An opaque artifact key.} Declaring a cluster level the key does not
carry is neither silently dropped nor silently assumed independent. This is
not a point on the design lattice at all (the lattice's bottom is
\emph{independent}, a specific design that receives Fisher's exact test
(Section~\ref{sec:lambdat}), not a default any mechanism can claim); it is
the case where no $\Sigma$ was ever formed: Definition~\ref{def:schema}'s
$\gamma \in F \cup \{\cdot\}$ is part of what makes a tuple a schema at all,
so a cluster field absent from the key does not yield a schema with an
undefined design, it yields no schema. \kw{observe} presupposes a
well-formed $\Sigma$; nothing in Python enforces that at construction time,
so this is the check that does, and Section~\ref{sec:threats}
names refusing over guessing as the policy for exactly this case:

\begin{whole}\begin{lstlisting}[style=tacetcode]
study.observe(
    rows, artifact=lambda r: r["key"],
    cluster="repo", passed=lambda r: r["passed"])
\end{lstlisting}\end{whole}
\begin{whole}\begin{lstlisting}
OpaqueKey: cluster level 'repo' was
declared, but the artifact key has fields
('system', 'task') and none of them is it.
  add 'repo' to the key, or drop the
  cluster declaration.
\end{lstlisting}\end{whole}

\textbf{An unplanned one-sided test.} The cheaper one-sided price is sound
only if the direction was fixed before the data was seen, and nothing in
this program says it was:

\begin{whole}\begin{lstlisting}[style=tacetcode]
study.claim("A beats B",
            board.where(system="A") > board.where(system="B"))
\end{lstlisting}\end{whole}
\begin{whole}\begin{lstlisting}
OutOfScope: 'A beats B' is one-sided, which is the cheaper
price, and it is sound only if the direction was fixed
before the data was seen. nothing here says it was.
\end{lstlisting}\end{whole}

\textbf{A late plan.} \kw{observe} is the scope boundary (Section~\ref{sec:implementation}); a plan
can only be opened before it:

\begin{whole}\begin{lstlisting}[style=tacetcode]
board = study.observe(
    rows, artifact=lambda r: r["key"],
    passed=lambda r: r["passed"])
study.preregister()
\end{lstlisting}\end{whole}
\begin{whole}\begin{lstlisting}
OutOfScope: cannot pre-register after
observe: the data is already in hand, so
a claim declared now could have been
chosen by looking at it.
\end{lstlisting}\end{whole}

\textbf{A reordered claim.} Section~\ref{sec:grade}'s transformers do not
commute, so a plan fixes an order and not merely a set:

\begin{whole}\begin{lstlisting}[style=tacetcode]
plan.claim("first", bound=0.01)
plan.claim("second", bound=0.01)
...
study.claim("second", a.vs(b))   # lands
study.claim("first", a.vs(b))    # declared earlier, claimed later
\end{lstlisting}\end{whole}
\begin{whole}\begin{lstlisting}
OutOfScope: 'first' is declared before 'second' in the plan
but is being claimed after it.
\end{lstlisting}\end{whole}

\textbf{An outcome-selected sample against a registration.} A registration
pins the whole comparison down to its constraints, and a sample narrowed by
outcome carries no pin to match it:

\begin{whole}\begin{lstlisting}[style=tacetcode]
plan.claim("A beats B", bound=0.01, left="A", right="B")
...
cherry = board.where(system="B").select(lambda k, ok: not ok)
study.claim("A beats B", board.where(system="A").vs(cherry))
\end{lstlisting}\end{whole}
\begin{whole}\begin{lstlisting}
OutOfScope: 'A beats B' declared its right side as ['B'], but
the sample given is selected by outcome, so it carries no
pinned key.
\end{lstlisting}\end{whole}

This message also marks the one deliberate divergence between the library
and $\mathcal{T}$. The library reads the design off observed key values
(Section~\ref{sec:schemas}), so it can afford to drop the pin at
\kw{select} and refuse on the empty pin above; $\mathcal{T}$ computes the
design from constraints, so \est{\textsc{T-SelectOutcome}} must leave
$\kappa$ intact, and the $\est{\clean}$ premise of
\clc{\textsc{T-OneSided}} and \clc{\textsc{T-Registered}}
(Appendix~\ref{app:plan}) carries the refusal instead. The two enforce the
same judgment at different doors, and not merely by construction:
Appendix~\ref{app:plan}'s ``Why dropping the pin cannot mismatch'' shows
the library's simplification can only ever coincide with, never fall short
of, $\mathcal{T}$'s refusal, because Section~\ref{sec:design-premise}'s own
non-degeneracy conditions already make an empty-pinned registration
permanently unmatchable by anything, honest samples included.

Each of these five is the direct, predictable consequence of a rule already
stated earlier in the paper (the design lattice's bottom element, the
one-sided price, the scope boundary, non-commutative composition, and the
pin a registration checks), exercised on the smallest program that trips
it. The guarantee is intended to operate at this level, one specific way of
getting it wrong at a time, well below a verdict on whether a whole
leaderboard turned out fine. Each program above also has a
one-line accepting twin: adding \lstinline|cluster="repo"| to the key,
registering the direction before \kw{observe} runs (or asking for the
honest two-sided price instead), calling \kw{preregister} before
\kw{observe} rather than after, claiming in the order the plan declared, or
comparing against the unfiltered sample instead of the outcome-selected one
is enough to turn every refusal here into an ordinary, priced claim: this
suite tests only that the checker refuses what it should, not how often it
wrongly refuses what it should not; Section~\ref{sec:threats} states that
limitation directly.

\subsection{Case Study: the SWE-bench Verified Leaderboard}
\label{sec:leaderboard}

The leaderboard is the easy case, because a table presents its claim
family directly. We took every public submission to
SWE-bench Verified~\cite{swebench} in the benchmark's own experiments
repository at the time of writing, 134 in total, each with the set of
instances it resolved, and priced every comparison the table's total order
implies. Treating the table as licensing all of them is the consistent
reading, since a reader who accepts that rank 3 beats rank 4 and rank 4
beats rank 5 has already accepted that rank 3 beats rank 5, by the
transitivity the order itself asserts: 8{,}911 claims over the 468 instances
resolved by at least one submission, the 36 pairs whose resolved counts tie
included, since the table prints one name above the other and ``no more than
chance'' is itself the finding for such a pair.\footnote{The narrower reading agrees: of the 133
adjacent-rank claims the full family contains, naive $p < 0.05$ supports 2
and the budget keeps 0, and $0.05/133$ is already below the sharpest
$p$-value this data can produce, so no adjacent claim rejects under
Bonferroni at the smaller family either.} The artifact
(Section~\ref{sec:declarations}) is the patch one submission produced for
one instance, keyed by its (submission, repo, instance) triple.

Neither of this corpus's two findings is novel on its own. That benchmark
comparisons in this literature are routinely underpowered is established
directly~\cite{card-power,dror-replicability}, and that a sequential budget's
power depends on the order its claims arrive in is a known property of
alpha-investing~\cite{foster-stine,aharoni-rosset}, one QUDE~\cite{qude}
already discusses as order-gaming. \sys{Tacet} recovers both inputs from
the program instead of assuming them: the unit of evidence comes
from the declared nesting, and the order comes from the claim sequence the
analysis itself executes.

Both matter here. Respecting the declared nesting, 12 repositories admit
only $2^{12}$ sign assignments, so the sharpest two-sided $p$-value any
comparison on this data can produce is $2/4{,}096 = 0.000488$, above
Bonferroni's threshold for 8{,}911 claims and above the largest stake the
budget can offer after five refusals; the family is unanswerable at this
resolution, and reading order goes bankrupt at the sixth claim. Run instead
at the instance level, where each of the 468 instances counts as its own
unit, the same analysis moves individual adjacent-rank $p$-values by up to
$48\times$ and supports 5{,}518 claims. Nothing in the $p$-values says which
of the two runs read the design correctly; the artifact keys do.
Tables~\ref{tab:ordering} and~\ref{tab:baselines} report both families in
full, and Section~\ref{sec:comparing-case-studies} reads them against the
second case study.

\subsection{Case Study: BIG-Bench Hard}
\label{sec:bbh}

The harder and more representative case is a paper that asserts its
comparisons in sentences rather than in a table. BIG-Bench Hard~\cite{bbh} is one:
its abstract argues that prompting without chain-of-thought substantially
underestimates model capability, and its authors released the per-example
outputs for \texttt{code-davinci-002} under both prompting conditions, which
makes the claim priceable at all.

The family we reconstruct is one claim per task: does chain-of-thought beat
direct prompting on this task? The design follows from the keys without our
help. Both conditions run on the same examples, so each comparison is
paired; and examples nest within tasks, so any pooled claim is a question
about tasks. The design demands McNemar's exact test on the task's
discordant examples, priced two-sided at $\alpha = 0.05$ under
\lstinline|fixed_fraction|'s defaults (Section~\ref{sec:grade}), the same
choice as every other numeric example in this paper: BIG-Bench Hard's own
abstract already states a direction, so pricing these claims one-sided
would spend the cheaper price on a direction this case study did not fix
before looking, precisely what Section~\ref{sec:gallery}'s unplanned
one-sided example refuses.

Our scoring of the released outputs matches the paper's own results table
on every point we can check: 22 of 22 comparison points agree to one
decimal place, the precision the original table reports, and everything
below rests on this. A point is not always one task, since BBH's table
averages some related tasks into a single entry; unpacking each average
back into its tasks expands the 22 checkable points to 26 individual task
files, one claim per file. This covers 26 of the 27 released tasks; the
27th, \texttt{word\_sorting}, is not one of the 22 verifiable points, so we
exclude it by the same rule. Nothing turns on the choice: including it
leaves 19 of 27 supported, under the ordering below.

Across 26 tasks and 6{,}261 examples, taken alphabetically by task name and
so predictable in Section~\ref{sec:analysis}'s sense, 19 of 26 claims are
supported on a budget.
Chain-of-thought scores higher on 21 tasks and 21 tasks clear $p < 0.05$,
both independently of ordering, but not the same 21: on
\lstinline[style=tacetcode]|causal_judgement| and
\lstinline[style=tacetcode]|ruin_names|, direct prompting outscores
chain-of-thought on the raw counts while the two-sided McNemar test still
clears $0.05$ ($p = 0.0474$, $p = 0.0498$), since a two-sided test rejects
on a large discordance in either direction and reports nothing about
which. \kw{spend} never checks that a rejected claim's observed sign
agrees with its label (Section~\ref{sec:threats}); neither task is among
the 19 the budget supports here, so the headline is unaffected, but the
gap is real and this is where it surfaces. Sorted instead by each task's
discordant-pair effect size (an oracle ordering, read off the data under
test rather than one the guarantee covers), 17 are supported, tying
Table~\ref{tab:ordering}'s random-draw median by coincidence rather than
by the same computation.

This analysis mostly survives, which is the substance of the result: a
corpus on which every case study returned zero would measure the
instrument's own severity and nothing about the analyses. A paper with real, large, and plentiful
effects keeps most of its claims and pays two for the accounting: the
two tasks above $0.025$, the largest stake the
ceiling wealth offers. This is not a corpus coincidence: 19 is exactly the
count clearing $0.025$, the ceiling stake Section~\ref{sec:grade}'s fixed
point fixes independently of $i$, so any discovery-rich, predictably-ordered
family that reaches the ceiling pays the same tax, provided it reaches the
ceiling at all, since a family whose early claims are weaker can still
bankrupt first (Section~\ref{sec:comparing-case-studies}'s actual divide on
the leaderboard). Whether a plain correction would have kept as much is
answered directly there too.

Clustering has its largest effect on the pooled claim. Treating the 6{,}261
examples as independent trials gives $p = 1.76 \times 10^{-192}$. Respecting
the task structure gives $p = 8.37 \times 10^{-6}$. Both are significant, so
the paper's conclusion is not in doubt, and we are not disputing it. The
size of the evidence changes, by 187 orders of magnitude: a $p$-value of
$10^{-192}$ is not something an experiment with 26 tasks can produce. In
other words, the evidence here amounts to 26 units, one per task, however
many examples underlie each one.

\subsection{Comparing the Two Case Studies}
\label{sec:comparing-case-studies}

\begin{table}[t]
\centering
\small
\caption{Claims supported under five orderings of the same claim family, on
each case study. Random is the median [min--max] over 500 independent
draws (seed 0); ``best first'' and ``worst first'' sort by
the realized $p$-values themselves.}
\label{tab:ordering}
\begin{tabular}{LNN}
\toprule
ordering & \multicolumn{1}{R}{leaderboard} & \multicolumn{1}{R}{BIG-Bench Hard} \\
\midrule
best first (oracle bound, not a valid ordering) & 6{,}545 of 8{,}911 & 19 of 26 \\
chronological / alphabetical (predictable) & 0 of 8{,}911 & 19 of 26 \\
random, median [min--max] & 28 [0--302] & 17 [14--19] \\
as a reader meets them (predictable) & 0 of 8{,}911 & 19 of 26 \\
worst first (oracle bound, not a valid ordering) & 0 of 8{,}911 & 14 of 26 \\
\bottomrule
\end{tabular}
\end{table}

\noindent ``Best first'' and ``worst first'' sort the claims by the very
$p$-values they go on to spend, so neither is an ordering \sys{Tacet}'s
guarantee covers; they are oracle bounds on what the most and least
favorable ordering could achieve. The rows the guarantee covers are the
predictable ones: the leaderboard's submission-date order, BBH's alphabetical order (also its reading
order, since the released task files carry no submission date or other
natural sequence, so the right column's two predictable rows are the
same run), and each individual random draw (interquartile ranges 9--57
and 16--18). Every count in Table~\ref{tab:ordering}, like those above, is measured by running
\lstinline[style=tacetcode]|Study.spend| over an ordering and counting
the calls that return: \emph{supported} claims, distinct from the
static \ty{\textsc{Ok}} verdict of Section~\ref{sec:metatheory}.
Nothing in the checker or the monitor inspects how an ordering was
constructed, a gap Section~\ref{sec:threats} returns to.

On these two corpora, order dominates when effects are marginal and is
nearly irrelevant when they are strong: plentiful discoveries keep
refunding the budget whatever sequence they arrive in, whereas an
analysis scraping for adjacent-rank distinctions is at the mercy of
early failures draining the pool, and we would have overclaimed from
either study alone. The leaderboard's gap between reading order (0) and
a random draw (median 28) is a threshold effect. Under
\lstinline|fixed_fraction| with no rejections, wealth decays
geometrically, and once the largest offerable stake falls below the
sharpest $p$-value the data's own discreteness can produce ($0.000488$
at the cluster level, Section~\ref{sec:leaderboard}), no claim can ever be
kept again. Reading order meets its weakest, hardest-to-separate
evidence first and crosses that threshold by the sixth claim; a random
draw sometimes meets a few strong, well-separated pairs early, and
those few refunds are worth an order of magnitude more support, while
the random maximum (302) still falls three orders of magnitude short of
the strongest-first ceiling (6{,}545). Since every ordering fixed
independently of the evidence is covered by Theorem~\ref{thm:mfdr}
equally, a study whose claim order does not otherwise matter to the
analyst has a meaningful amount to gain by shuffling the family with a
seed fixed before any claim is priced; the library ships this as
\lstinline[style=tacetcode]|shuffled_order|. It is not the study
default, because the reading-order and chronological orderings are the
ones an analyst would write without thinking about ordering at all, and
reporting them shows what an unconsidered sequence costs; at
no seed we tried does shuffling close the gap to
Benjamini--Yekutieli's 2{,}038 (Table~\ref{tab:baselines}).

\textbf{Baselines.} Table~\ref{tab:baselines} places the budget beside
five fixed multiple-comparison corrections, on both claim families,
with the leaderboard priced at both the instance level and the cluster
level. Two readings matter. First, no procedure in the table catches
the design mistake on its own: Bonferroni, the most conservative of the
five fixed corrections, still endorses 6{,}467 of 8{,}911 comparisons
once every SWE-bench instance is treated as an independent unit, so the
multiplicity correction and the design correction answer different
questions, and a corpus can fail the second while every procedure
passes the first. Second, on the leaderboard's own dependent family,
Benjamini--Yekutieli, the honest offline comparator
Section~\ref{sec:background} names, keeps 2{,}038 of 8{,}911 while the
budget matches Holm and Bonferroni at 0 under both orderings we ran;
and on BIG-Bench Hard the budget's 19 of 26 is identical to what plain
Benjamini--Hochberg keeps without any of \sys{Tacet}'s machinery, a
count even the strongest ordering available to the budget does not
exceed. The budget adds its accounting
(Section~\ref{sec:limitations}) and a design
read automatically off the keys rather than assumed by whichever
procedure the analyst reaches for; neither property makes it the
stronger procedure on these corpora. The five QUDE investment rules the
implementation ships (\lstinline|tacet.qude|) behave consistently with
this, at QUDE's own example parameters: on the leaderboard in reading
order every rule keeps 0 of 8{,}911, and on BIG-Bench Hard the thrifty
$\beta$-farsighted rule keeps 19 of 26, tying
\lstinline|fixed_fraction|, while the three rules anchored to the
initial wealth keep 15, tying Bonferroni. These counts are measurements
of unproved rules: Theorem~\ref{thm:mfdr} covers
\lstinline|fixed_fraction| alone (Section~\ref{sec:threats}).

\begin{table}[t]
\centering
\small
\caption{The same claim families under five fixed multiple-comparison
corrections and \sys{Tacet}'s budget, at two orderings each fixed
independently of the claims' own $p$-values. Holm is Bonferroni's
step-down, uniformly-more-powerful cousin, at the same family-wise
guarantee; Benjamini--Yekutieli is Benjamini--Hochberg's extension to
arbitrary dependence among the tests, the setting the cluster-level
leaderboard is actually in (Section~\ref{sec:background}). The leaderboard
is shown at both the instance level, where McNemar's exact test treats
each of the 468 instances as an independent unit of evidence, and the
cluster level, where the sign-flip treats each of the 12 repositories as
one unit, the correction Section~\ref{sec:analysis} argues for.}
\label{tab:baselines}
\begin{tabular}{LNNN}
\toprule
procedure & \multicolumn{1}{R}{leaderboard (instance)} & \multicolumn{1}{R}{leaderboard (cluster)} & \multicolumn{1}{R}{BIG-Bench Hard} \\
\midrule
naive $p < 0.05$ & 7{,}835 of 8{,}911 & 6{,}987 of 8{,}911 & 21 of 26 \\
Benjamini--Hochberg & 7{,}808 of 8{,}911 & 6{,}813 of 8{,}911 & 19 of 26 \\
Benjamini--Yekutieli & 7{,}361 of 8{,}911 & 2{,}038 of 8{,}911 & 15 of 26 \\
Holm & 6{,}616 of 8{,}911 & 0 of 8{,}911 & 15 of 26 \\
Bonferroni & 6{,}467 of 8{,}911 & 0 of 8{,}911 & 15 of 26 \\
budget, predictable ordering & 7{,}611 of 8{,}911 & 0 of 8{,}911 & 19 of 26 \\
budget, reading order & 5{,}518 of 8{,}911 & 0 of 8{,}911 & 19 of 26 \\
\bottomrule
\end{tabular}
\end{table}

\textbf{Synthetic validation.} A real corpus's ground truth is unknown, so
Table~\ref{tab:baselines} cannot show which kept claims are false
discoveries; a simulation can. We generate synthetic leaderboards shaped
like Section~\ref{sec:leaderboard}'s (six systems, twelve clusters, largest
carrying 47\%, close to Django's 46\%), with pass probability set by a
system's skill, a per-cluster \emph{affinity}, and independent noise, so
two equal-skill systems form a genuine null that can still look, within
one cluster, as though one beats the other.

\begin{table}[t]
\centering
\small
\caption{Empirical $\mathrm{mFDR}_\eta$ under a complete null (2{,}000
synthetic leaderboards, 15 pairwise comparisons each, sharing the same
generative model), at the instance and cluster resolutions, for the same
four procedures as Table~\ref{tab:baselines}, beside the realized false
discovery rate and family-wise error rate the same trials measure.
Under a complete null every rejection is a false one, so a trial with any
rejection at all contributes a false-discovery share of 1 and one with none
contributes 0; FDR (the trial-averaged share, $\widehat{\mathbb{E}}[V/\max(R,1)]$) and
FWER (the fraction of trials with $V \ge 1$) therefore coincide exactly and are
reported as one column. Bold exceeds $\alpha = 0.05$.}
\label{tab:mfdr-sim}
\begin{tabular}{LNNNN}
\toprule
 & \multicolumn{2}{C}{instance-level} & \multicolumn{2}{C}{cluster-level} \\
\cmidrule(lr){2-3}\cmidrule(lr){4-5}
procedure & \multicolumn{1}{R}{$\mathrm{mFDR}_\eta$} & \multicolumn{1}{R}{FDR\,=\,FWER} & \multicolumn{1}{R}{$\mathrm{mFDR}_\eta$} & \multicolumn{1}{R}{FDR\,=\,FWER} \\
\midrule
naive $p < 0.05$ & \textbf{0.8280} & \textbf{0.9095} & \textbf{0.3601} & \textbf{0.3285} \\
Benjamini--Hochberg & \textbf{0.7592} & \textbf{0.6825} & 0.0251 & 0.0180 \\
Bonferroni & \textbf{0.6713} & \textbf{0.6605} & 0.0181 & 0.0160 \\
sequential budget & \textbf{0.6445} & \textbf{0.5005} & 0.0145 & 0.0125 \\
\bottomrule
\end{tabular}
\end{table}

At the instance level, every procedure in Table~\ref{tab:mfdr-sim},
the sequential budget this paper proves sound included, exceeds its own
$\alpha$ under a complete null, by a wide margin, on all three
measures. This is Theorem~\ref{thm:mfdr} at work under a violated
hypothesis: (H2) requires the correlation structure the design premise
assumes to be true of the run, and at the instance level it is declared
away rather than met. At the cluster level, naive thresholding still
exceeds $\alpha$ because it was never a correction to begin with, and
Benjamini--Hochberg, Bonferroni, and the budget all hold, on all three
measures. Once there is something real to discover,
$\mathrm{mFDR}_\eta$ and FWER part ways even at the cluster level: at
skill gap 2, Benjamini--Hochberg's $\mathrm{mFDR}_\eta$ is a controlled
$0.0130$ against an FWER of $0.1110$, which is no defect in a procedure
never built to control FWER; the budget's own FWER stays under
$\alpha$ at the gaps we swept, and Theorem~\ref{thm:mfdr} promises
nothing about that number. A correctly declared design controls what
the design premise is priced to control, with no other error rate
bundled in, and the design read, more than the procedure applied on top
of it, decides whether any of these numbers mean what they say.

\subsection{A Certified Run}
\label{sec:certified}

The two case studies above are retrospective, and (H3) is unverified for
both, for the same reason: nobody knows the true dependence structure of
SWE-bench Verified or BIG-Bench Hard, so it can only be declared, not
checked. A synthetic analysis has no such excuse. This section builds one
whose true generative process is known exactly, checks every hypothesis of
Theorem~\ref{thm:mfdr} against it directly, and runs it once, end to end,
as an ordinary pre-registered study rather than as a reconstruction.

Three synthetic systems, A, B, and C, are compared pairwise, each pairing
over its own disjoint pool of items nested in 16 clusters. A and B share a skill level,
so \lstinline[style=tacetcode]|"A beats B"| is a genuine null; C is ahead
by a fixed gap. Per-cluster, per-system affinity and per-item difficulty
are drawn independently within each pairing's own pool, so observations
are dependent only in the way the declared schema says: paired within an
item across the two systems being compared, mutually independent across
clusters, exactly Theorem~\ref{thm:folding}'s antecedent, met by
construction. That discharges (H2). Giving every pairing its own disjoint
pool is stronger than (H2) alone requires, and it lets (H3) hold
outright: (H3) is \emph{conditional} super-uniformity, given everything
the account has recorded before a claim is priced, including the
$p$-values of the claims before it, and a claim built from data no earlier
claim touched cannot be biased by conditioning on those earlier $p$-values,
whatever that conditioning would otherwise do. With the three pairings
drawing on mutually independent randomness, the three realized $p$-values
are themselves independent, so (H3) for each claim's own null reduces to
the same super-uniformity every exact test enjoys under its own null
(Section~\ref{sec:background}), with no argument about conditioning
required. (H4) holds under the study defaults for
every $\alpha \le 1/2$, mechanically, and (H1) is a fact about the source rather than an assumption about the
analyst: the study script is straight-line code
with no conditional claim-skipping logic of any kind, so every plan entry
is claimed exactly once, in declared order, by construction.

The plan is declared with a schema, before a single row exists: two claims
(\lstinline[style=tacetcode]|"C beats A"|,
\lstinline[style=tacetcode]|"C beats B"|) at a bound of $0.01$ each, chosen
by a power analysis over 400 independent replicates of the same design
(clearing on 88\% and 89\% of
them respectively), and a third (\lstinline[style=tacetcode]|"A beats B"|)
left undeclared, because nothing in that analysis suggested a small
$p$-value was likely there; the partial-pre-registration corollary of
Appendix~\ref{app:proof} is why leaving one claim unbound costs the other
two nothing. The static audit
reports both declared claims affordable (stakes $0.0125$ and $0.0188$) and
the third \emph{unknown}, before \kw{observe} runs. The study is then run
once, on data generated after the plan was written, at a seed independent
of the power analysis: both declared claims are honest and kept
($p=0.00012$ against the declared $0.01$ for
\lstinline[style=tacetcode]|"C beats A"|, $p=0.00244$ for
\lstinline[style=tacetcode]|"C beats B"|), and
\lstinline[style=tacetcode]|"A beats B"|, the genuine null, is correctly
refused at $p=0.9048$.

Every hypothesis of Theorem~\ref{thm:mfdr} holds for this run, so it
controls $\mathrm{mFDR}_\eta$ at $\alpha$, and the honesty of both declared
bounds additionally puts it under Theorem~\ref{thm:bridge}: the static
audit's affordability verdict, computed before a row of data existed, is not
merely consistent with what the runtime monitor went on to do, but is the
fact Theorem~\ref{thm:static} proves it has to be. This is a synthetic
demonstration, and it is not a claim that any real analysis in this paper
meets these hypotheses: the leaderboard and BIG-Bench Hard remain
retrospective, and (H1) for either would still need an external record no
script here supplies. It establishes something narrower and still worth
having: the hypotheses are not vacuous, and a checker-accepted, honestly run
analysis really does inherit the guarantee the theorem promises.

\textbf{Checker runtime.} Appendix~\ref{app:runtime} times
\lstinline[style=tacetcode]|tacet.static| against StatWhy's reported
SMT solver latency (the cost trade of Section~\ref{sec:bhl}) and against a
synthetic sweep of up to 10{,}000 claims.
On the two real case studies' own scale the check is not measurable as a
serious cost at all: \lstinline[style=tacetcode]|prereg.py|'s six
declared plans audit in under 2~ms, and the repository-wide design gate
completes in under a second.

\section{Limitations and Threats to Validity}
\label{sec:limitations}
\label{sec:threats}

Table~\ref{tab:assumptions} collects every assumption the guarantee
depends on, defined or discharged where it is defined or discharged, and
stated as still open where it is not. The paragraphs below add the
limits and threats that do not reduce to a row of the table, each with
the sharpest number we have for it.

\begin{table}[t]
\centering
\footnotesize
\caption{Every assumption the guarantee depends on, defined once here rather
than repeated at each point in the paper that relies on it.}
\label{tab:assumptions}
\begin{tabular}{@{}P{0.28\linewidth}P{0.64\linewidth}@{}}
\toprule
assumption & status \\
\midrule
(H1) complete account (Def.~\ref{def:complete}) &
Assumed, unchecked; the deletion attack
(Table~\ref{tab:deletion}; Appendix~\ref{app:crossing}) measures its cost. \\
(H2) within-claim validity &
Discharged by Theorem~\ref{thm:folding}, given the design is genuine
(next row). \\
(H3) conditional super-uniformity at the stake &
Assumed; \emph{unverified for both case studies} (Section~\ref{sec:mfdr}),
since each family is drawn from one shared, dependent pool. \\
claim order fixed independently of the evidence &
A precondition for (H3), stated separately because it is a fact about the
\emph{program}; taken on trust the same
way (H1) is (``Data-dependent ordering'' below). \\
(H4) bounded refund &
An assumption of Theorem~\ref{thm:mfdr}, which is proved for every $\omega
\le \alpha$; the shipped defaults ($w_0 = \omega = \alpha/2$, $\eta = 1 -
\alpha$) satisfy it for every $\alpha \le 1/2$, mechanically. $\omega > \alpha$
falls outside the
theorem, and the unbounded-refund paragraph below records what that
costs. \\
honesty of declared bounds (Def.~\ref{def:honest}) &
Assumed; dishonesty costs feasibility (Theorem~\ref{thm:static}), never
$\mathrm{mFDR}_\eta$ control (Theorem~\ref{thm:bridge}), by construction. \\
conformance: declared $D$, $I$ genuine (Def.~\ref{def:conforms}) &
Assumed, unchecked by the type system; the runtime alignment check
independently refuses a residual key that is not actually injective. \\
\lstinline[style=tacetcode]|spend|'s $p$-value is valid for the design &
Assumed; \lstinline[style=tacetcode]|spend| prices any $p$-value it
receives, unchecked against the design. \\
fixed horizon (Theorem~\ref{thm:mfdr}) &
Bounds $\mathrm{mFDR}_\eta(t)$ at a horizon fixed in advance rather than
chosen from the data, and the uncovered case is the ordinary one:
\texttt{on\_violation}'s abort is one instance, but an analyst who simply stops
once the pool empties, the behavior this paper's own bankruptcy claims
narrate, is exercising the same uncovered, data-dependent stopping rule.
Optional stopping is the standard, unattempted extension: the
confidence-sequence literature's time-uniform,
anytime-valid bounds~\cite{howard-confseq} are built precisely to survive
a data-dependent stopping rule, and are the natural place to look for it. \\
\bottomrule
\end{tabular}
\end{table}

\textbf{The meaning of $\mathrm{mFDR}_\eta$.}
Alpha-investing~\cite{foster-stine} controls $\mathrm{mFDR}_\eta =
\mathbb{E}[V] / (\mathbb{E}[R] + \eta)$, which differs from the raw
share of discoveries that are false: with the standard $\eta = 1 -
\alpha$, in a simulation of 4{,}000 trials with models of equal skill,
so that every discovery is false by construction, the raw share reached
$1.00$ while $\mathrm{mFDR}_\eta$ read $0.0031$ and passed. Nothing is
broken; this is Foster and Stine's guarantee, which QUDE inherits and so
do we, and because our target analyses are small we state plainly what
we control: $\mathrm{mFDR}_\eta$ at $\alpha$, and never a bound on the
share of an analyst's discoveries that are false. It is also the guarantee the static story
reaches today: the fixed-fraction rule sets its stake from the current
pool alone, which the checker's invariant $\hat W_j \le W_j$
(Theorem~\ref{thm:static}) and the Lean development mirror, whereas
rules controlling FDR itself (LORD and SAFFRON~\cite{lord,saffron}) set
their stakes from the verdict history and have no obvious analogue of
that static under-approximation. An e-value foundation would close
several of this paper's open gaps at once, including (H3)'s unverified
dependence and Theorem~\ref{thm:mfdr}'s fixed-horizon scope, since
e-BH~\cite{wang-ramdas-ebh} controls FDR under arbitrary dependence and
e-processes survive optional stopping; it does not yet exist for the
three mechanisms this paper serves, because turning an exact test into a
composable e-value is itself unsolved work, per mechanism. The
design-from-schema computation and the footprint's purity bit are
orthogonal to which foundation carries them, so we regard the migration
as the separable next step.

\textbf{Weakness of the guarantee.} $\mathrm{mFDR}_\eta$ is a ratio of
expectations and implies neither the false discovery rate nor the
family-wise error rate. On the running example the budget supports
exactly as many claims as Bonferroni, which is none, under both
predictable orderings we tried, and Bonferroni controls the stronger
guarantee; a random ordering typically does better (a median of 28 of
8{,}911, Table~\ref{tab:ordering}), but a random draw is not a claim
about how any particular analysis will actually be run. The budget
contributes its accounting: it is incremental, requiring
no advance knowledge of how many comparisons the analysis will make,
and it reports \emph{when} the analysis ran out. These are properties
of the accounting rather than of its strength.

\textbf{Completeness of the account.} The guarantee is conditional on
being told about every test that was run, and this is the limit no
monitor can detect: a refund arrives only for a supported claim, so
deleting a claim that failed returns its stake, and for an analyst who
never pays for failures $\mathbb{E}[R]$ grows linearly and
$\mathrm{mFDR}_\eta \to 1$, crossing $\alpha$ at 5 tests under a
uniform null and at 4 to 15 on nulls built from the real corpora
(Appendix~\ref{app:crossing}, Tables~\ref{tab:deletion}
and~\ref{tab:crossing}). A claim that was never run leaves no trace for
a process that only observes claims as they are made, and Hardt and
Ullman show that even with every query observed, no computationally
efficient procedure can certify validity against an adaptively chosen
sequence of that many queries~\cite{hardt-ullman}; this is why we
assume (H1) rather than try to verify it, and why the footprint's job
is to make a deletion visible to an external, timestamped record
(Section~\ref{sec:implementation}'s pre-registration scope).

\textbf{The account boundary.} A \kw{Study} instance is the unit every
guarantee is stated over, and nothing stops an analyst from opening a
fresh one per claim: $\alpha$ resets every time, a cheaper version of
the deletion attack, and the per-\kw{Study} sealing gap
Section~\ref{sec:implementation} states is the same boundary seen from
the pre-registration side. Neither gap has a language-level fix; (H1)
is a claim about every \kw{Study} an analysis opened, and an external,
timestamped record supplies it.

\textbf{Claim family reconstruction.} The claim family is a modeling
choice: both case studies reconstruct what an analysis asserts, a
different reader would enumerate differently, and the guarantee is
relative to that choice. Section~\ref{sec:leaderboard} reports the
leaderboard headline under both the full 8{,}911-claim family and the
133-claim adjacent-rank family it contains; we take the visibility of
the disagreement as an argument for the work, since prose leaves the
family implicit and unarguable.

\textbf{No false-refusal rate, no annotation-burden measurement, and no
non-author analyst.} Section~\ref{sec:gallery}'s five refusals show
that the implementation matches the specification, and nothing more: we
never assemble a corpus of analyses known to be valid and measure how
many the checker wrongly turns away, so nothing here separates the
leaderboard's 0 of 8{,}911 into the experiment's real lack of
resolution versus an overly strict checker. Annotation burden is
unmeasured on the same terms: two declarations is an API-surface count,
with no user study behind it. Finally, no third-party analyst and no
non-author codebase appear anywhere in Section~\ref{sec:evaluation}.
All three gaps run in the same direction, toward an evaluation that
demonstrates the checker on cases selected for it and stops short of
what it does in independent use.

\textbf{Predictive validity.} Nothing in this paper tests whether
support under the budget predicts anything external to the accounting
itself, such as whether a supported comparison is more likely to
replicate than a refused one. The evidence here bears on whether the
accounting is internally sound, and the external question is open.

\textbf{Pre-\kw{observe} filtering.} The footprint begins at
\kw{observe}, so data narrowed before that call (a pandas filter on the
raw table, a row dropped in a preprocessing script) is invisible to
every mechanism in this paper, the same blind spot \sys{Tea} and
\sys{Tisane} share (Section~\ref{sec:related}).

\textbf{One investment rule.} The operational semantics, the proof, and
the mechanization cover \lstinline|fixed_fraction| and no other rule;
the five shipped QUDE rules set their stakes from verdict history, and
stating the condition an investment rule must satisfy for the induction
to go through is open. The e-value foundation above is the candidate
way past a rule-by-rule argument, and we have not attempted it.

\textbf{Unbounded refund.} Nothing constrains $\omega$ beyond
$\omega \ge 0$; Theorem~\ref{thm:mfdr}'s accounting and
\lstinline|Plan.audit|'s identification of $T_{1.0}$ with $T_\bot$
(Appendix~\ref{app:semantics}) both hold for the shipped
$\omega = \alpha/2$ and are unproved for arbitrary $\omega$.

\textbf{Unchecked bounds and post-selection $p$-values.} Declared
bounds and realized $p$-values are both unchecked at run time: a
dishonest bound breaks Theorem~\ref{thm:static}, and \kw{spend} prices
whatever number it receives, so a $p$-value invalid for its design
voids Theorem~\ref{thm:mfdr} the same way. In particular, an
outcome-selected sample is refused the one-sided and confirmatory
prices, and that refusal is the whole of the protection: a two-sided
claim built from such a sample is still priced at its unconditional
$p$-value, without the conditioning post-selection inference shows
validity requires~\cite{berk-posi,taylor-tibshirani-si}.

\textbf{Directional labels.} A claim's label names a direction, but a
two-sided rejection certifies only that the two samples differ, and by
default nothing checks the observed sign against the label: two
BIG-Bench Hard tasks reject at $p < 0.05$ with direct prompting ahead
on the raw counts, the opposite of what
\lstinline[style=tacetcode]|"CoT beats direct"| asserts (neither is
among the supported claims, so the headline counts are unaffected).
\kw{claim} and \kw{test} take an opt-in
\lstinline[style=tacetcode]|expect=| argument naming the asserted side
and raising on a rejected claim whose sign disagrees, after the pool
has been charged; it is opt-in because a label is free text, and
inferring the asserted direction from the string is a separate,
unattempted problem.

\textbf{Data-dependent ordering.} Neither the checker nor the monitor
can tell a predictable claim sequence from one built by sorting on the
$p$-values under test
(Section~\ref{sec:comparing-case-studies}'s oracle orderings are the
case in point), and choosing a claim's own position from the evidence
it is about to spend is exactly the dependence (H3) excludes. The
retry after refusal is the one-claim version of the same channel: a
refusal message prints the realized $p$-value by default, and a refused
comparison can be claimed again under a fresh label once refunds
refill the pool. Two opt-in settings
(\lstinline[style=tacetcode]|disclose_realized_price=False| and
\lstinline[style=tacetcode]|dedupe_comparisons=True|) close the channel
when both are on; leaving either off, the default, leaves it open as
stated.

\textbf{Opaque keys and optional clusters.} The design is read off the
key's fields, so a key exposing none is refused rather than defaulted
to \emph{independent}, which would silently assume away whatever
structure the opacity hides; the surviving threat is the override,
since an analyst who knows the structure may declare it, on trust like
every other declaration. An undeclared nesting is likewise invisible by
default, and Section~\ref{sec:leaderboard}'s factor of 48 is its measured
cost on the leaderboard; \kw{observe}'s off-by-default
\lstinline[style=tacetcode]|warn_undeclared_clusters=True| heuristic
flags a key field with few distinct values and never refuses, because
whether a flagged field is a genuine nesting or a comparison target is
a modeling judgment a program cannot make.

\textbf{Partial pairing.} Partial pairing is treated as full pairing on
the intersection, which is conservative and wasteful.

\textbf{Binary outcomes.} \kw{observe}'s fold averages every outcome
sharing an artifact into a pass/fail bit, and Fisher's, McNemar's, and
the cluster sign-flip all consume counts of such bits. The
implementation also ships
\lstinline[style=tacetcode]|observe_continuous| (folding to a mean) and
a paired sign-flip mechanism, so a paired continuous comparison, the
case both case studies need, prices through \kw{spend}; an unmatched
continuous comparison has no mechanism here and remains outside this
paper's scope.

\textbf{Folding.} Folding flatters concentrated failures
(Section~\ref{sec:folding}): a group whose failures cluster on few
artifacts scores better after folding than one whose failures are
spread, at identical raw counts. The declaration makes this visible
rather than preventing it.

\textbf{Per-footprint budgets.} The natural alternative of giving
disjoint footprints disjoint budgets is unsound at scale, because the
$\mathrm{mFDR}_\eta$ cushion is granted per pool and $k$ pools collect
$k$ cushions: measured $\mathrm{mFDR}_\eta$ crosses $\alpha = 0.05$ by
10 pools under a complete null (0.056 at 10 pools, 0.40 at 100), while
a single pooled budget over the same families holds at every $k$
tried. We record why it fails because the idea is a natural one to reach
for.

\section{Related Work}
\label{sec:related}

\textbf{Statistical procedures.} Sequential control of expected false discoveries is Foster
and Stine's~\cite{foster-stine}, and the online FDR literature that followed
it~\cite{lord,saffron,aharoni-rosset} has produced better investment rules
than the one we use, already packaged for a practitioner in the
\lstinline|onlineFDR| R package~\cite{robertson-onlinefdr}: our answer to
``why not just use this'' is that the package still takes $k$ correctly-typed
$p$-values as input, and the two facts our design check and footprint
recover (how the observations are arranged, and what was read) are exactly
what deciding whether a given $p$-value belongs in that sequence at all
requires. The same literature has begun to reach the dependence (H3)
leaves with the analyst: Zrnic, Ramdas, and Jordan handle locally
dependent $p$-value streams~\cite{zrnic-async}, and Wang and Ramdas's
e-BH procedure controls FDR under arbitrary dependence by working with
e-values~\cite{wang-ramdas-ebh}, so the assumption
Section~\ref{sec:mfdr} does not discharge has candidate replacements on
the shelf; adopting either would change the investment rule while
leaving unchanged the two inputs this paper recovers.
Benjamini and Yekutieli~\cite{benjamini-yekutieli}
extended false discovery control itself to the dependent case our design
check is built to recognize, and Benjamini and
Hochberg~\cite{benjamini-hochberg} is the standard against which any of
this is measured. Within machine learning,
Dietterich~\cite{dietterich1998} and Demšar~\cite{demsar2006} are the
standard references for comparing classifiers across datasets, and
Demšar's recommended Friedman-then-Nemenyi procedure is offline in
exactly Bonferroni's sense: the family of classifiers is fixed before
any test runs. Miller shows how to attach a valid confidence interval to
a single language-model evaluation score~\cite{miller-error-bars}, and
Bouthillier et al.\ measure how much of a reported improvement is
attributable to nuisance variation rather than to the intervention under
test~\cite{bouthillier-variance}; both sharpen the estimate a single
comparison reports, which is complementary to the multiplicity question
a whole leaderboard raises. Koehn's paired bootstrap~\cite{koehn2004} is
machine translation's own answer to the design question
Section~\ref{sec:background} states abstractly, and precisely the kind
of test Marie et al.'s survey found papers running less and less often.
Adaptive data analysis~\cite{adaptive-data-analysis} established that
reusing a dataset costs something quantifiable, the reusable
holdout~\cite{reusable-holdout} and Hardt and Ullman's matching hardness
result~\cite{hardt-ullman} mark what a computationally efficient monitor
can and cannot certify, and the Ladder~\cite{ladder} demonstrated what a
leaderboard protocol looks like when that cost is taken seriously. The
same arithmetic already runs by hand outside machine learning: clinical
trials spend an alpha budget across interim analyses by written
protocol~\cite{pocock1977,obrien-fleming1979,lan-demets1983}, and continuous
experimentation platforms spend one across sequential A/B tests inside a
runtime statistics engine~\cite{johari-abtest}; both enforce the same
composition rule this paper does, and neither enforces it with a type
system. Our contribution is upstream of all of it: these procedures
need inputs, and two of those inputs are properties of a program. In the
other direction, large-scale measurements of benchmark
reuse~\cite{roelofs-meta,recht-imagenet,mania-testoveruse} found
substantially less adaptive overfitting in practice than the theory
predicts as a worst case, so we state the guarantee narrowly, as an
accounting property; Section~\ref{sec:leaderboard}'s leaderboard is a case
where the worst case is realized.

\textbf{Meta-science.} Two independent annotation efforts, by researchers
with no stake in our solution, establish that the problem is real: 180 ACL
and TACL papers~\cite{dror-hitchhiker} and 769 machine translation
papers~\cite{marie-credibility}, different fields and years, each finding
the same pattern of ignored or misapplied significance testing. Both are
human annotation; \lstinline|statcheck|~\cite{nuijten-statcheck} checks
a related fact, a $p$-value inconsistent with its own reported test
statistic, mechanically across a quarter of a million printed values,
and is evidence for the same point from a different direction: a
validity property that requires a human referee to notice by hand mostly
goes unnoticed. Dror et al.'s replicability
analysis~\cite{dror-replicability} poses the closest version of
Section~\ref{sec:intro}'s question already: how many comparisons a
paper's evidence actually supports. Gelman and Loken's garden of forking
paths~\cite{gelman-loken} names the failure mode our
\est{\textsc{T-SelectOutcome}} rule prices, and the preregistration
movement~\cite{nosek-preregistration} and the disclosure-based response
to undisclosed analytic flexibility~\cite{simmons-false-positive} name
exactly the discipline Section~\ref{sec:implementation} encodes as a
scoping rule. Multiverse analysis~\cite{steegen-multiverse} and
specification-curve analysis~\cite{simonsohn-spec-curve} answer the same
forking-paths problem by reporting every reasonable analytic path and
its result, a defensible answer to ``which choices would have changed
the conclusion''; \sys{Tacet} answers a different question, ``what did
this specific, already-chosen path cost'', which a multiverse of
unexecuted paths cannot state.

\textbf{Interactive data exploration.} QUDE~\cite{qude} is the closest prior
work and the fairest point of comparison, because it states in print what it
cannot do. It already draws the descriptive-versus-inferential boundary,
though by UI heuristic rather than by declaration, since it observes
interactions and not programs. It already offers five investment rules
designed for exploration, which we do not attempt to improve on. And it is
already correct about order-gaming under mFDR, the reason this paper frames
order dependence as a statement about power (Section~\ref{sec:leaderboard}).
It assumes independence, and names it as an open limitation~\cite{qude};
its suggested mitigation, estimating correlations from the data, cannot
recover a fact like two outcomes sharing a repository.

\textbf{Design-driven test selection.} \sys{Tea}~\cite{tea} and
\sys{Tisane}~\cite{tisane} arrived at design-declaration-to-mechanism first,
in an HCI frame: \sys{Tea} compiles a declared study design into a
constraint satisfaction problem over which tests are valid, and
\sys{Tisane} reasons from a declared conceptual model to the generalized
linear model that respects it (we compare against \sys{Tisane} rather
than its successor \sys{rTisane}~\cite{rtisane}, since \sys{Tisane}'s
released implementation is the one Section~\ref{sec:metatheory}'s gate
can evaluate). Jun et al.'s study of hypothesis
formalization~\cite{jun-hypothesis} is the closest measurement of the
gap this paper also targets, from the declaration side rather than the
multiplicity side. The premise of Section~\ref{sec:lambdat}, that the
design decides the mechanism, predates this paper, appearing first in
these systems; what neither carries over a sequence of such decisions is
why a user of either can still accumulate exactly the family-wise error
of Section~\ref{sec:leaderboard}'s leaderboard, one right test at a time. That
is QUDE's independence gap approached from the other side: QUDE prices a
session without reading the design, and \sys{Tea} and \sys{Tisane} read
the design without pricing the session. Section~\ref{sec:tea-tisane}
gives the full comparison. Probabilistic programming languages such as
Stan~\cite{stan} and Pyro~\cite{pyro} answer what a model believes given
data, a different problem from how many comparisons an already-fitted
analysis can afford to report; nothing in this paper competes with them.

\textbf{Verification of statistical programs.} Belief Hoare Logic~\cite{bhl}
and its tool StatWhy~\cite{statwhy}, extended to Python by
Why3-py~\cite{why3py}, are programming-languages work squarely in this
space: pre- and postconditions on hypothesis tests, discharged by an SMT
solver, that can already name $p$-hacking as the pattern of testing several
datasets and reporting the smallest $p$-value. That logic checks whether the
stated assumptions of the tests a program runs hold, and it can verify a
classical correction across the tests one specification enumerates; it
cannot hold a family whose size no specification states.
Section~\ref{sec:bhl} develops the comparison, including the weight
difference: a specification per test discharged by a solver against two
declarations per analysis discharged by a fold.

\textbf{Graded and linear type systems.} The machinery we borrow is the
graded-type lineage from \sys{Fuzz}~\cite{fuzz} and \sys{DFuzz}~\cite{dfuzz} through
\sys{Duet}~\cite{duet}, \sys{Solo}~\cite{solo}, and \sys{Jazz}~\cite{jazz}, with
\sys{DDuo}~\cite{dduo} as the precedent for tracking a quantity on values at run
time. \sys{Solo}'s design in particular is why footprint inference costs one
annotation: because the units are declared, propagation is a union over
operands, with no linear-type reconstruction to perform. All of these are
instances of a more general framework this paper has not, until now,
named: grading a type system by an arbitrary semiring or resource monoid,
sensitivity being one choice among many, in the sense of coeffects~\cite{coeffects,brunel-coeffect} and their categorical
account as graded (co)monads combining with
effects~\cite{gaboardi-grading}. $\mathcal{T}$'s footprint is such an instance,
but an unusual one for the lineage: the grade is a set of declared units
under union, a where-provenance label with no sensitivity bound in it,
which is why Section~\ref{sec:grade}'s composition can be non-monotone in a way
none of \sys{Fuzz} through \sys{Jazz}'s numeric grades are.
$\clc{\mathsf{Claim}}$'s own grade departs from the same lineage on the
opposite axis: it is numeric, like every prior system's including
\sys{Solo}'s, but composes under function composition rather than under an
addition or scaling action that always commutes, which is where
Figure~\ref{fig:order}'s order dependence comes from
(Section~\ref{sec:claims}'s ``Grades outside the numeric lineage'').

\textbf{Relational cost analysis.} The design point we are answering is
RelCost~\cite{relcost}, which types the difference in cost between two runs
using relational refinements, and BiRelCost~\cite{birelcost}, which exists
to reduce RelCost's annotation burden and frames the problem as relational
type systems not yet achieving the practical appeal of their unary
counterparts. That relationship is the same one \sys{Duet} had to
HOARe$^2$~\cite{hoare2}: a refinement system got there first and was more
expressive, and the graded system that followed traded expressiveness for
automation and inference. We are making the same trade again, for a
different resource.

Note that differential privacy appears throughout that lineage and nowhere
in our subject. It is a tool other researchers used to obtain generalization
guarantees~\cite{adaptive-data-analysis}; here it is neither the mechanism
nor the goal.

\subsection{Comparison with \sys{Tea} and \sys{Tisane}}
\label{sec:tea-tisane}

\sys{Tea} and \sys{Tisane} are the only prior systems that already hold
Section~\ref{sec:lambdat}'s central premise, that the instrument is a
consequence of the study's declared structure rather than a choice the
analyst makes by hand, so the comparison below states what each system
does, what \sys{Tacet} adds, and where each is ahead. Everything below
is read off the released implementations.

\textbf{Declarations.} The declarations line up almost item for item.
\sys{Tea}'s dataset takes an optional key, which its documentation is
explicit makes a design within-subjects, and a study design names which
variable varies within that subject: together those are $\mathcal{T}$'s
$I$, the same information \clc{\textsc{T-Claim}} reads when it decides a
comparison is paired. \sys{Tisane} declares a \lstinline|ts.Unit|, hangs
measures off it, and relates units with \lstinline|nests_within|, which
is $D$ read in the same parent-to-child direction and decides where its
inferred model places a random effect, which is $\gamma$.
Table~\ref{tab:correspondence} lines the three systems up: the upper
block predates this paper, and the ``distributional'' row belongs to
\sys{Tea} and \sys{Tisane} alone.

\begin{table}[t]
\footnotesize
\caption{What each system asks the analyst to declare, and what it does with
the answer. The upper block is the design; the lower block is the
session-level account. The instrument row reads: we refuse a mechanism that
assumes the design away, \sys{Tea} selects one from its test space, \sys{Tisane} infers
a model.}
\label{tab:correspondence}
\begin{tabular}{@{}LLLL@{}}
\toprule
 & \sys{Tacet} & \sys{Tea} & \sys{Tisane} \\
\midrule
unit & artifact & \texttt{key=} & \texttt{ts.Unit} \\
nesting & $D$ & -- & \texttt{nests\_within} \\
cluster level & $\gamma$ & -- & random effect \\
instrument & refuse & select & infer \\
distributional & -- & checked & family/link \\
\midrule
what was read & $\varphi$ & -- & -- \\
price of a claim & $T_p$ & -- & -- \\
family scope & session & one call & one model \\
static verdict & $\bar p$ audit & -- & -- \\
\bottomrule
\end{tabular}
\end{table}

\textbf{Family scope.} In \sys{Tea}, the family is as wide as one call:
a \lstinline|hypothesize| call naming several comparisons corrects for
them, but no call can see the call before it. Two calls are two families
of one, and the leaderboard is 8{,}911 of them, each individually
corrected against a family of size one, the arithmetic
Section~\ref{sec:background} computes as a false-positive probability
above $0.999999$. This is exactly Section~\ref{sec:background}'s
classical setting rather than an oversight in it: the family is fixed
before the analysis starts, and the analysis is one call long.
\sys{Tisane}'s scope is one model for one design, a different and
defensible boundary we return to below.

\textbf{Provenance.} In both systems the dataframe arrives already
chosen, and whatever produced it is upstream and unrepresented: an
analyst who writes Section~\ref{sec:intro}'s narrowing (keep the
instances one system failed, then compare a second system on those) does
it in pandas and produces a table whose provenance neither language
can ask about. This follows from taking a dataframe as the input,
because a value that arrives as a rectangle carries no record of the
rows that are not in it. \kw{observe} is the boundary that makes
\est{\textsc{T-SelectOutcome}} expressible at all, because every
narrowing after it goes through the language rather than around it; the
same analysis is well-formed in all three systems, and only one of them
charges for having looked. Naming this plainly: a footprint is
where-provenance in the database sense~\cite{buneman-provenance}, and
its algebra (a set of declared units, closed under union, tracked
through arithmetic) is the Boolean specialization of the
provenance-semiring structure identified for exactly this
purpose~\cite{provenance-semirings}; Titian~\cite{titian} and
Smoke~\cite{smoke-provenance} are the same idea at production scale,
built for lineage queries rather than for pricing. \sys{Tacet} did not
need to invent a representation for ``what did this value depend on'';
it needed the typing discipline that makes that representation
load-bearing for a statistical budget.

\textbf{Composition.} Neither system carries state between calls, so
nothing composes: \sys{Tea} re-reads $\alpha$ per hypothesis,
\sys{Tisane} infers a model per design, and with no object that survives
a call neither can express Figure~\ref{fig:order}, in which the same two
claims in the other order buy one finding instead of two. Beyond the
fact that ordering can be gamed (QUDE is already correct about that,
and we frame our order result as power rather than soundness), a
running balance lets an analysis say \emph{when} it ran out:
``exhausted by the sixth of 8{,}911'' is a sentence about a quantity
that persists across claims, and a system without one cannot say it at
any severity.

\textbf{Static checking.} Both systems need the data in hand: \sys{Tea}
loads the dataset and checks the analyst's declared assumptions against
it, and \sys{Tisane} assigns data to the design before inferring the
model. \sys{Tacet}'s plan audit runs before \kw{observe} with no data
present, and \lstinline[style=tacetcode]|tacet.static| runs on
unexecuted source, which Section~\ref{sec:grade}'s antitonicity buys. The same gap shows in miniature on one-sided tests: \sys{Tea}
accepts a one-sided hypothesis at the moment of testing, with nothing
recording whether the direction was fixed before the data was seen,
whereas \sys{Tacet} refuses a one-sided claim that no plan declared.

\textbf{Composing the systems.} \sys{Tea} and \sys{Tisane} compose with
\sys{Tacet}, and the seam already exists: \kw{spend} takes a label, a
$p$-value, and a footprint, indifferent to which instrument produced
the number, so a front end that selects the mechanism from the declared
design supplies the part \sys{Tacet} does worst, and \sys{Tacet} prices
the result and carries the balance. The design premise of
\clc{\textsc{T-Claim}} is discharged by construction in such a pairing,
and the declarations would not be written twice, since \sys{Tisane}'s
\lstinline|ts.Unit| is the artifact and its \lstinline|nests_within| is
$D$. Neither half alone is sufficient: an analysis can run the right
test on every comparison and still accumulate the full family-wise
error of Section~\ref{sec:leaderboard}'s leaderboard.

\subsection{Comparison with Belief Hoare Logic and StatWhy}
\label{sec:bhl}

Where \sys{Tea} and \sys{Tisane} choose the instrument, Belief Hoare
Logic proves its use. BHL is a program logic with an epistemic modality,
its judgments concern the statistical beliefs a program is entitled to
after running a test, and StatWhy discharges those judgments over OCaml
programs through Why3 and an SMT solver. The released tool's
specification for a paired $t$-test requires, among other conditions, that
the two datasets are paired, that both populations are normal in the
current belief state, and that each dataset was sampled from the
population it represents; dropping one requirement, as the tool's own
negative example does, causes verification to fail. This is verification
of exactly the judgment \kw{spend} takes on trust, and its docstring
concedes as much. Everything below is read off the released
implementation and its tool paper.

\textbf{Composition across tests.} BHL does carry a price across tests:
StatWhy's running example executes two one-sided $t$-tests and licenses
belief in the disjunction of the two hypotheses at $p_1 + p_2$, the
union bound, verified, and the library carries verified specifications
for Tukey's, Dunnett's, and Steel's procedures and for Fisher's and
Stouffer's combinations. The gap is that the family is as wide as the
specification: $k$ is fixed the moment the ensures clause is written,
which is Section~\ref{sec:background}'s classical setting with a proof
attached. A fully verified leaderboard analysis is available today, and
it is the leaderboard's Bonferroni row (Section~\ref{sec:leaderboard}): a verified conclusion
that supports nothing, at any effect size, because the experiment lacks
the resolution for the family, the same diagnosis power analysis makes
of underpowered comparisons in NLP~\cite{card-power}. What no
specification in StatWhy's released style (one precondition per test,
fixed before the solver runs) can express is any sequential row of the
same table, because such a specification cannot enumerate a family whose
next member the analyst has not decided on yet. The budget's
contribution over a verified Bonferroni lies elsewhere:
Section~\ref{sec:limitations}'s accounting properties,
bought at the price of the weaker $\mathrm{mFDR}_\eta$ guarantee.

\textbf{The design premise.} The design enters BHL's system as an axiom
and \sys{Tacet}'s as a computation: the paired $t$-test's precondition
literally reads \lstinline|paired y1 y2|, an assertion the analyst
writes and the solver takes as given, whereas
Section~\ref{sec:lambdat} computes \emph{paired} and \emph{clustered}
from the artifact keys, so on structured keys the premise is derived
rather than trusted. The correspondence is sharpest on one-sided tests:
StatWhy's \lstinline|Up| alternative requires the belief state to
already exclude the opposite direction, an epistemic rendering of ``the
direction was fixed in advance'' that the analyst asserts, while
\sys{Tacet}'s \kw{preregister} scope makes the same fact syntactic,
dated by its position before \kw{observe}. Neither system can check the
analyst's mind; one trusts a stated belief, the other checks where a
declaration appears in the program, and the second is checkable by a
machine that has never met the analyst.

\textbf{Weight and audience.} A StatWhy user writes a specification per
test in a program logic and waits on a solver (the CAV tool paper's
contribution is partly constructs and libraries to make that burden
bearable, which concedes where the obstacle is); our analyst writes two
declarations, and the check is a fold. Against
Section~\ref{sec:intro}'s population, where most published comparisons
run no test at all, that cost difference is the difference between a
tool for the analysts who already care and a gate cheap enough to put
in front of the ones who do not, a difference
Appendix~\ref{app:runtime} makes a number rather than an assertion. The
expressiveness runs the other way, as it did for HOARe$^2$ against
\sys{Duet}: nothing prevents a Hoare-style logic from carrying wealth
as ghost state and verifying every budget invariant in this paper, at a
specification per claim. The graded system states less and states it
automatically, and we are making that trade knowingly, a third time.

\textbf{Selection, verification, accounting.} The three systems answer three
different quantifiers about the same analysis: \sys{Tea} and \sys{Tisane}
answer \emph{which} instrument the design admits, BHL and StatWhy answer
whether \emph{this} use of the instrument is sound, and \sys{Tacet} answers
\emph{how many} uses the analysis can afford, in the order it makes them.
The three compose in principle because all three meet at the same interface,
the mechanism: in \clc{\textsc{T-Claim}}, $\mathit{assumes}(m)$ is today a
declaration taken on trust, a \sys{Tea} or \sys{Tisane} front end would
supply $m$, and a StatWhy certificate is precisely a proof of the
requirements that declaration abbreviates. Figure~\ref{fig:three-layers}
(Appendix~\ref{app:threelayers}) lays the three surfaces out on the smallest
analysis all three can express; no software connects the three panes today.
As the family grows, \sys{Tisane} still names the right instrument, StatWhy
still proves each use and its union bound, and only \sys{Tacet} can report
that the analysis went bankrupt at claim six.

\section{Conclusion}
\label{sec:conclusion}

We have presented \sys{Tacet}, a language in which an analysis declares what it
generated, states what it expects to find, and is refused any claim it
cannot afford or cannot properly test. The wealth accounting descends from
Foster and Stine's alpha-investing, by way of Aharoni and Rosset's
generalized rule (Section~\ref{sec:grade}). What we add precedes it: two
of the inputs a multiple-comparison procedure needs are properties of the
program rather than of its results, and a type system is therefore in a
position to supply them. Whether a narrowing predicate consulted an outcome
is syntactic. Whether a comparison is paired or clustered is computed from
declared dependencies between key fields, before the analysis runs, and we
prove that computation sound (never licensing a mechanism whose alignment
can silently misalign observations) given that the declared dependencies
and identifying key are genuine facts about the run (Theorem~\ref{thm:design}), a proviso the type
system cannot check and a runtime alignment check catches
independently of it. Neither is recoverable from a list of $p$-values,
which is all a statistician usually receives.

\bibliographystyle{alphaurl}
\bibliography{refs}

\appendix

\section{Syntax of \texorpdfstring{$\mathcal{T}$}{T}}
\label{app:syntax}

Section~\ref{sec:lambdat} gives this grammar without the metavariables it
ranges over, since Section~\ref{sec:metatheory}'s appendices need them
named. $r$ ranges over runs (sets of observed rows), $a$ over artifact-key
expressions, $k$ over key literals, $q$ over predicates given the outcome, $s$
over free text, $\ell$ over claim labels, $m$ over mechanisms, $\bar p \in
[0,1] \cup \{\bot\}$ over declared bounds ($\bot$ meaning none was
declared), $\pi$ over plans, $\sigma \in \{1,2\}$ over a declared
claim's sidedness ($1$ one-sided, $2$ two-sided), $\kappa$ over constraints
$F \rightharpoonup \mathit{Val}$ (Section~\ref{sec:lambdat}), and $\oplus$
over $\mathit{Ops}$, the fixed signature of arithmetic operators
Section~\ref{sec:estimation} names. $v$ ranges over the sample sort:
Section~\ref{sec:lambdat} restricts selection to it precisely so that
$[k]$ and $\langle q \rangle$ can never be written over an arithmetic
combination, which has no per-observation structure left to select from.

{\setlength{\arraycolsep}{2pt}
\[
\begin{array}{llll}
\est{\text{sample}} & \est{v} & \est{::=} &
  \est{\kw{observe}_\Sigma(r, a) \mid v[k] \mid v\langle q \rangle} \\
\est{\text{estimation}} & \est{e} & \est{::=} &
  \est{v \mid e_1 \oplus e_2} \\
\clc{\text{claim}} & \clc{c} & \clc{::=} &
  \clc{\kw{note}(s) \mid \kw{claim}^m_{\bar p}(\ell, v_1 \mathbin{\kw{vs}} v_2) \mid c_1 \,;\, c_2} \\
\clc{\text{plan}} & \clc{\pi} & \clc{::=} &
  \clc{\cdot \mid \pi, \ell{:}\langle \bar p, \sigma, \kappa_1, \kappa_2 \rangle} \\
\clc{\text{program}} & \clc{P} & \clc{::=} &
  \clc{\kw{plan}\ \pi\ \kw{in}\ c}
\end{array}
\]}

Appendix~\ref{app:semantics} steps through the normal form
$c_1 \,;\, \dots \,;\, c_n$ this grammar associates every program into.
Appendix~\ref{app:plan} gives the typing rules $\pi$ adds.

\section{Soundness of the Footprint and the Purity Bit}
\label{app:estimation}

Section~\ref{sec:lambdat} states that no typing rule consumes the
footprint, and that the purity bit is set by the form of a selection
rather than by an analysis of its predicate. On their own those are facts
about syntax. The two guarantees the rest of the paper reads off the
tracking (that a footprint is a faithful report of what was read, and that
a $\est{\clean}$ comparison cannot have been chosen by looking) are facts
about what estimation terms denote, and this appendix states and proves
both.

An estimation term selects observations from a run and computes over them,
so its denotation is the set of observations the resulting value rests on.
For a run $r$ and an artifact key expression $a$, $U(r,a)$ is the set of
folded observations of Section~\ref{sec:analysis}: one observation $u$ per
artifact key that $a$ assigns to the rows of $r$, carrying that artifact's
key fields and its folded outcome $\mathit{out}(u)$, so that
$\est{\mathit{units}(r,a)}$ (Section~\ref{sec:estimation}) is exactly the
key set of $U(r,a)$. Matching of an observation against a key literal is
Section~\ref{sec:schemas}'s.

\begin{defi}[Sample denotation]
\label{def:denote}
$\den{\est{\kw{observe}_\Sigma(r, a)}} = U(r,a)$;
$\den{\est{v[k]}} = \{u \in \den{\est{v}} \mid u \text{ matches } k\}$;
$\den{\est{v\langle q \rangle}} = \{u \in \den{\est{v}} \mid
q(\mathit{out}(u))\}$; and
$\den{\est{e_1 \oplus e_2}} = \den{\est{e_1}} \cup \den{\est{e_2}}$, the
observations the derived value rests on, with the value itself
$\est{\oplus}$ applied directly to $e_1$'s and $e_2$'s own values, as
Section~\ref{sec:estimation} fixes. The first three clauses define
$\den{\cdot}$ on the sample sort $v$; the last extends it to all of $e$,
agreeing with the first three wherever $e$ is itself a sample. For a set
$X$ of observations, $\mathit{keys}(X)$ is the set of artifact keys they
carry, the same notion $\est{\mathit{units}(r,a)}$ names for $U(r,a)$
(Section~\ref{sec:estimation}); a footprint $\varphi$ is always a set of
keys, so a claim comparing it against $\den{\est{e}}$, a set of
observations, is stated through $\mathit{keys}(\cdot)$.
\end{defi}

\begin{lem}[Footprint soundness]
\label{lem:footprint}
If $\est{\vdash e : \Stat{\Sigma}{\varphi}{\kappa}{t}}$, then:
(a) $\mathit{keys}(\den{\est{e}}) \subseteq \varphi$;
(b) every observation of $\den{\est{e}}$ satisfies every pin $\kappa$
declares; and
(c) whether an observation belongs to $\den{\est{e}}$ is a function of the
artifact keys and that observation's own outcome alone. In particular,
altering the outcomes of observations outside $\varphi$ changes neither
$\den{\est{e}}$ nor any value computed from it.
\end{lem}

\begin{proof}
By induction on the typing derivation, one rule at a time.
\est{\textsc{T-Observe}}: $\den{\est{v}} = U(r,a)$, whose key set is
$\varphi$ by definition; the constraint is empty; and membership consults
no outcome at all. \est{\textsc{T-SelectKey}}: the selection keeps exactly
the members of $\den{\est{v}}$ matching $k$, so (a) follows from the
induction hypothesis and $\psi$'s definition as the matching subset of
$\varphi$; the survivors match $k$ and satisfied $\kappa$ already, giving
(b) for $\kappa \uplus k$; and the filter consults the key alone,
preserving (c). \est{\textsc{T-SelectOutcome}}: the sample shrinks while
$\varphi$ and $\kappa$ are unchanged, so (a) and (b) are inherited, and the
filter consults exactly $u$'s own outcome, $q(\mathit{out}(u))$,
preserving (c). \est{\textsc{T-Arith}}: for (a),
$\mathit{keys}(\den{\est{e_1}} \cup \den{\est{e_2}}) =
\mathit{keys}(\den{\est{e_1}}) \cup \mathit{keys}(\den{\est{e_2}}) \subseteq
\varphi_1 \cup \varphi_2$ by the induction hypotheses; for (b), a member of
$\den{\est{e_1}}$ satisfies $\kappa_1$, which agrees with $\kappa_1 \sqcap
\kappa_2$ everywhere the latter is defined, and symmetrically for
$\den{\est{e_2}}$; for (c), membership in a union of sets with the
property has the property. For the final sentence: by (c), changing
$\mathit{out}(u)$ for $u \notin \varphi$ (reading $u \notin \varphi$ as
$\mathit{key}(u) \notin \varphi$) changes no observation's membership,
since membership depends only on that observation's own outcome and, by
(a), every member's key lies in $\varphi$; the value is then computed
from the keys and the outcomes of $\den{\est{e}}$'s members alone.
\end{proof}

\begin{lem}[Purity soundness]
\label{lem:clean}
If $\est{\vdash e : \Stat{\Sigma}{\varphi}{\kappa}{\clean}}$, then
membership in $\den{\est{e}}$ is a function of the artifact keys alone:
replacing the run at each \kw{observe} leaf by one folding to the same
artifact keys, with any outcomes whatever, leaves the selection unchanged.
\end{lem}

\begin{proof}
By induction on the derivation. \est{\textsc{T-Observe}} selects every
folded observation and consults no outcome. \est{\textsc{T-SelectKey}}
filters by a key literal, preserving the property.
\est{\textsc{T-SelectOutcome}} concludes $\est{\rd}$, so the case is
vacuous. \est{\textsc{T-Arith}} concludes $\est{\clean}$ only when both
operands do, since $\est{t_1 \sqcup t_2} = \est{\clean}$ forces
$\est{t_1} = \est{t_2} = \est{\clean}$, and a union of key-determined sets
is key-determined.
\end{proof}

\noindent Lemma~\ref{lem:clean} is the noninterference statement behind
the purity bit, in the sense the Dependency Core Calculus citation of
Section~\ref{sec:estimation} promises: $\est{\clean}$ in a type means the
selection is fixed by the keys, before any outcome exists, which is the
property pre-registration needs, and it entitles
\clc{\textsc{T-OneSided}} and \clc{\textsc{T-Registered}}
(Appendix~\ref{app:plan}) to read their $\est{\clean}$ premises as ``this
comparison could not have been chosen by looking''. The bit is set
syntactically and conservatively, since a predicate that ignores the
outcome argument still types $\est{\rd}$
(Section~\ref{sec:estimation}); the lemma is the sound direction of that
discipline, and no derivation can launder an outcome-dependent selection
into a $\est{\clean}$ type, because \est{\textsc{T-SelectOutcome}} is the
only rule whose selection consults an outcome and its conclusion is always
$\est{\rd}$.

The value of a $\est{\clean}$ term is a different matter, deliberately: a
board's pass rate reads every outcome it averages, and must. The bit
governs which observations a value rests on, and a registered comparison's
$p$-value is then computed from outcomes through its mechanism, over a
selection fixed in advance, which is the shape blind analysis prescribes
(Section~\ref{sec:implementation}). Lemma~\ref{lem:footprint}(b) is the
$\kappa_i$-consistency Theorem~\ref{thm:design} assumes of its two
samples, discharged for every sample the language itself produced
(Section~\ref{sec:metatheory}); clause (c) makes the footprint a report
worth printing, since a report that could omit an outcome the value
depends on would be decoration. Clauses (a) and (b) of Lemma~\ref{lem:footprint}
and Lemma~\ref{lem:clean} in full are mechanized
(Section~\ref{sec:mechanization}); clause (c)
alone remains on paper, since it is a two-run substitution property rather
than a fact a single term's structural recursion states.

\section{Design-Premise Details}
\label{app:design}

Section~\ref{sec:design-premise} defines the design computation; this
appendix records three details of it: the closure's direction, the
single-field restriction on $D$, and the measured cost of pricing the
fourth lattice point.

\textbf{Direction of the closure.} $D$'s edges are written child $\mapsto$
parent, and the closure runs \emph{parent to child}: it adds $f$ because its
parent $g$ is already in, never $g$ because its child $f$ is. This is why
$\mathit{distinguishing}$ is stated as a least fixed point rather than
as ``the closure of $\mathit{differing}$ under $D$'': read literally, the
latter closes forward along $D$'s edges, which is child to parent, the
opposite of what the rule is. A task determines its mode, so two selections
pinning different modes necessarily draw on disjoint tasks, and the task
field is distinguishing too; the converse fails, since two different tasks
may share a mode. Running the closure in the wrong direction makes every
cross-mode comparison look paired, and comparing the static answer against
the value-level one (Section~\ref{sec:metatheory}) catches that mistake
immediately.

\textbf{Single-field dependencies.} $D \subseteq F \times F$ admits only
dependencies with a single field on the left. Functional dependencies in
general are relations on $2^F \times F$: a pair of fields may jointly
determine a third that neither determines alone. The restriction is a
simplification, not something the rest of the development needs. Reading
$D \subseteq 2^F \times F$ and taking the least $X \supseteq
\mathit{differing}$ closed under ``$(G \mapsto f) \in D$ and $G \subseteq X$
imply $f \in X$'' is the same fixed point over a larger rule set, and
everything downstream consumes only the resulting set
$\mathit{distinguishing}$. We keep the single-field form because every
dependency the case studies declare has that shape (an instance determines
its repository, a task determines its mode) and because it keeps the
closure a graph reachability rather than a datalog evaluation.

\textbf{The cost of pricing the fourth lattice point.} Serving the cluster
sign-flip on the refused point (clustered without pairing) has a measured
cost: the alignment key degenerates to
the cluster alone, all but one observation per cluster per side silently
vanishes, and which one survives is dictated by dictionary insertion order.
On a comparison carrying no signal at all, with each side passing 8 of 16
so that Fisher's exact test returns $p = 1.0$, the served test returned
$p = 0.0078$ and the claim was supported; the identical dataset returns
$p = 1.0$ or $p = 0.0078$ depending only on the order its rows arrived in.

\section{Typing Rules for Pre-registration}
\label{app:plan}

Section~\ref{sec:metatheory} states Theorem~\ref{thm:mfdr} (mFDR control)
against the account this appendix formalizes;
Section~\ref{sec:implementation} describes the same feature in prose and
gives one judgment, \ty{\textsc{Ok}}, for the audit it makes possible. The
typing rules that prose describes follow, extending
Section~\ref{sec:lambdat}'s judgment to $\pi \vdash c : \Claim{w}$, a claim
sequence checked against a declared plan. \clc{\textsc{T-Note}},
\clc{\textsc{T-Claim}}, and \clc{\textsc{T-Seq}} carry $\pi$ unused; only
the two rules below consult it.

\begin{mathpar}
\inferrule[T-Plan]{\clc{\pi \vdash c : \Claim{w}}}
  {\clc{\vdash \kw{plan}\ \pi\ \kw{in}\ c : \Claim{w}}}
\end{mathpar}

\noindent There is no premise relating $\pi$ to any estimation term: a plan
entry is a tuple of literals (a label, a bound, a sidedness flag, and two
constraints) and the grammar (Appendix~\ref{app:syntax}) gives it no way to
mention an estimation term. \kw{observe} introduces the only values an
estimate can come from (\est{\textsc{T-Observe}}), and $\pi$ is written
before it does, so ``$\pi$ was fixed without having looked'' is not a fact
this rule checks; the grammar makes it true.

\begin{mathpar}
\inferrule[T-OneSided]{\est{\vdash v_i : \Stat{\Sigma}{\varphi_i}{\kappa_i}{\clean}} \\
  \ty{\mathit{design}(\Sigma, \kappa_1, \kappa_2) = \langle \mathit{pr},
   \mathit{cl} \rangle} \\
  \ty{\mathit{assumes}(m) = \langle \mathit{pr}, \mathit{cl}' \rangle} \\
  \ty{\mathit{cl} \sqsubseteq \mathit{cl}'} \\
  \ty{\mathit{sided}(m) = 1} \\
  \clc{\pi(\ell) = \langle \bar p, 1, \kappa_1, \kappa_2 \rangle}}
  {\clc{\pi \vdash \kw{claim}^m_{\bar p}(\ell, v_1 \mathbin{\kw{vs}} v_2) :
   \Claim{T_{\bar p}}}}

\inferrule[T-Registered]{\est{\vdash v_i : \Stat{\Sigma}{\varphi_i}{\kappa_i}{\clean}} \\
  \ty{\mathit{design}(\Sigma, \kappa_1, \kappa_2) = \langle \mathit{pr},
   \mathit{cl} \rangle} \\
  \ty{\mathit{assumes}(m) = \langle \mathit{pr}, \mathit{cl}' \rangle} \\
  \ty{\mathit{cl} \sqsubseteq \mathit{cl}'} \\
  \ty{\mathit{sided}(m) = \sigma} \\
  \clc{\pi(\ell) = \langle \bar p, \sigma, \kappa_1, \kappa_2 \rangle}}
  {\clc{\pi \vdash \kw{claim}^m_{\bar p}(\ell, v_1 \mathbin{\kw{vs}} v_2) \;
   \mathsf{confirmatory}}}
\end{mathpar}

\noindent \clc{\textsc{T-OneSided}} is the only rule in $\mathcal{T}$ that can
give a one-sided mechanism a type, and $\mathit{sided}(m) = 1$ is the
premise that makes that true rather than a convention: the matching premise
$\mathit{sided}(m) = 2$ on \clc{\textsc{T-Claim}}
(Section~\ref{sec:lambdat}) stops that rule from typing a one-sided
mechanism against no registration at all; stating the restriction in prose
instead would leave the refusal Section~\ref{sec:gallery} advertises (the
unplanned one-sided test) outside the rules. Soundness is the reason the
premise exists: a one-sided $p$-value is
valid only if the direction of the comparison was fixed before the data that
decides it was read, and \clc{\textsc{T-Plan}} is the only rule that can
witness that a choice predates \kw{observe}.

\clc{\textsc{T-Registered}} is a second judgment layered onto the same
leaf's existing typing: a claim already typed by \clc{\textsc{T-Claim}} or
\clc{\textsc{T-OneSided}} is additionally \emph{confirmatory} exactly when
$\pi$ holds a matching label at the same constraints. Sharing the
metavariables $\kappa_1, \kappa_2$ between the estimation and plan premises
forces the match: the registered constraints must be the literal
ones the claim used, or the label alone would be fixed in advance.

It repeats the design premise rather than assuming it. Being a second
judgment on the same leaf is not the same as being derivable only for a leaf
that already has a type: without the repetition, nothing above stops
\clc{\textsc{T-Registered}} from certifying \textsf{confirmatory} for a
claim whose design premise fails, because no premise left standing would
constrain what $m$ \emph{assumes}: $\mathit{sided}(m) = \sigma$ still
constrains $m$'s sidedness, but says nothing about $\mathit{assumes}(m)$,
the design-lattice fact the repeated premise pins down. The
$\mathit{sided}(m) = \sigma$ premise ties the same $m$ to the sidedness
$\pi$ declared, which is the other half of what a registration is supposed
to fix. It does not fix which mechanism among those with the same sidedness
is used: $\pi(\ell)$ constrains $\bar p, \sigma, \kappa_1, \kappa_2$ but not
$m$ itself, so a registration is silent on whether the confirmed comparison
is priced by Fisher's exact test or by any other two-sided instrument
satisfying the design premise. Naming $m$ in the plan entry would close
this, at the cost of a fifth field every registration would have to carry;
we have not made that change.

A claim built from \est{$v\langle q \rangle$} can never be confirmatory, and
the purity bit makes that true. Sharing $\kappa_1, \kappa_2$ does not
suffice on its own: \est{\textsc{T-SelectOutcome}}
(Section~\ref{sec:lambdat}) propagates its input's constraint unchanged, so
$\est{\kw{observe}_\Sigma(r,a)[\mathit{repo}{=}A]}\allowbreak\est{\langle q
\rangle}$ carries exactly the constraint $\{\mathit{repo} \mapsto A\}$ that
$\est{\kw{observe}_\Sigma(r,a)}\allowbreak\est{[\mathit{repo}{=}A]}$ does,
and a plan naming that constraint would match it.
\est{\textsc{T-SelectOutcome}} separates the two by setting the purity bit
to $\est{\rd}$, and both rules above demand $\est{\clean}$ on both
estimation premises. This is a premise, not a consequence of the constraint
machinery, which is why it is written down: the constraint records which key
literals were pinned, and an outcome-dependent narrowing pins none of them
while changing which observations survive. Setting $\kappa := \emptyset$ in
\est{\textsc{T-SelectOutcome}} instead would close this hole and break the
design premise, since the design is computed from exactly those constraints
and would then read every outcome-filtered comparison as degenerate. The
implementation can and does drop the pin at \kw{select} instead, because it
reads the design off observed key values rather than off constraints
(Section~\ref{sec:schemas}), and Section~\ref{sec:gallery} shows the
resulting refusal.

\textbf{Why dropping the pin cannot mismatch.} Emptying the pin, rather
than leaving it at whatever \kw{select} was called on, only diverges from
$\mathcal{T}$ when some registration's declared side could itself be
$\emptyset$: a registration pinning nothing is exactly what an emptied pin
would then satisfy, which is the scenario worth ruling out rather than
assuming away. It cannot arise. A plan entry $\pi(\ell) = \langle \bar p,
\sigma, \kappa_1, \kappa_2 \rangle$ can only ever be matched by
\clc{\textsc{T-Registered}} against a leaf already typed by
\clc{\textsc{T-Claim}} or \clc{\textsc{T-OneSided}}, both of which share the
same design premise, and Section~\ref{sec:design-premise}'s non-degeneracy
conditions make that premise \emph{undefined} whenever
$\mathrm{dom}(\kappa_1) = \emptyset$: $\mathrm{dom}(\kappa_1) =
\mathrm{dom}(\kappa_2)$ would then force $\mathrm{dom}(\kappa_2) = \emptyset$
too, and both empty forces $\mathit{differing} = \emptyset$, which the same
paragraph refuses outright. So no plan entry declaring an empty pin on
either side can ever be confirmed by \emph{any} claim, honest or
outcome-selected: such an entry is not a laxer registration the library's
simplification might wrongly satisfy, it is one the design premise already
makes permanently unmatchable, for the same reason a comparison of a set
against its own subset is. Every registration the simplification could
possibly affect therefore pins something nonempty on both sides, and an
emptied pin can never equal a nonempty declared one; the two mechanisms
enforce the same judgment at different doors because the design premise's
own non-degeneracy requirement makes the library's simplification sound,
beyond the fact that both were built to agree. A
leaf typed by \clc{\textsc{T-Claim}} or \clc{\textsc{T-OneSided}} whose
label is absent from $\pi$, or whose constraints do not match, is
\emph{exploratory} instead, priced identically and not otherwise
distinguished by the type system.

\textbf{Order, omission, and repetition.} Order and omission are not typing
obligations. Nothing above stops a derivation from placing a plan's later-declared label
before its earlier-declared one in $c_1 ; \dots ; c_n$, or from omitting a
declared label entirely: $\clc{\mathsf{Claim}}$'s grammar has no production
that could express either constraint, because both are properties of a
whole program relative to $\pi$, beyond the reach of any single leaf's
derivation. Repetition is the same story once more: nothing above stops two
distinct leaves from both matching label $\ell$ and both being certified
confirmatory, each independently, while $\mathit{affordable}$'s fold charges
$\ell$'s stake once regardless, so a derivation with a repeated label
describes a run the audit never priced, the reordering mismatch reached by
repetition instead. The implementation therefore refuses both: a second
plan \emph{entry} under an already-declared label, and a second
\emph{claim} against an already-claimed one (\lstinline|OutOfScope|), so a
plan's labels are claimed at most once and the fold's once-per-label
charge matches the claims the monitor can see. At the audit, all three of
order, omission, and repetition are
discharged the way Section~\ref{sec:implementation} discharges affordability:
$\mathit{affordable}$ folds $\pi$'s declared bounds in declared order, so
evaluating it on a differently-ordered $c$ answers a question about a
program that was not checked, and it charges every declared label's stake
whether or not a matching leaf appears in $c$, so an omitted claim recovers
no wealth. A plan's declared order is load-bearing without being a premise
of any rule above, in the same way that a key literal is admissible in $\pi$
without being reachable by any of Section~\ref{sec:lambdat}'s estimation
rules: neither is data. Insertion is the residual case: a \kw{claim} leaf
whose label appears nowhere in $\pi$ (an exploratory claim, typed by
\clc{\textsc{T-Claim}} like any other) can stand between two declared
ones, and $\mathit{affordable}$'s fold never charges it, where the checker
of Appendix~\ref{app:semantics}, walking the claim sequence itself,
charges it $T_\bot$. Section~\ref{sec:implementation}'s fourth caveat
states what that difference costs the shipped audit's certificate.

\begin{defi}[Complete account]
\label{def:complete}
$\pi, c$ is a \emph{complete account} if, for every $\ell \in
\mathrm{dom}(\pi)$, $\ell$ is absent from $c$ only for a reason independent
of the price any mechanism consistent with the registration would have
returned on the registered comparison, in
particular, not because an earlier attempt at $\ell$ was made, produced an
unfavorable price, and was deleted before being logged.
\end{defi}

\noindent Nothing in $\mathcal{T}$'s grammar or rules can check
Definition~\ref{def:complete}: it is a claim about the process that produced
$c$, external to any property of $c$ itself, in exactly the sense of
Appendix~\ref{app:semantics}'s honest bounds: neither is decidable from the
term, and an external, timestamped record has to supply both
(Section~\ref{sec:implementation}). \clc{\textsc{T-Plan}} contributes legibility instead: a label $\pi$ declares with no matching leaf in $c$ is a hole
$\mathit{affordable}$ prices anyway and Section~\ref{sec:implementation}'s audit names in its
report, so an incomplete account leaves a trace where an unplanned analysis
leaves none. Section~\ref{sec:limitations}'s deletion experiment measures what an incomplete
account costs in the one case that has no $\pi$ at all to leave a trace in.

\section{Operational Semantics of the Runtime Monitor and the Checker}
\label{app:semantics}

The runtime monitor and the static checker are two labelled transition
systems over the same claim sequence, differing only in where the $p$-value
that drives $T$ comes from. Both are stated here for a single wealth pool,
which is the design throughout (Section~\ref{sec:threats} reports why the
per-footprint alternative is unsound).

A \emph{runtime configuration} is a pair $\langle W, c \rangle$ of the current
wealth and the remaining program. Let $\mathit{price}(m, v_1, v_2)$ be the
actual $p$-value mechanism $m$ returns when run on the concrete data $v_1$ and
$v_2$ denote. It is an external oracle this appendix does not model further,
exactly as Section~\ref{sec:related} treats the arithmetic itself as Foster and Stine's.

\begin{mathpar}
\inferrule[E-Note]
  { }
  {\clc{\langle W,\; \kw{note}(s) \,;\, c \rangle \longrightarrow \langle W,\; c \rangle}}

\inferrule[E-Claim]
  {\est{p = \mathit{price}(m, v_1, v_2)} \\ \ty{\alpha_j = i \cdot W} \\
   \ty{b = [\, p \le \alpha_j \,]}}
  {\clc{\langle W,\; \kw{claim}^m_{\bar p}(\ell, v_1 \mathbin{\kw{vs}} v_2) \,;\, c \rangle
   \overset{b}{\longrightarrow} \langle T_p(W),\; c \rangle}}
\end{mathpar}

\noindent $T_p$ is Equation~\ref{eq:transformer}: the stake $\alpha_j = i
\cdot W$ at this step $j$ (distinct from $\alpha$, the overall mFDR target
level Theorem~\ref{thm:mfdr} states the guarantee at) is deducted from the
pool directly, whether or not the claim rejects, with $\omega$ refunded on
rejection. This is \lstinline|tacet.fixed_fraction|, the study default and
the only rule this section, Appendix~\ref{app:proof}, and the mechanization
cover (Section~\ref{sec:grade}). Because $\alpha_j = iW$ and $0 < i \le 1$,
$T_p(W) = (1-i)W + \omega \cdot [\,p \le \alpha_j\,] \ge 0$ for every $W \ge 0$
by arithmetic alone: $(1-i)W \ge 0$ since $i \le 1$, and the refund term is
non-negative. \clc{\textsc{E-Claim}} therefore carries no solvency premise,
and no non-empty runtime configuration is stuck under it.

At $W = 0$ the stake $\alpha_j$ is itself $0$, and $b = [\, p \le 0 \,]$ is
true only if $\mathit{price}$ can return exactly $0$. It cannot for any of
the three served instruments: Fisher's, McNemar's, and the cluster sign-flip
are exact tests over a finite outcome space, so every achievable $p$-value
is a sum of strictly positive probabilities of discrete events and is
itself strictly positive. A mechanism that could return $p = 0$ would be
accepted and refunded at zero cost, repeatably; none of the three this
paper serves can.

A different rule can need one. The five QUDE investment rules the
implementation also ships (\lstinline|tacet.qude|) price a claim at
$\alpha_j$ but deduct
\[
  \ty{\mathit{cost}(\alpha_j) \;=\; \frac{\alpha_j}{1 - \alpha_j}}
\]
on a loss (Foster and Stine's original deduction, strictly more than
$\alpha_j$ itself, \lstinline|price_of| in \lstinline|tacet.qude|) and
refund a flat $\omega$ on a win. Because $\mathit{cost}$ is superlinear, a
stake the pool nominally covers can still price the pool negative, so those
rules carry their own solvency premise, $\mathit{cost}(\alpha_j) \le W$, and
each keeps it true by capping $\alpha_j$ at $W/(1+W)$
(\lstinline|_cap_nonneg|). \lstinline|fixed_fraction| and the five QUDE
rules are two different, independently sound choices of $(\alpha_j,
\text{deduction})$, not the same rule under two names; Theorems~\ref{thm:static}
and~\ref{thm:mfdr}, and everything mechanized in Lean, are stated for
\lstinline|fixed_fraction| alone, exactly as Section~\ref{sec:grade}'s scope
paragraph says.

\clc{\textsc{E-Claim}} labels the step with $b$, true exactly when the
runtime monitor \emph{accepts} that claim; this is \kw{spend} (Section~\ref{sec:implementation})
taken one call at a time. Iterating from
$W_0 = w_0$, the initial budget, over $c_1 \,;\, \dots \,;\, c_n$ produces a
wealth trace $W_0, \dots, W_n$ and verdicts $b_1, \dots, b_n$; claim $j$ is
accepted by the runtime monitor iff $b_j$.

\textbf{Solvency, support, affordability.} These are the three facts
Section~\ref{sec:implementation} names, one word apiece, and this is the
machinery that states them precisely.
\emph{Solvency} is $\mathit{cost}(\alpha_j) \le W$: the pool can pay for this
claim's stake if the claim loses, in a currency (the superlinear deduction
$\mathit{cost}$) that \lstinline|fixed_fraction| does not use. Under $T_p$,
a claim deducts $\alpha_j = iW \le W$ directly, so solvency holds at every
reachable wealth by arithmetic alone and is not a premise of
\clc{\textsc{E-Claim}}; it becomes a real condition, one the five QUDE rules
must enforce with their own cap, only under a rule that deducts
$\mathit{cost}(\alpha_j)$ instead. \emph{Support} is $b$, the verdict
$[\,p \le \alpha_j\,]$: the null was rejected at the level the stake bought.
\emph{Affordability} is the whole-program static verdict \ty{\textsc{Ok}}
(Section~\ref{sec:implementation}): every declared bound clears the stake
available when its claim is reached, computed by $\mathit{affordable}$
folding $\pi$ from $w_0$ with no data in reach. Support is a fact about the
data and cannot be known statically; affordability is a fact about the
budget and the declared bounds alone.

Two consequences follow. First, $b$ is a rejection
verdict, distinct from any affordability condition; the vocabulary above
keeps the two apart. Second, the counts
reported in Section~\ref{sec:evaluation} (``6{,}545 of 8{,}911 claims are supported'' under the strongest-first oracle bound, ``0 of 8{,}911'' in reading order, ``0 of 8{,}911'' under a predictable one) are measured by running
\lstinline|Study.spend| and counting the calls that return, and
\lstinline|spend| raises \lstinline|Unaffordable| exactly when
$p > \alpha_j$, never on solvency, which \lstinline|fixed_fraction| never
fails. Those counts are therefore \emph{support at the level the
budget afforded}: the conjunction of a stake the pool could still assign and
a $p$-value that cleared it. They are not the static judgment
\ty{\textsc{Ok}}, which Section~\ref{sec:metatheory}'s ``static ok''
column reports, and the two answer different questions on the same corpus.
The exception's name is a misnomer we have kept only because it is a shipped
API name.

A \emph{static configuration} $\langle \hat W, c \rangle$ replaces the real
price with the declared bound. Where a bound is declared it steps by the same
transformer \clc{\textsc{T-Claim}} assigns, $T_{\bar p}$, evaluated at that
bound; where none is, it falls back on the no-refund one:

\begin{mathpar}
\inferrule[S-Note]
  { }
  {\clc{\langle \hat W,\; \kw{note}(s) \,;\, c \rangle \Longrightarrow \langle \hat W,\; c \rangle}}

\inferrule[S-Claim]
  {\ty{\bar p \ne \bot} \\ \ty{\hat b = [\, \bar p \le i \cdot \hat W \,]}}
  {\clc{\langle \hat W,\; \kw{claim}^m_{\bar p}(\ell, v_1 \mathbin{\kw{vs}} v_2) \,;\, c \rangle
   \overset{\hat b}{\Longrightarrow} \langle T_{\bar p}(\hat W),\; c \rangle}}

\inferrule[S-Unbound]
  {\ty{\bar p = \bot}}
  {\clc{\langle \hat W,\; \kw{claim}^m_{\bar p}(\ell, v_1 \mathbin{\kw{vs}} v_2) \,;\, c \rangle
   \Longrightarrow \langle T_\bot(\hat W),\; c \rangle}}
\end{mathpar}

\noindent $T_{\bar p}$ is Equation~\ref{eq:transformer} and $T_\bot$ is
Equation~\ref{eq:static-transformer}. \clc{\textsc{S-Claim}}
fires only when a bound is declared, and labels the step with $\hat b$, true
exactly when the checker certifies that claim statically. It steps by
$T_{\bar p}$ rather than $T_\bot$: the declared bound is the only $p$-value
the checker has, and pricing the step at it makes the shadow trace
mechanize the exact grade \clc{\textsc{T-Claim}} assigns and the wealth
\kw{Plan.audit} (Section~\ref{sec:implementation}) threads, with no extra
pessimism. This stays sound only because $T$ is antitone in $p$ (Lemma~\ref{lem:anti}): an honest run's real price is at most
$\bar p$, so the real transformer refunds at least as often as the shadow one
credits. \clc{\textsc{S-Unbound}}
still advances the shadow wealth ($T_\bot$ needs no $p$-value to evaluate)
but produces no verdict, which is the formal reading of Section~\ref{sec:declarations}'s
statement that an analysis omitting a bound still runs under a monitor: the
checker tracks wealth through an undeclared claim without certifying it. The
checker \emph{certifies} claim $j$ iff $\hat b_j$ holds at its step, and
\emph{accepts the program} iff it certifies every claim
\clc{\textsc{S-Claim}} fired on; Theorem~\ref{thm:static} is stated per
certified claim, so whole-program acceptance appears in it only as the
special case in which no step is refused.

\textbf{The identification of $T_{1.0}$ with $T_\bot$.}
\clc{\textsc{S-Unbound}} steps by $T_\bot$. The implementation
(\lstinline|Plan.audit|) instead steps an unbounded claim by $T_{1.0}$: it
declares the vacuous bound $\bar p = 1$ and runs the same code path as a
declared one, on the grounds that no stake ever reaches $1$, so
$[\,1 \le i \hat W\,]$ is false and $T_{1.0} = T_\bot$ pointwise. This is an
identification of the two transformers' \emph{arithmetic} alone, not of
their bookkeeping: a truly undeclared claim and a claim the analyst
declares at the literal bound $\bar p = 1$ step wealth identically, but
only the latter is \emph{decided} in the audit's sense and counted in its
affordability tally, correctly, as \textsc{REFUSED}, since no positive
stake can ever clear $1$. Section~\ref{sec:implementation}'s leaderboard
transcript exercises exactly this case: its six bounds are declared,
literal cluster-level $p$-values, one of which happens to be $1.0000$, and
that row is reported \textsc{REFUSED} and counted among the six decided
claims, unlike the genuinely unbound
\lstinline[style=tacetcode]|"A beats B"| claim of
Section~\ref{sec:certified}, reported \emph{unknown} and
excluded from its plan's tally. The two
agree exactly when $i \hat W < 1$ at every step, i.e.\ under the invariant
\[
  \ty{\hat W_j \;<\; 1/i \qquad \text{for every } j.}
\]
We state this as an assumption rather than a fact, because nothing in the
rules enforces it. The static trace does not step by $T_\bot$ alone: at a
declared claim, \clc{\textsc{S-Claim}} steps by $T_{\bar p}$, which can add
$\omega$ on certification, so $\hat W$ is not simply non-increasing: the
shipped audit trace of Section~\ref{sec:implementation} shows it rising,
$0.025 \to 0.0375 \to 0.04375$. The recursion's fixed point bounds it instead: $\hat W_{j+1} \le (1-i)\hat W_j + \omega$ at every step (with
equality possible only when \clc{\textsc{S-Claim}} certifies; $\hat W$ is
unchanged or smaller under $T_\bot$), and $(1-i)x + \omega = x$ at
$x = \omega/i$, so $\hat W_j \le \max(w_0, \omega/i)$ for every $j$. The
invariant above therefore holds whenever $\max(w_0, \omega/i) < 1/i$, i.e.\
$w_0 < 1/i$ \emph{and} $\omega < 1$, both of which hold for the shipped
defaults ($w_0 = \alpha/2$, $i = 0.5$, $\omega = \alpha/2 \le 1/2$). The runtime side does not guarantee it even under this bound: $T_p$
adds $\omega$ on every rejection and $\omega$ is
constrained only by $\omega \ge 0$, so a long enough run of rejections can
in principle carry $W$ past $1/i$, at which point a claim would be certified
by arithmetic rather than by evidence. The same unbounded-$\omega$ gap is
the one Section~\ref{sec:threats} records against Theorem~\ref{thm:mfdr}. We have not
proved a bound on $\omega$ that closes it, and we do not claim one; the
identification of $T_{1.0}$ with $T_\bot$ in \lstinline|Plan.audit| is
sound under the stated invariant and is an open gap without it.

One more discrepancy the identification leaves open: read
\clc{\textsc{Ok}} literally, declaring $\bar p = 1$ makes
\clc{\textsc{S-Claim}} fire with $\hat b_j = [\,1 \le i \hat W_j\,]$, false
at every reachable $\hat W_j$ under the invariant above, so a program that
\emph{encodes} an unbounded claim as the literal bound $\bar p = 1$, the
identification \lstinline|Plan.audit| makes, is refused by the formal
system outright, not given the ``no verdict'' \clc{\textsc{S-Unbound}}
assigns a genuinely undeclared claim ($\bar p = \bot$) directly.
\lstinline|Plan.audit|
does not do this: it excludes a vacuous-bound claim from the acceptance
conjunction entirely rather than folding its always-false verdict in, which
is the shipped behavior Section~\ref{sec:implementation} describes and this
appendix's grammar, which permits $\bar p = 1$ as an ordinary declared
bound, does not itself express.

\begin{defi}[Honest bounds]
\label{def:honest}
A run is \emph{honest} for a program $c_1 \,;\, \dots \,;\, c_n$ if, writing
$p_j$ for the price \clc{\textsc{E-Claim}} computes at step $j$ and $\bar p_j$ for
that claim's declared bound, $p_j \le \bar p_j$ whenever $\bar p_j \ne \bot$.
\end{defi}

\textbf{Progress and preservation.} Progress and preservation, for the
\emph{static} configuration above, are immediate, and are not proved
separately. \emph{Progress}: a claim sequence has two leaf forms, each with
exactly one applicable rule (\clc{\textsc{S-Claim}} and
\clc{\textsc{S-Unbound}} between them cover every value of $\bar p$), so no
non-empty static configuration is stuck; the rules above are written
$\langle \hat W, \mathit{leaf} \,;\, c \rangle \Longrightarrow \langle \hat
W', c \rangle$ for a trailing $c$ that may itself be empty, the sequence's
last leaf stepping to the empty configuration that ends the run rather than
to a stuck one. This is a fact about $\langle \hat W, c \rangle$
specifically; the runtime configuration under the five QUDE rules can stick
on an insolvent step by design (Section~\ref{sec:grade}), and nothing above
claims otherwise. \emph{Preservation}: no rule recomputes a $\Claim{w}$ type
mid-execution, since \clc{\textsc{T-Seq}} assigns the whole sequence its
type before any step is taken. Theorems~\ref{thm:static}
and~\ref{thm:mfdr} certify a domain-specific soundness about the grade
rather than freedom from getting stuck, with Theorem~\ref{thm:bridge}
chaining the two, as \sys{Duet} and \sys{Fuzz}'s cost-soundness theorems
play the same role for the same reason.

\textbf{Decidability of checking, relative to an oracle.} The
metatheoretic question this system actually raises is whether
$\vdash e : \Stat{\Sigma}{\varphi}{\kappa}{t}$ is decidable. The obstruction
is \est{\textsc{T-Observe}}, whose conclusion carries
$\est{\mathit{units}(r,a)}$: the footprint of an \kw{observe} is a function
of the run $r$ and the artifact key $a$, so an estimation term's type
mentions data a static checker does not have in general, since $a$ is an
arbitrary host-language expression and deciding what it assigns without
running it is, for a Turing-complete host, the halting problem in
disguise: no general oracle for $\mathit{units}$ can exist, and this is a
fact about embedding in Python rather than a gap this type system's rules
leave open. Fixing that one obstruction as an assumption isolates
everything else as a positive result.

\emph{Claim.} Given an oracle deciding $\mathit{units}(r,a)$ at every
\kw{observe} leaf, and decidable equality on $\mathit{Field}$ and
$\mathit{Val}$, both $\vdash e : \Stat{\Sigma}{\varphi}{\kappa}{t}$ and
$\pi \vdash c : \Claim{w}$ are decidable.

For the estimation judgment, $\varphi$, $\kappa$, and $t$ are total,
syntax-directed functions of the term alone away from
\est{\textsc{T-Observe}} (\est{\textsc{T-SelectKey}}'s $\psi$ filters
$\varphi$ by a decidable key-match test, once $\mathit{Val}$'s equality
is decidable; \est{\textsc{T-SelectOutcome}} is selected by the
form $v\langle q \rangle$ alone and inspects nothing to decide it applies;
\est{\textsc{T-Arith}} joins both sides), so checking reduces to one oracle
call per \kw{observe} leaf plus a bounded number of decidable tests at
every other node. This is no longer merely asserted: the mechanized footprint, constraint,
and purity recursions (Appendix~\ref{app:estimation},
Section~\ref{sec:mechanization})
\emph{are} exactly these functions, total and defined by structural
recursion on the sample sort, one syntax case per rule; the estimation
half of the claim is discharged by their existence rather than by the
syntax-directedness argument alone. For the claim judgment,
$\mathit{design}(\Sigma,\kappa_1,\kappa_2)$ (\S\ref{sec:design-premise}) is
a least fixed point over the finite set $F$, decidable given decidable
equality on $\mathit{Field}$ for the closure step and on $\mathit{Val}$
for $\mathit{differing}$; $\mathit{assumes}(m)$ and $\mathit{sided}(m)$ are
lookups against $m$'s finite, fixed table (\S\ref{sec:design-premise}'s
three served mechanisms plus Fisher's, McNemar's, and the sign-flip); and
plan membership $\pi(\ell)$ is a lookup in a finite list. None of this is
mechanized (the mechanized checker consumes the claim-sequence normal form
with $\mathit{price}$ already an oracle, not $\mathcal{T}$'s design premise
itself), so the claim-level half remains an
argument rather than a proof, but it is a routine one: every ingredient is
a decision over a set fixed in advance, with no recursion of its own.

\lstinline|python3 -m tacet.static| is this relative decision procedure,
instantiated with a specific, conservative oracle: on straight-line code
built from literal arguments it can resolve $\mathit{units}$ symbolically
enough to answer, and the moment a plan entry is not resolvable to a
literal (a claim built inside a loop over data, or from anything the
walker cannot reduce to a literal AST pattern,
Section~\ref{sec:implementation}), it exits its own, third code rather
than guessing, exactly the honest response an oracle that cannot answer
locally deserves. The residual difficulty is therefore located precisely:
not in $\mathcal{T}$'s rules, all of which admit a decision procedure
given the oracle, but in how much of a Turing-complete host language a
particular oracle implementation can see through without running it.

\section{Proof of Theorem~\ref{thm:static}}
\label{app:proof}

Throughout this appendix $0 < i \le 1$ and $\omega \ge 0$, the $i,\omega$
hypotheses shared by Theorem~\ref{thm:static} and Theorem~\ref{thm:mfdr};
the latter's full hypothesis list also includes (H1)--(H4), of which the
Lean development carries (H3) and (H4) directly and leaves (H1) and (H2) as
the meta-level obligations Section~\ref{sec:mfdr} already states they are.
Three facts about $T$ drive the whole proof, all already stated in prose in
Section~\ref{sec:grade} or Section~\ref{sec:mfdr} and property-tested over
400 wealth values; the Lean development
proves all three lemmas below outright.

\begin{lem}[Monotone in $W$]
\label{lem:mono}
For fixed $p$, $W \le W'$ implies $T_p(W) \le T_p(W')$.
\end{lem}
\begin{proof}
$T_p(W) = (1-i)W + \omega \cdot [p \le iW]$ is non-decreasing in $W$
termwise: $(1-i)W$ is, since $i \le 1$ makes $1 - i \ge 0$; and
$[p \le iW]$ can only go from false to true as $W$ grows (using $i > 0$),
never the reverse, so the indicator term is too, since $\omega \ge 0$.
\end{proof}

\begin{lem}[Antitone in $p$]
\label{lem:anti}
For fixed $W$, $p \le p'$ implies $T_p(W) \ge T_{p'}(W)$.
\end{lem}
\begin{proof}
Only the indicator depends on $p$, and $[p \le iW]$ can only go from false to
true as $p$ shrinks.
\end{proof}

\begin{lem}[No refund is a lower bound]
\label{lem:bot}
For every $p$ and every $W$, $T_\bot(W) \le T_p(W)$.
\end{lem}
\begin{proof}
$T_p(W) - T_\bot(W) = \omega \cdot [p \le iW] \ge 0$, since $\omega \ge 0$.
\end{proof}

\begin{proof}[Proof of Theorem~\ref{thm:static}]
Fix a program $c_1 \,;\, \dots \,;\, c_n$ and an honest run of it, and let
$W_0, \dots, W_n$ and $\hat W_0, \dots, \hat W_n$ be the
runtime and static traces from Appendix~\ref{app:semantics}, both starting at
$w_0$. This proof indexes the claim that steps $W_j \to W_{j+1}$ as ``claim
$j$'' (so claim $0$ is the first); Section~\ref{sec:mfdr} instead calls that
same claim ``claim $j+1$'', priced at $\alpha_{j+1} = i \cdot W_j$, to match
its ordinal position $1, \dots, n$ in the sequence. The trace $W_0, \dots,
W_n$ is the same object either way; only the label attached to a given step
differs. We show by induction on $j$ that $\hat W_j \le W_j$, whatever mix
of certified, refused, and undeclared claims precedes step $j$, and that
whenever \clc{\textsc{S-Claim}} certifies claim $j$ ($\hat b_j$ holds), the
runtime monitor also accepts it, i.e.\ $b_j$ holds.

\textbf{Base case.} $\hat W_0 = W_0 = w_0$.

\textbf{Inductive step.} Assume $\hat W_j \le W_j$. If leaf $j$ is
\kw{note}, \clc{\textsc{E-Note}} and \clc{\textsc{S-Note}} both step by the
identity, so $\hat W_{j+1} = \hat W_j \le W_j = W_{j+1}$ and the invariant
is restored trivially, with no $b_j$ to show since a note carries none.
The remaining case follows, a \kw{claim} leaf with some declared bound $\bar p_j \in [0,1] \cup \{\bot\}$. On the runtime side the leaf
steps by \clc{\textsc{E-Claim}} either way, so $W_{j+1} = T_{p_j}(W_j)$;
the two cases below differ in what the checker charges.

If \clc{\textsc{S-Unbound}} fired ($\bar p_j = \bot$), then
$\hat W_{j+1} = T_\bot(\hat W_j)$, and there is no verdict to show.
Lemma~\ref{lem:mono} and the induction
hypothesis give $T_\bot(\hat W_j) \le T_\bot(W_j)$, and Lemma~\ref{lem:bot}
gives $T_\bot(W_j) \le T_{p_j}(W_j)$. Honesty is not needed here, which is
why an undeclared claim costs the guarantee nothing.

If \clc{\textsc{S-Claim}} fired ($\bar p_j \ne \bot$), certified or not,
then $\hat W_{j+1} = T_{\bar p_j}(\hat W_j)$, and the chain runs through
the declared bound instead of through $T_\bot$:
\[
  \ty{T_{\bar p_j}(\hat W_j) \;\le\; T_{p_j}(\hat W_j) \;\le\; T_{p_j}(W_j).}
\]
The first inequality is Lemma~\ref{lem:anti} applied at the fixed wealth
$\hat W_j$, using honesty's $p_j \le \bar p_j$; the second is
Lemma~\ref{lem:mono} applied at the fixed price $p_j$, using the induction
hypothesis. The chain never consults $\hat b_j$, which carries the
invariant past a refused claim and, with it, the verdicts of every claim
behind the refusal; at a refused step the shadow earns no refund ($\hat
b_j$ false means $\bar p_j > i \cdot \hat W_j$), so the first inequality
there already follows from Lemma~\ref{lem:bot} alone, with honesty
unneeded. This is also the step that requires antitonicity: the
checker charges the claim at the bound the analyst declared, and the run can
only do better, because a smaller realized $p$-value can only turn the refund
indicator on. Were the checker to charge $T_\bot$ here, the shadow trace
would be sound for a checker neither Section~\ref{sec:lambdat} types nor
Section~\ref{sec:implementation} ships.

For the verdict, suppose $\hat b_j$ holds, so $\bar p_j \le i \cdot \hat
W_j$. By honesty, $p_j \le \bar p_j$. Combined with the induction
hypothesis and $i > 0$:
\[
  \ty{p_j \;\le\; \bar p_j \;\le\; i \cdot \hat W_j \;\le\; i \cdot W_j,}
\]
so $p_j \le i \cdot W_j$, which is exactly $b_j$'s defining condition in
\clc{\textsc{E-Claim}}. The runtime monitor accepts claim $j$.

In all cases $\hat W_{j+1} \le W_{j+1}$, restoring the invariant. By
induction, $\hat W_j \le W_j$ for all $0 \le j \le n$, and $b_j$ holds at
every step the checker certified, whether or not it certified the rest.
That is the theorem, with whole-program acceptance the case in which every
\clc{\textsc{S-Claim}} step certified.
\end{proof}

\begin{cor}[Partial pre-registration]
The wealth invariant $\hat W_j \le W_j$ holds for every program in which
every declared bound is honest, regardless of how many claims declare one:
the step for an undeclared claim (\clc{\textsc{S-Unbound}}) needs only
Lemmas~\ref{lem:mono} and~\ref{lem:bot}, not honesty, so mixing in
undeclared claims costs the invariant nothing. A declared claim still needs
its own bound to be honest, exactly as \clc{\textsc{S-Claim}}'s case in the
proof above uses. Only a claim that declares an honest, statically-affordable bound
inherits the runtime guarantee; every other claim runs, as it always did,
under the monitor alone.
\end{cor}

\section{Empirical Validation of Theorem~\ref{thm:static}}
\label{app:validation}

Section~\ref{sec:metatheory} summarizes this experiment; the corpus
construction and full counts follow. Random claim sequences are generated
and the checker is compared against the monitor, with bounds declared as
$\bar p = p \times \mathit{slack}$ and each bound priced by the same
mechanism the monitor will choose. Table~\ref{tab:validation}
reports two corpora. The first draws random pairwise comparisons over
the fixed data of a code-model comparison (8 open-weight models scored on
15 paired coding tasks in two prompting modes),
and is a narrow test by construction:
no pairwise comparison in this fixed corpus of 15 paired tasks realizes more
than 8 discordant pairs, so no $p$-value below $0.0078$ is ever available
to test, while the opening stake is $0.0125$ and halves on any refusal, so
only a program's first claim can ever be statically certified, and its 18
certified claims are all at position 1 over near-flat wealth trajectories.
The second corpus is tuned to exercise
what the first cannot: synthetic paired outcomes with 150 to 400 tasks per
program, so exact $p$-values span decades; planted effects in half the
models; a fifth of claims declaring no bound, interleaving
\clc{\textsc{S-Unbound}} with \clc{\textsc{S-Claim}}; account parameters
$\alpha$, $i$, $\omega$ drawn per program within Theorem~\ref{thm:static}'s
hypotheses and shared by both sides; and a third of programs ordered
strongest-first, legitimate here (unlike for the mFDR guarantee) because
the theorem is a per-path statement with no distributional hypothesis on
the claim sequence.

\begin{table}[t]
\centering
\small
\caption{Static checker against runtime monitor on two corpora of randomly
generated programs, at declared bounds $\bar p = p \times \mathit{slack}$:
299 programs over the fixed comparison data (1{,}911 claims, all declaring
bounds) and 300 tuned programs over synthetic paired outcomes with planted
effects (2{,}969 claims, 2{,}367 declaring bounds). The top four rows
declare honest bounds; the bottom two declare dishonest ones (overstating
the evidence) as a negative control.}
\label{tab:validation}
\begin{tabular}{NNNNNNN}
\toprule
 & \multicolumn{3}{C}{fixed corpus} & \multicolumn{3}{C}{tuned corpus} \\
\cmidrule(lr){2-4}\cmidrule(lr){5-7}
\multicolumn{1}{R}{slack} & \multicolumn{1}{R}{static ok} & \multicolumn{1}{R}{runtime ok} & \multicolumn{1}{R}{viol.} & \multicolumn{1}{R}{static ok} & \multicolumn{1}{R}{runtime ok} & \multicolumn{1}{R}{viol.} \\
\midrule
1.0 & 18 & 18 & \textbf{0} & 1120 & 1430 & \textbf{0} \\
1.5 & 18 & 18 & \textbf{0} & 1087 & 1430 & \textbf{0} \\
3.0 & 0 & 18 & \textbf{0} & 1029 & 1430 & \textbf{0} \\
10.0 & 0 & 18 & \textbf{0} & 912 & 1430 & \textbf{0} \\
\midrule
0.10 & 539 & 18 & 521 & 1449 & 1430 & 299 \\
0.01 & 1911 & 18 & 1893 & 2001 & 1430 & 847 \\
\bottomrule
\end{tabular}
\end{table}

Each corpus's program set is fixed once and re-checked at every
slack, so Table~\ref{tab:validation}'s rows differ only in the precision of
the declared bound.
Honest violations stay at zero on both corpora, as
Theorem~\ref{thm:static} requires, and the tuned corpus makes that zero
informative: where the fixed corpus certifies 18 claims, all at
position 1 (its slack $3.0$ and $10.0$ rows certify none and hold
vacuously), the tuned corpus certifies 1{,}120 of 2{,}367 declared claims
at honest bounds, spread across positions (148 at position 1, 972 later),
including 566 certified after an earlier refusal had already drained the
pool and 855 after an earlier refund had refilled it, with zero violations
in every stratum. The 566 are the stratum the per-claim statement of
Theorem~\ref{thm:static} exists for: they occur in programs the checker does
not accept as a whole, so a theorem hypothesized on whole-program
acceptance would leave exactly this stratum untested. The implication is thereby tested across the wealth
trajectories the theorem quantifies over, beyond first claims alone. The
bottom rows are a negative control with dishonest bounds, because a
guarantee that cannot be broken by lying is not being tested; it breaks on
both corpora.

\textbf{The cost of imprecision.} The two corpora disagree about the shape
of the cost of an imprecise bound. On the fixed corpus the cost is a
cliff: declaring within $1.5\times$ of the truth costs nothing, and
declaring threefold high costs everything, because the opening stake is
$0.0125$, a bound overstated threefold requires a $p$-value of $0.00417$,
and the smallest $p$-value any pairwise comparison in this fixed corpus
realizes (Qwen2.5-Coder vs.\ CodeQwen, 8 of 15 tasks discordant) is
$0.0078$. On the tuned corpus, whose $p$-values span decades, the same
overstatement is a slope: 1{,}120 claims certified at slack $1.0$, 1{,}087
at $1.5$, 1{,}029 at $3.0$, and 912 at $10$. The cliff shows what low
resolution does to an imprecise bound, and the slope shows what resolution
buys back.

\section{The Deletion Attack}
\label{app:crossing}

Section~\ref{sec:threats} states that the guarantee is conditional on a
complete account (H1); this appendix measures the cost of violating it.
Under a complete null, consider an analyst who runs every pairwise
comparison among $n$ models and simply does not pay for the ones that
did not reject; Table~\ref{tab:deletion} reports the measured cost.

\begin{table}[t]
\centering
\small
\caption{$\mathrm{mFDR}_\eta$ under a complete null, all pairwise comparisons
among $n$ models, for an honest analyst who reports every test run against
an analyst who deletes tests that failed to reject before reporting. Bold
exceeds $\alpha=0.05$. Monte Carlo, 4{,}000 trials per row, same seed
across rows.}
\label{tab:deletion}
\begin{tabular}{NNNN}
\toprule
\multicolumn{1}{R}{models} & \multicolumn{1}{R}{tests} & \multicolumn{1}{R}{honest $\mathrm{mFDR}_\eta$} & \multicolumn{1}{R}{deleting $\mathrm{mFDR}_\eta$} \\
\midrule
4  & 6   & 0.0091 & 0.0264 \\
6  & 15  & 0.0096 & \textbf{0.0645} \\
8  & 28  & 0.0102 & \textbf{0.1173} \\
12 & 66  & 0.0109 & \textbf{0.2498} \\
20 & 190 & 0.0120 & \textbf{0.5308} \\
\bottomrule
\end{tabular}
\end{table}

The honest analyst is flat in the number of tests, because wealth
decays geometrically and $\mathbb{E}[R]$ saturates. The deleting
analyst never pays for a failure, so wealth never falls below $w_0$,
the stake never falls below $i w_0$, every further test carries at
least that probability of a spurious rejection, $\mathbb{E}[R]$ grows
linearly, and $\mathrm{mFDR}_\eta \to 1$. Because $\eta = 1 - \alpha$,
the ratio crosses $\alpha$ once $\mathbb{E}[R] > \alpha$, so a
perfectly uniform null $p$-value puts the crossing at 5 tests, whatever
the data. Where the crossing actually falls depends on how near the
exact test comes to uniform, a question of how much data each
comparison has; Table~\ref{tab:crossing} reports where it lands on null
$p$-values drawn from the real corpora.

\begin{table}[t]
\centering
\small
\caption{Where the deleting analyst's $\mathrm{mFDR}_\eta$ crosses
$\alpha=0.05$ under a complete null, per average rejection probability at
the opening stake (``per test''), for four sources of null $p$-values;
five tests is the closed-form floor, reached only at perfect uniformity.
Monte Carlo, 4{,}000 trials per row.}
\label{tab:crossing}
\begin{tabular}{LNN}
\toprule
null $p$-values from & \multicolumn{1}{R}{per test} & \multicolumn{1}{R}{crossing} \\
\midrule
uniform, in closed form & 0.0125 & $\le 5$ \\
SWE-bench Verified, permuted & 0.0160 & 4 to 6 \\
BIG-Bench Hard, permuted & 0.0094 & 6 \\
20 tasks per comparison, synthetic (Table~\ref{tab:deletion}'s setup) & 0.0057 & 15 \\
\bottomrule
\end{tabular}
\end{table}

Table~\ref{tab:crossing}'s two corpus rows destroy the real effects by
permutation and
keep everything else: within an instance, which systems resolved it is
shuffled; within an example, the two prompting conditions are swapped.
Instance difficulty, the correlation it induces across systems, and how
discordance is distributed over a task are the corpus's own. BIG-Bench Hard
lands one test above the floor, at a per-test rejection rate close enough
to uniform that the closed-form crossing is nearly exact; the synthetic
sweep is conservative instead, since 20 tasks
leave McNemar's exact test too few discordant pairs to be anything but
conservative, understating the deletion's effect threefold. SWE-bench
Verified lands well clear of the floor: its crossing range straddles it,
because shuffling which system resolved a given instance leaves each
instance's own difficulty in place, so systems permuted this way stay
correlated through the instances they share. That is exactly the (H3)
dependence Section~\ref{sec:mfdr} already concedes is unverified for this
corpus, surfacing here as a per-test rejection rate above the uniform
floor rather than below it; on data drawn from correlated systems, the
uniform floor is not a bound on the deletion attack's severity, only a
convenient closed-form reference point.

\section{Selection, Verification, and Accounting Side by Side}
\label{app:threelayers}

Figure~\ref{fig:three-layers} lays out the three released surfaces
Section~\ref{sec:bhl} compares, on the smallest analysis all three can
express: StatWhy's own two-drug example, since neither the released
selection layer nor the released verification library yet covers the
paired-nominal instrument the leaderboard needs. No software connects the
three panes today, so the figure is an architecture and a statement of
where the seams fall.

\begin{figure}[ht]
\sys{Tisane}\textit{: declares the design and infers the instrument.}
\begin{lstlisting}[basicstyle=\ttfamily\footnotesize]
import tisane as ts

patient = ts.Unit("patient")
drug = patient.nominal("drug", cardinality=3)
improve = patient.numeric("improvement")
drug.causes(improve)

ts.infer_statistical_model_from_design(
    design=ts.Design(dv=improve, ivs=[drug]))
# -> which instrument the declared design admits
\end{lstlisting}

\textit{StatWhy: runs both tests and proves the composed bound.}
\begin{lstlisting}[basicstyle=\ttfamily\footnotesize]
let h1 = Disj (h_new_drug1, h_new_drug2)

let cmp_with_existing_drugs d_new d_drug1 d_drug2 =
  let p1 = exec_ttest_ind_eq ppl_new ppl_drug1 d_new d_drug1 Up in
  let p2 = exec_ttest_ind_eq ppl_new ppl_drug2 d_new d_drug2 Up in
  p1 +. p2
(*@ requires sampled d_new ppl_new /\ independent d_new d_drug1 /\ ...
    ensures (World !st interp) |= StatB (Leq p) h1 *)
\end{lstlisting}

\sys{Tacet}\textit{: prices both claims against one pool.}
\begin{lstlisting}[style=tacetcode, basicstyle=\ttfamily\footnotesize]
from tacet import Study

study = Study(alpha=0.05)
plan = study.preregister()
plan.claim("new beats drug1", sided=1, left="new", right="drug1")
plan.claim("new beats drug2", sided=1, left="new", right="drug2")

trial = study.observe(rows, artifact=patient, passed=improved)
study.claim("new beats drug1", trial["new"] > trial["drug1"])
study.claim("new beats drug2", trial["new"] > trial["drug2"])
\end{lstlisting}
\caption{Selection, verification, accounting: one analysis, StatWhy's own
two-drug comparison, in the three released surfaces. The StatWhy panel is
abridged from the released example, with line breaks added; its elided
requirements include normality and that each direction was already excluded
or possible in the belief state, the one-sidedness \sys{Tacet} instead
declares in a plan that \kw{observe} seals. Each panel is real, and none is
wired to the others; the family here has two members, small enough that
every panel handles it without yet showing where they diverge
(Section~\ref{sec:bhl}).}
\label{fig:three-layers}
\end{figure}

\section{Checker Runtime}
\label{app:runtime}

StatWhy's own CAV tool paper reports execution time
flat around 8--9s per hypothesis pair and per compared group, dominated by
SMT solver startup rather than by what is being checked.
Table~\ref{tab:checker-runtime} reports the analogous number for
\lstinline[style=tacetcode]|tacet.static|, which calls no solver at all: an
AST walk and a fold over declared bounds against one syntactic pattern
(Apple M4, 32\,GB,
Python 3.14, warm process, median of 7 in-process repeats after 2 discarded
warm-up runs). This is not the same measurement as StatWhy's table,
and the two should not be divided into each other: StatWhy's seconds price
an SMT solver discharging a verification condition and returning a proof;
these milliseconds price a syntactic fold that returns a design and a
price, with no theorem behind either, so a raw ratio between them measures
the gap between proving and folding, not a speedup available to anyone who
needs StatWhy's guarantee. We report both because the comparison motivates
the cost trade in Section~\ref{sec:bhl}.

\begin{table}[t]
\centering
\small
\caption{Checker runtime, one representative invocation: medians of 7
in-process repeats, with [min--max] over those repeats for the total.
\emph{Parse} is \lstinline[style=tacetcode]|ast.parse| alone; \emph{fold}
consumes its output; \emph{total} is their sum. The first block of
synthetic rows is literal, hand-unrolled
\lstinline[style=tacetcode]|plan.claim(...)| statements, written flat
rather than in a loop; the second block is the same claim counts written
as a single \lstinline|for| loop over one literal list instead, timed
through the unroller (Section~\ref{sec:implementation}) rather than
around it. The last row is the repository-wide design gate (Section~\ref{sec:metatheory}).}
\label{tab:checker-runtime}
\begin{tabular}{LNNNN}
\toprule
measurement & \multicolumn{1}{R}{claims} & \multicolumn{1}{R}{parse} & \multicolumn{1}{R}{fold} & \multicolumn{1}{R}{total [min--max]} \\
\midrule
\lstinline[style=tacetcode]|prereg.py|, 6 plans & 8 & 0.6~ms & 0.6~ms & 1.3~ms [1.2--1.5~ms] \\
synthetic plan & 10 & 0.05~ms & 0.08~ms & 0.13~ms [0.13--0.13~ms] \\
synthetic plan & 100 & 0.46~ms & 0.62~ms & 1.1~ms [1.1--1.2~ms] \\
synthetic plan & 1{,}000 & 4.7~ms & 6.9~ms & 12~ms [11--12~ms] \\
synthetic plan & 10{,}000 & 64~ms & 113~ms & 176~ms [164--186~ms] \\
\midrule
literal loop, unrolled & 10 & \multicolumn{2}{M}{--} & 0.5~ms [0.5--0.6~ms] \\
literal loop, unrolled & 100 & \multicolumn{2}{M}{--} & 3.4~ms [3.3--3.7~ms] \\
literal loop, unrolled & 1{,}000 & \multicolumn{2}{M}{--} & 32~ms [32--33~ms] \\
literal loop, unrolled & 5{,}000 & \multicolumn{2}{M}{--} & 180~ms [170--233~ms] \\
\midrule
design gate, whole repository & -- & -- & -- & 641~ms [638--677~ms] \\
\bottomrule
\end{tabular}
\end{table}

Two cautions on the synthetic sweep, both pointing the same way: read the
table as orientation, not a benchmark. Between-invocation variance on a
shared machine exceeds the in-process ranges it brackets (three back-to-back
runs of the 10{,}000-claim point gave medians of 169, 173, and 176~ms), and
the \emph{fold} column grows faster than linearly (0.62, 6.9, and 113~ms per
decade of claims, parsing contributing a third to a half of each total);
naming the fold's super-linear term is still open (an earlier $O(n^2)$
duplicate-label scan is now an $O(1)$ set check, which shrank but did not
remove it). Even so, 10{,}000 claims, comparable to the leaderboard's own
8{,}911-claim family and nearly four hundred times BIG-Bench Hard's 26,
complete in a fifth of a second on every invocation measured. The
literal-loop rows are a separate sweep on the same machine, orientation of
the same kind: unrolling 5{,}000 claims from a single \lstinline|for| loop
over a literal list, deep-copying and substituting the loop body once per
item (Section~\ref{sec:implementation}), costs about 180~ms, growing
somewhat faster than linearly, which that section's stated cap on
loop-unrolling size exists to bound.

Only the checker's own cost is timed above; the runtime monitor's cost is
a separate question this appendix did not previously answer, and one
mechanism's own cost dominates it at scale. The cluster sign-flip's exact
$p$-value is now computed by dynamic-programming convolution rather than
by enumerating every sign assignment directly (Section~\ref{sec:lambdat}):
at BIG-Bench Hard's pooled claim ($k=26$ clusters), enumeration measured
close to a minute in pure Python against under two milliseconds convolved,
identical $p$-value either way (the two routes are checked against each
other on 20 random configurations up to
$k=12$, where enumeration is still fast enough to check against). At the
leaderboard's own $k=12$, $2^{12} = 4{,}096$ sign assignments already
finish in well under a millisecond, so the two routes are indistinguishable
in wall-clock terms there; the difference is specific to a claim family
whose cluster count runs larger, and BIG-Bench Hard's pooled claim is
exactly that case.

\end{document}